\documentclass[a4paper,UKenglish,cleveref,autoref,thm-restate]{lipics-v2021}

\pdfoutput=1
\nolinenumbers
\hideLIPIcs

\usepackage[most]{tcolorbox}

\usepackage[T1]{fontenc}
\usepackage{float}
\newenvironment{problem}[1]{%
  \refstepcounter{theorem}%
  \begin{tcolorbox}[
    enhanced,
    colback=white,
    colframe=black,
    boxrule=0.5pt,
    arc=4pt,
    left=6pt, right=6pt,
    top=6pt, bottom=6pt,
    before upper=\textbf{Problem \thetheorem. #1}\par\medskip
  ]
}{%
  \end{tcolorbox}
}

\DeclareMathOperator{\pw}{pw}
\DeclareMathOperator{\Ori}{Ori}

\usepackage{listings}
\usepackage{xcolor}

\title{Edge Geography is XNLP-hard for Pathwidth and in XP for Tree-Partition Width}
\titlerunning{Edge Geography Parameterized by Width}

\author{Thobias Kvalvik H{\o}ivik}
{Department of Computer Science, Electrical Engineering and Mathematical Sciences, Western Norway University of Applied Sciences, Førde, Norway}
{thobiashoivik@gmail.com}
{https://orcid.org/0009-0008-7639-5666} 
{}

\author{Erlend Raa V{\aa}gset}
{Department of Computer Science, Electrical Engineering and Mathematical Sciences, Western Norway University of Applied Sciences, Førde, Norway}
{Erlend.Raa.Vagset@hvl.no}
{https://orcid.org/0000-0003-2289-2268}
{}

\authorrunning{T. K. H{\o}ivik and E. R. V{\aa}gset}
\Copyright{Thobias Kvalvik H{\o}ivik and Erlend Raa V{\aa}gset}

\ccsdesc{Theory of computation~Parameterized complexity and exact algorithms}
\keywords{Geography games, graph games, 
pathwidth, parameterized complexity, XNLP-hardness, 
PSPACE-complete games, tree-partition width}

\acknowledgements{We thank the anonymous reviewers for their careful reading and for their thoughtful and constructive feedback, which helped us improve the exposition and clarify several aspects of the paper.}

\begin{document}
\maketitle

\begin{abstract}
\textsc{Directed Edge Geography} and \textsc{Undirected Edge Geography} are classical
\textsc{PSPACE}-complete two-player graph games in which players alternately
make moves along edges, deleting each one after use; the first player unable to
move loses. We prove that both problems are \textsc{XNLP}-hard when
parameterized by pathwidth, addressing a question raised by Bodlaender over 30
years ago. On the positive side, we observe that \textsc{Directed Edge
Geography} is fixed-parameter tractable when parameterized by treewidth and
maximum degree. We also prove that both problems are in \textsc{XP} on simple graphs
when parameterized by tree-partition width. These results develop modern
lower-bound and decomposition-based algorithmic methods for width-based
questions in \textsc{PSPACE}-complete graph games.
\end{abstract}

\paragraph*{Conference version.}
A shortened version of this paper has been accepted for presentation at ESA 2026.

\section{Introduction}
Many natural computational problems lie beyond NP, yet their parameterized
complexity is still relatively poorly understood~\cite{deHaanSzeider2017BeyondParaNP,BackstromJonsson2011EqualPlanning}.
Among them, \textsc{PSPACE}-hard problems are especially challenging, since
their complexity is often driven by long sequences of interacting choices and
evolving states. 
Such phenomena arise in puzzles such as \textsc{Sokoban} and
\textsc{Rush Hour}~\cite{culberson1998sokoban,flake2002rushhour}, in
generalized games such as \textsc{Chess} and \textsc{Go}~\cite{fraenkel1981chess,lichtenstein1980go},
in motion planning~\cite{hopcroft1984warehouseman,solovey2015unlabeled}, in
task and classical planning~\cite{vegabrown2020tamp,bylander1994strips}, and in
temporal reasoning~\cite{sistla1985ltl}. In this paper, we study
\textsc{Directed Edge Geography} and \textsc{Undirected Edge Geography},
classical \textsc{PSPACE}-complete graph
games~\cite{fraenkel1993geography,FraenkelScheinermanUllman1993} that offer a
natural setting in which to ask what can and cannot be captured by structural
width parameters.

We are not the first to study \textsc{PSPACE}-hard problems from a
parameterized perspective. Bäckström and Jonsson~\cite{BackstromJonsson2011EqualPlanning}
develop finer complexity analyses for planning; Mouawad et
al.~\cite{MouawadNishimuraRamanSimjourSuzuki2017Reconfiguration}
study natural process parameters for reconfiguration; and Bonnet et
al.~\cite{BonnetGaspersLambilliotteRuemmeleSaffidine2017PositionalGames}
analyze short winning-strategy problems in games. 
A particularly relevant line of study for us is quantified reasoning.
Chen~\cite{Chen2004QCSPTreewidth} gives an early tractability result for the
\emph{quantified constraint satisfaction problem} (QCSP) under bounded
treewidth, and Pan and Vardi~\cite{PanVardi2006PSPACE} show that for
\emph{quantified Boolean formulas} (QBF), width remains useful when paired with
additional control over alternation. Atserias and
Oliva~\cite{AtseriasOliva2014BoundedWidthQBF}, however, show that ordinary graph
width alone can still be too weak, proving \textsc{PSPACE}-completeness even at
constant pathwidth. This in turn motivates the prefix-aware viewpoint of Eiben,
Ganian, and Ordyniak~\cite{EibenGanianOrdyniak2020MindThePrefix}, who introduce
decomposition parameters that reflect the dependency structure induced by the
quantifier prefix.

In contrast to quantified reasoning, where graph structure and dependency
structure need not align, \textsc{Edge Geography} is played directly on a
graph. The issue is therefore not representational mismatch, but whether graph
width alone can control an evolving game state. Bodlaender showed that
stronger restrictions can make the problem easier. In particular,
\textsc{Edge Generalized Geography} is solvable in linear time on graphs of
bounded treewidth and bounded maximum degree, and on directed graphs whose
underlying undirected graph is a cactus~\cite{bodlaender1993}. However, these
results do not explain what width bounds alone imply.

We address this question from both a hardness and an algorithmic perspective.
For hardness, we study pathwidth and related linear-width parameters, where
\textsc{XNLP} is the relevant complexity
framework~\cite{elberfeld2012space,elberfeld2015spacecircuit,bodlaender2024xnlp,bodlaender2025linearxnlp}.
For algorithms, we consider rooted tree partitions as a stronger decomposition
model~\cite{wood2009treepartition}. Our main results are as follows:

\begin{enumerate}
    \item As our main result, we prove that \textsc{Directed Edge Geography} is
    \textsc{XNLP}-hard when parameterized by pathwidth. By a
    parameter-preserving reduction to the undirected setting, the same also
    holds for \textsc{Undirected Edge Geography}. Using this transfer together with
    Bodlaender's bounded-treewidth-and-degree algorithm for
    \textsc{Edge Generalized Geography}~\cite{bodlaender1993}, we also obtain
    that \textsc{Directed Edge Geography} is fixed-parameter tractable when
    parameterized by treewidth together with maximum degree.

    \item On the positive side, we prove that \textsc{Directed Edge Geography}
    and \textsc{Undirected Edge Geography} on simple graphs, given together
    with a rooted tree partition of bounded width, are solvable in
    \textsc{XP}. Combined with the known XP algorithm for computing such
    decompositions, this yields XP algorithms parameterized by
    tree-partition width. We also implemented a prototype of the algorithm for
    correctness validation against a minimax solver; details are summarized in
    Appendix~\ref{sec:code}.
\end{enumerate}

\begin{remark*}
    Due to space constraints, Sections~3 and~5 are presented here in abbreviated form, with their full technical developments deferred to Appendices~A and~B, respectively. The shorter transfer argument of Section~4 is included in full.
\end{remark*}

\section{Preliminaries}

We briefly introduce key definitions, problems and notation. Throughout, we assume all graphs to be finite and simple and we may use the shorthand $[k] := \{1,\ldots,k\}.$

\subsection{Edge Geography}

\begin{definition}[Edge Geography]
An instance of \textsc{Edge Geography} consists of a graph $G=(V,E)$ together
with a designated start vertex $s\in V$, where $G$ is either directed or
undirected. Two players alternately move a token, initially placed at $s$. If the token is
currently at $v$, a move consists of selecting an unused edge from $v$ to some
vertex $u$, moving the token to $u$, and deleting that edge. In the directed
variant, the chosen edge must be an outgoing edge $(v,u)\in E$; in the
undirected variant, it must be an incident edge $\{v,u\}\in E$. The first
player unable to move loses.
\end{definition}

\begin{figure}[ht]
    \centering
    \includegraphics[width=\textwidth]{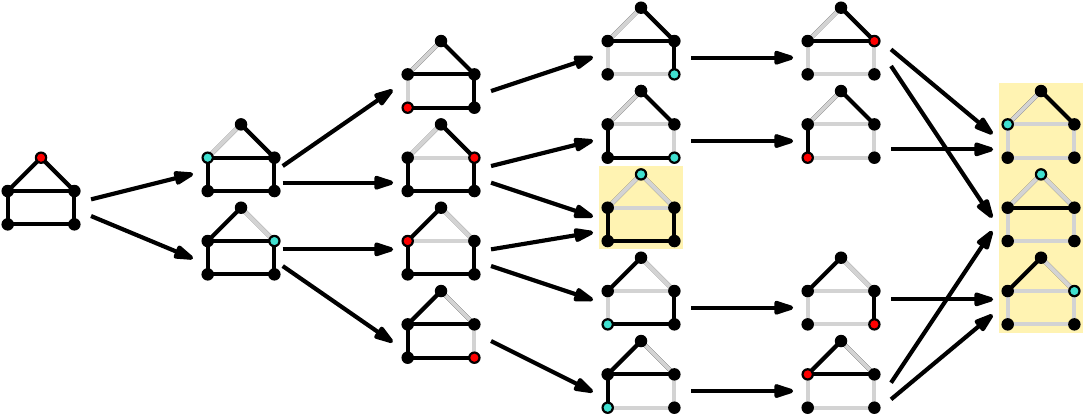}
    \caption{All possible states for a given geography game. This game DAG shows that it is the turquoise player that always becomes ``stuck'' and hence will always lose the game for this particular graph and initial position.}
    \label{fig:geographygame}
\end{figure}

\newpage
\begin{problem}{Directed/Undirected Edge Geography} 
\textbf{Input:} A directed or undirected graph $G$ and a start vertex $s\in V(G)$.\\
\textbf{Question:} Does Player~1 have a winning strategy?
\end{problem}

\noindent
\textbf{Move/edge notation. } 

\noindent 
If the token is at some vertex $v \in V(G)$, then a move consists 
of choosing an outgoing edge with some endpoint, e.g. $w \in V(G)$. 
When we write $v \to w$, we mean the move which corresponds to 
moving the token from vertex $v$ to $w$, deleting the edge 
$(v,w)$ ($\{v,w\}$ when referring to the undirected variant), resulting 
in the opposing player being the one to make a move at $w$. We also use 
this notation to refer to the edge $(v,w)$ itself. 
Often there will only be one legal move from a given vertex. 
Suppose $v,w,u \in V(G)$ such that there is only one outgoing edge 
from $w$, going to $u$, and that there is an edge from $v$ to $w$. 
Then we will often write $v \to w \to u$ to denote the forced 
sequence that arises from the player to move at $v$ choosing 
to play $v \to w$, since the opposing player is forced to move the token 
from $w$ to $u$. Similarly, when we have edges connected in this 
way where the intermediate vertex has outdegree exactly $1$ we will also 
often use this notation.

\subsection{Graphs and pathwidth}

\begin{definition}
A \emph{path decomposition} of a graph $G=(V,E)$ is a sequence of vertex sets
$
(X_1,\dots,X_r)
$
such that:
\begin{enumerate}
    \item for every vertex $v\in V$, there exists $i\in[r]$ with $v\in X_i$;
    \item for every edge $\{u,v\}\in E$, there exists $i\in[r]$ with
    $\{u,v\}\subseteq X_i$;
    \item for every vertex $v\in V$, the set
    $
    \{\, i \mid v\in X_i \,\}
    $
    forms a contiguous interval of $[r]$.
\end{enumerate}
\end{definition}

We will also use the notion of a \emph{nice path decomposition}:
a path decomposition $(B_1,\dots,B_r)$ is nice if $B_1=B_r=\emptyset$ and,
for every $i\in [r-1]$, the bags $B_i$ and $B_{i+1}$ differ in exactly one
vertex. It is well-known that every path decomposition of width $k$ can be
transformed in polynomial time into a nice path decomposition of the same
width.

\begin{definition}
The \emph{width} of a path decomposition $(X_1,\dots,X_r)$ is
$
\max_i |X_i|-1.
$
The \emph{pathwidth} of a graph $G$, denoted $\pw(G)$, is the minimum width over all
path decompositions of $G$. The pathwidth of a directed graph $G=(V,E)$, is the pathwidth
of its underlying undirected graph.
\end{definition}

\subsection{Parameterized complexity}
A parameterized problem is a language $L \subseteq \Sigma^* \times \mathbb{N}$.
An instance is a pair $(x,k)$, where $k$ is the parameter. 
A parameterized problem is in \emph{FPT} if it can be decided in time
$
f(k)\cdot |x|^{O(1)}
$
for some computable function $f$.
A parameterized problem is in \emph{XP} if it can be decided in time
$
|x|^{f(k)}
$
for some computable function $f$.

We will work with the class \textsc{XNLP}, consisting of parameterized problems
solvable by a nondeterministic Turing machine in time
$
f(k)\cdot |x|^{O(1)}
$
and space
$
O(k\log |x|).
$
This class was introduced under the notation $\mathrm{N}[f\mathrm{poly},f\log]$
by Elberfeld, Stockhusen, and Tantau~\cite{elberfeld2012space}, and has
subsequently been studied under the name \textsc{XNLP}, e.g. by
Bodlaender et al.~\cite{bodlaender2024xnlp}.
We will also use the following two results of Bodlaender~\cite{bodlaender1993}.

\begin{theorem}
\label{thm:bodlaender-tw-degree}
For every fixed $k,d\ge 1$, \textsc{Edge Generalized Geography} can be solved
in time $O(n)$ on graphs of treewidth at most $k$ and maximum degree at most
$d$, provided together with a tree decomposition of width at most $k$.
\end{theorem}

\begin{theorem}
\label{thm:bodlaender-cactus}
\textsc{Edge Generalized Geography} can be solved in time $O(n)$ on directed
graphs whose underlying undirected graph is a cactus.
\end{theorem}

In particular, since pathwidth bounds treewidth from above, Bodlaender's first
theorem immediately yields linear-time solvability of \textsc{Undirected Edge
Geography} on graphs of bounded pathwidth and bounded maximum degree, provided
with a path decomposition of bounded width.

\section{XNLP-hardness of Directed Edge Geography}

We prove \textsc{XNLP}-hardness of \textsc{Directed Edge Geography} by a
parameterized logspace reduction from \textsc{Chosen Maximum Outdegree}, which is
known to be \textsc{XNLP}-complete when parameterized by
pathwidth~\cite{bodlaender2022hardtw}. 
Following~\cite[Section~2.3]{bodlaender2022hardtw}, we assume that all integers in the source instance are encoded in unary.

\noindent
\textbf{Source problem.}

\noindent 
Let \(G\) be an undirected graph. An \emph{orientation} of \(G\) assigns a
direction to each edge. For an orientation \(\omega\) and a vertex \(x\), let
\(\theta_\omega(x)\) denote the set of edges directed out of \(x\), and let
\(\Ori(G)\) denote the set of all orientations of \(G\).

\begin{problem}{Chosen Maximum Outdegree}
\textbf{Input:} An undirected graph \(G=(V,E)\), a weight function
\(w:E\to \mathbb Z_{>0}\), and a vertex bound
\(t:V\to \mathbb Z_{>0}\).\\
\textbf{Question:} Is there an orientation \(\omega\in \Ori(G)\) such that
\[
\sum_{e\in \theta_\omega(v)} w(e)\le t(v)
\]
for every \(v\in V(G)\)?
\end{problem}
We assume that the input is given together with a path decomposition of width at
most \(k\). 

\begin{figure}[H]
    \centering
    \includegraphics[width=\textwidth]{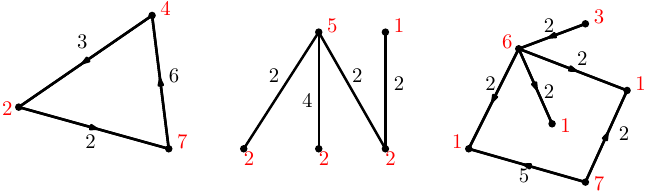}
    \caption{Three instances of \textsc{Chosen Maximum Outdegree}. Vertex bounds
    \(t(x)\) are shown in red. A feasible orientation is shown whenever one exists.
    The middle instance admits no feasible orientation.}
    \label{short:fig:orientations}
\end{figure}

\begin{figure}[H]
    \centering
    \includegraphics[width=\textwidth]{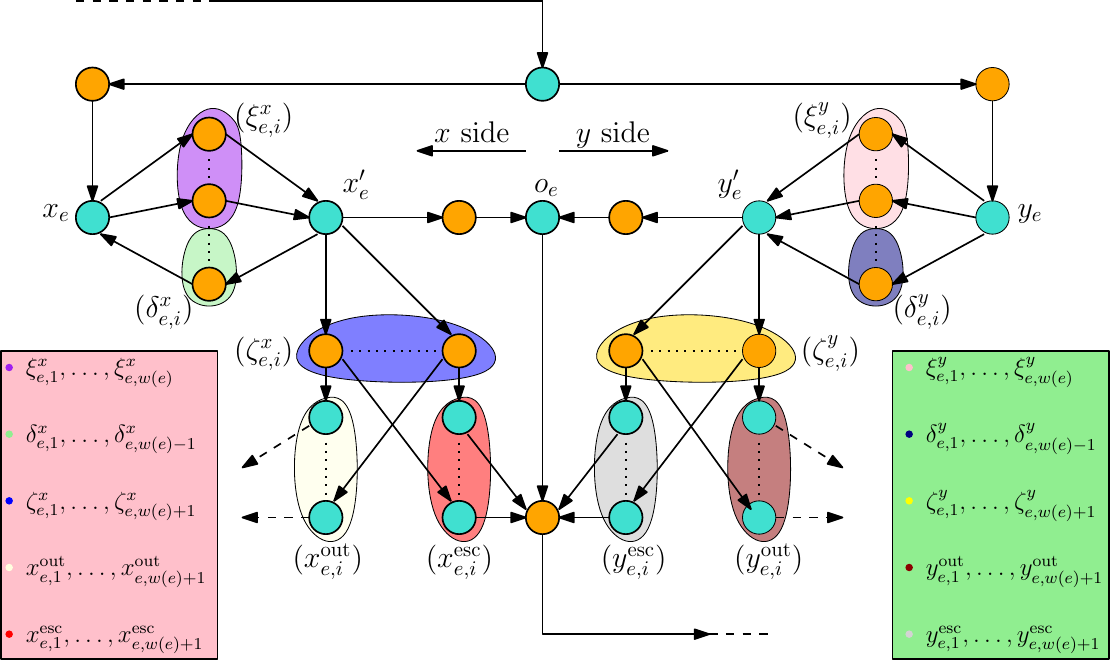}
    \caption{Edge gadget used in the reduction. 
    The \(x\)- and \(y\)-sides are symmetric; the \(\xi\)- and \(\delta\)-vertices implement the weighted forward/backward traversal, 
    while the \(\zeta\)-vertices lead either to return vertices \(x^{\mathrm{out}}_{e,i},y^{\mathrm{out}}_{e,i}\) or to escape vertices \(x^{\mathrm{esc}}_{e,i},y^{\mathrm{esc}}_{e,i}\).}
    \label{short:fig:edge-gadget}
\end{figure}

\noindent 
\textbf{Proof intuition.}

\noindent 
The reduction encodes an orientation of the input graph by letting Player~1
choose, for each edge gadget, one of two symmetric sides during an initial
traversal phase. For an edge \(e=\{x,y\}\), traversing the \(y\)-side encodes
the orientation \((x,y)\), and traversing the \(x\)-side encodes \((y,x)\).
After all gadgets have been traversed, Player~2 chooses a vertex \(x\) in a
challenge phase and repeatedly enters the gadget sides corresponding to edges
oriented out of \(x\). If the total encoded outgoing weight at \(x\) exceeds \(t(x)\),
then Player~2 can force one more successful challenge excursion than there are
available return paths. If the encoded orientation is feasible, then every
non-losing challenge excursion consumes one challenge edge and one return path,
so Player~2 eventually runs out of non-losing moves.

\newpage
\noindent 
\textbf{The edge gadget.}

\noindent
For each input edge \(e=\{x,y\}\), we construct a directed gadget \(G_e\) with
two symmetric sides, one corresponding to \(x\) and one to \(y\); see
Figure~\ref{short:fig:edge-gadget}. Entering one side during the initial phase encodes the
orientation of \(e\). Each side contains:
\begin{itemize}
    \item a forward/backward traversal structure that simulates the weight
    \(w(e)\),
    \item a merge-and-return path through which the initial traversal leaves the
    gadget, and
    \item an escape structure that is used only in the later challenge phase.
\end{itemize}
The full vertex-and-edge definition appears in
Appendix~\ref{app:sec:hardness}.

\noindent 
\textbf{The full reduction.}

\noindent
Given an instance \((H,w,t)\) of \textsc{Chosen Maximum Outdegree}, we build
one copy of the edge gadget \(G_e\) for each \(e\in E(H)\), and chain these
gadgets in the order induced by the given path decomposition. Thus the first
phase of the game traverses the gadgets one by one and encodes an orientation of
every edge of \(H\). After the last gadget, play enters a challenge gadget in
which Player~2 chooses a vertex \(x\in V(H)\) and enters the sides corresponding
to edges oriented out of \(x\). For each vertex \(x\), the challenge gadget
provides exactly \(t(x)\) return paths. The complete formal construction is
given in Appendix~\ref{app:sec:hardness}, specifically 
Definition~\ref{def:reduction-graph}.

\begin{figure}[H]
    \centering
    \includegraphics[width=\textwidth]{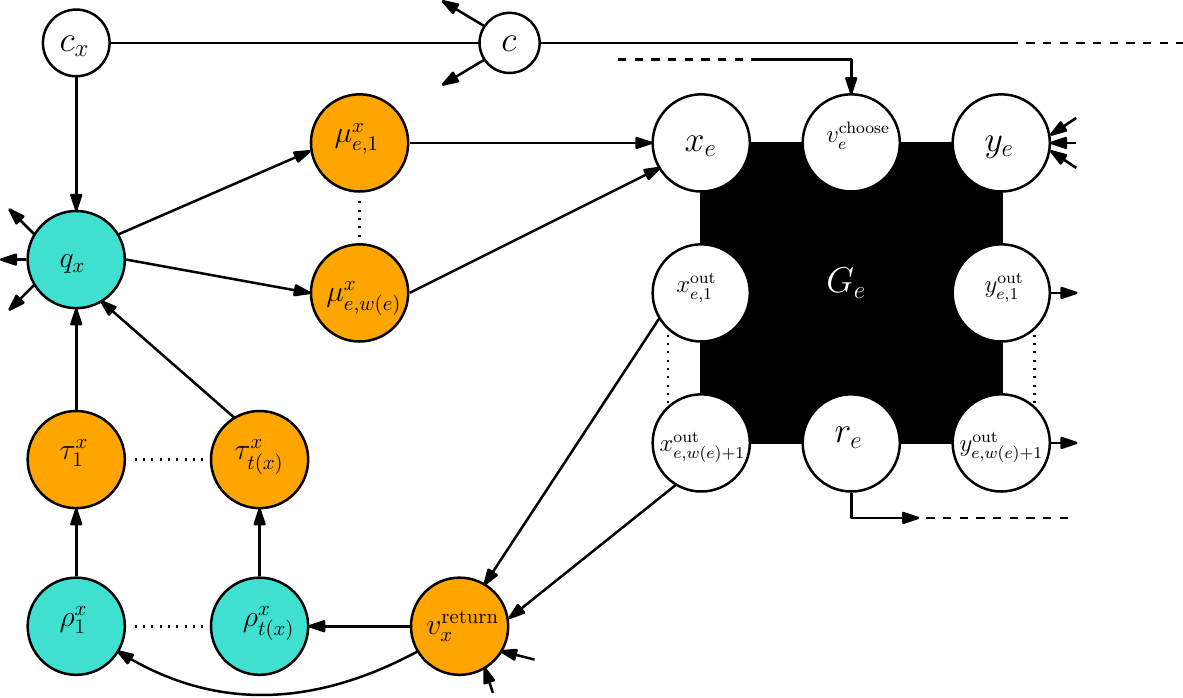}
    \caption{Figure showing how the challenge gadget for a vertex $x$ is connected 
    to an edge gadget $G_e$, where the behavior of $G_e$ is abstracted away. 
    Notice how, due to parity swap after the last edge gadget, Player~2 is now the 
    turquoise player. This is the case inside of $G_e$ also.}
    \label{short:fig:challenge-gadget}
\end{figure}

The following lemma summarizes the two key properties of the reduction; the full
proof is deferred to Appendix~\ref{app:sec:hardness}.

\begin{lemma}[Gadget behavior]
\label{lem:gadget-behavior-short}
The reduction has the following properties.
\begin{enumerate}
    \item During the initial traversal phase, once Player~1 chooses a side of an
    edge gadget, all moves of Player~2 inside that gadget are forced, and play
    leaves the gadget at its return vertex with Player~2 to move.
    \item In the challenge phase, challenging a side that was already used in the
    initial phase is immediately losing for Player~2. Challenging a side that
    was left untouched yields a non-losing excursion if and only if one
    challenge edge and one return path are consumed.
\end{enumerate}
\end{lemma}

\begin{figure}[H]
    \centering
    \includegraphics[width=\textwidth]{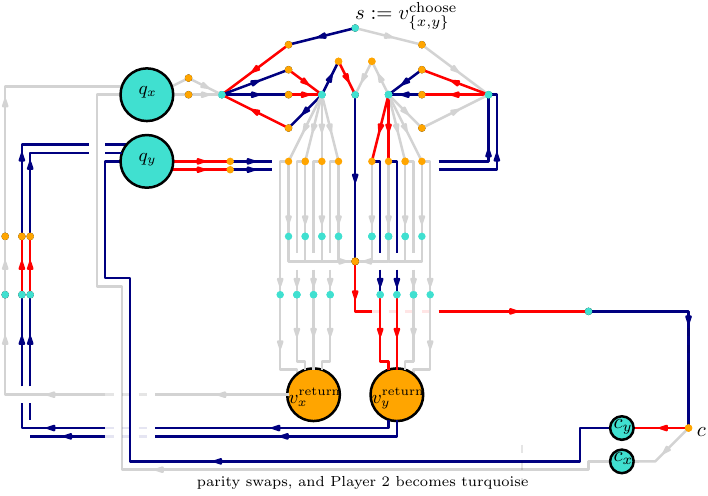}
    \caption{A reduced instance of our input problem for which 
    there is one edge $\{x,y\}$ with weight $2$, and $t(x) = 1$ and 
    $t(y) = 2$. This is a feasible instance and we observe that Player~1 has 
    a winning strategy, with Player~2 to move at 
    $q_y$ with no legal continuation. In this figure we assume not only that the 
    players are playing optimally, but that they want to stay 
    alive ''as long as possible''. The edges played by Player~1 are 
    colored navy and those played by Player~2 are colored red. In particular, 
    observe that in the phase where the edge is oriented, 
    Player~2 only has one legal response every time they move.}
    \label{short:fig:single_edge_example}
\end{figure}

\begin{theorem}
\label{thm:deg-xnlp-hard}
\textsc{Directed Edge Geography} is \textsc{XNLP}-hard when parameterized by
pathwidth.
\end{theorem}

\begin{proof}[Proof sketch]
We reduce from \textsc{Chosen Maximum Outdegree}.

\smallskip
\noindent
\emph{Correctness.}
Suppose first that \((H,w,t)\) is a yes-instance, and fix a feasible orientation
\(\omega\in \Ori(H)\). Player~1 uses the initial traversal phase to encode
\(\omega\) by choosing, for each edge gadget, the side corresponding to the
head of the oriented edge. By Lemma~\ref{lem:gadget-behavior-short}, all responses of
Player~2 inside the gadgets are forced, and after all gadgets have been
processed, play reaches the challenge gadget with Player~2 to move.

Now Player~2 chooses a vertex \(x\). The only non-losing challenges are exactly
those corresponding to edges directed out of \(x\), and there are
\(\sum_{e\in \theta_\omega(x)} w(e)\) such challenge opportunities in total.
Each non-losing challenge consumes exactly one return path, and the number of
return paths available at \(x\) is \(t(x)\). Since \(\omega\) is feasible, we
have
\[
\sum_{e\in \theta_\omega(x)} w(e)\le t(x),
\]
so eventually all non-losing challenges are exhausted and Player~2 loses.

Conversely, suppose \((H,w,t)\) is a no-instance. Whatever Player~1 does in the
initial phase, the choices of sides encode some orientation \(\omega\). Since
\((H,w,t)\) is infeasible, there exists a vertex \(x\) with
\[
\sum_{e\in \theta_\omega(x)} w(e)> t(x).
\]
Player~2 chooses this vertex in the challenge phase. By
Lemma~\ref{lem:gadget-behavior-short}, each non-losing challenge consumes exactly one
return path, and Player~2 has strictly more such challenges available than the
number of return paths at \(x\). Hence Player~2 can force one final successful
challenge after all return paths have been exhausted, leaving Player~1 without a
move. Therefore Player~2 wins.

\begin{figure}[H]
    \centering
    \includegraphics[width=\textwidth]{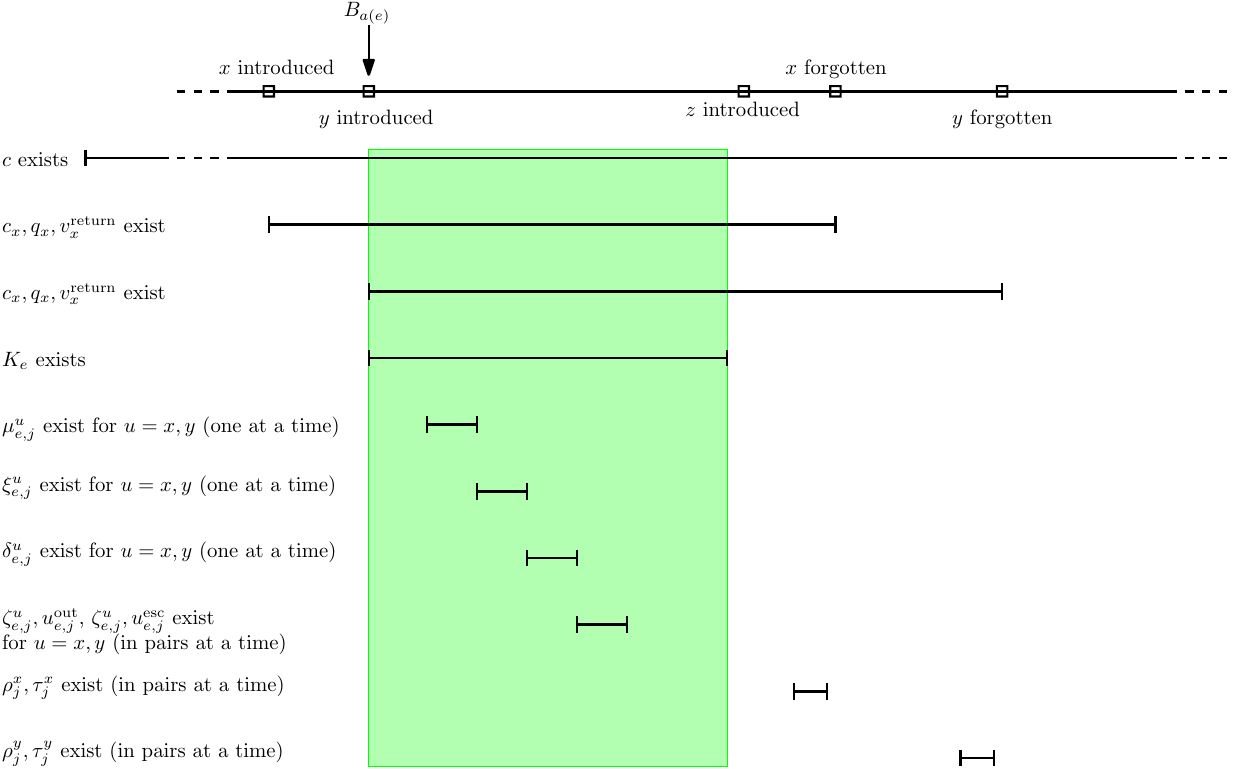}
    \caption{Schematic view of the pathwidth-preserving
    construction for an edge gadget $e=\{x,y\}$. The top timeline indicates
    the intervals during which $x$ and $y$ are present. The rows below indicate 
    the corresponding lifetimes of the anchor vertices, together with the global 
    vertex $c$. The green block indicates the dedicated contiguous block of bags 
    in the constructed decomposition where $e$ is realized: throughout 
    $K_e$, which is meant to represent the constant-size core of the edge gadget, 
    is present together with $S_i$, while the auxiliary vertices
    are introduced 
    only locally, one at a time or in pairs. Finally, before a vertex is forgotten, 
    the path-return vertices are realized in pairs at a time.}
    \label{short:fig:pathwidth-decomposition}
\end{figure}

\smallskip
\noindent
\emph{Pathwidth preservation.}
Let \((B_1,\dots,B_r)\) be a nice path decomposition of \(H\) of width \(k\).
We order the edge gadgets according to the first bag containing both endpoints
of the corresponding input edge. For each bag \(B_i\), we keep a set of
\emph{anchor vertices}: the global challenge vertex together with the challenge
vertices associated with the original vertices currently present in \(B_i\).
Since \(|B_i|\le k+1\), the number of anchor vertices in a bag is \(O(k)\).

Each edge gadget is then realized in a contiguous block of bags in which the
relevant anchor vertices are kept throughout, while the auxiliary gadget
vertices are introduced and forgotten locally. Similarly, when an original
vertex \(x\) is forgotten, the \(t(x)\) return paths associated with \(x\) are
realized in a short contiguous block before the anchor vertices of \(x\) are
removed. This yields a valid path decomposition of the reduction graph of 
width $O(k)$. The full construction and verification are given in
Appendix~\ref{app:sec:hardness}.

The construction is computable by a deterministic parameterized logspace transducer. 
Indeed, the output graph can be streamed by scanning the input path decomposition and 
enumerating edge gadgets in the order of the first bag containing both endpoints. 
Each output vertex and edge is specified by a constant-size type label, an input 
vertex or edge, and a counter bounded by \(w(e)+O(1)\) or \(t(v)+O(1)\). Since all 
integers in the source instance are encoded in unary, the output has polynomial 
size and these counters use only \(O(\log n)\) bits. The path decomposition witnessing the \(O(k)\) 
pathwidth bound is generated analogously. 

Together with the correctness argument and the pathwidth bound above, this proves
the claimed XNLP-hardness.
\end{proof}

\section{Reducing to the undirected variant}

We now transfer hardness from the directed game to the undirected one.

\begin{lemma}
\label{lem:dir-to-undir}
Let $G$ be a directed graph with $\pw(G)=k$.
There exists an undirected graph $G'$ such that:
\begin{enumerate}
    \item the \textsc{Directed Edge Geography} instance on $G$ is equivalent to the
    \textsc{Undirected Edge Geography} instance on $G'$;
    \item each directed edge of $G$ is replaced by a constant-size gadget;
    \item $\pw(G')\le f(k)$ for some function $f(k)=O(k^2)$.
\end{enumerate}
\end{lemma}

\begin{proof}
For each directed edge $(u,v)\in E(G)$, replace it by the well-known constant-size gadget
from the \emph{Undirected Edge Geography} paper~\cite[Theorem 2.1]{FraenkelScheinermanUllman1993}
that simulates a directed edge:
it admits traversal from $u$ to $v$, while entering from $v$ results in a losing position
under optimal play.

Let $(X_1,\dots,X_r)$ be a path decomposition of $G$ of width $k$.
For each edge $(u,v)$, let $I_{uv}$ be the interval of bags containing both $u$ and $v$.
Place every vertex of the gadget replacing $(u,v)$ in all bags indexed by $I_{uv}$.

All gadget edges are covered inside the corresponding interval, and contiguity is immediate.
In a fixed bag $X_i$, there are at most $\binom{k+1}{2}=O(k^2)$ pairs of vertices from $X_i$,
and each contributes only a constant number of gadget vertices.
Hence each bag increases in size by at most $O(k^2)$.
\end{proof}

\begin{figure}[H]
    \centering
    \includegraphics[width=\textwidth]{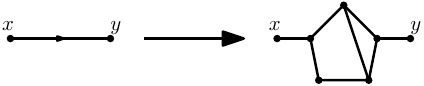}
    \caption{Pseudoarc used in reduction. Left: a directed 
    edge $(x,y)$. Right: a pseudoarc simulating the 
    directed edge on the left. We observe 
    that entry from $x$ yields safe exit at $y$ under 
    optimal play, while entry the other way results in 
    a loss for the entrant.}
    \label{fig:pseudoarc}
\end{figure}

\begin{corollary}
\textsc{Undirected Edge Geography} is XNLP-hard when parameterized by pathwidth.
\end{corollary}

\begin{proof}
By the previous theorem, \textsc{Directed Edge Geography} is XNLP-hard when
parameterized by pathwidth.
By Lemma~\ref{lem:dir-to-undir}, there is a parameter-preserving reduction from the
directed game to the undirected one.
Therefore \textsc{Undirected Edge Geography} is XNLP-hard when parameterized by
pathwidth.
\end{proof}

\begin{corollary}
    \textsc{Directed Edge Geography} can be solved in time
    $f'(k,d)\cdot n$ on directed graphs of pathwidth at most $k$
    and maximum degree at most $d$, provided with a path decomposition
    of width at most $k$.

    Moreover, the same holds on directed graphs of treewidth at most $k$
    and maximum degree at most $d$, provided with a tree decomposition
    of width at most $k$.
\end{corollary}

\begin{proof}
    Let $G$ be an instance of \textsc{Directed Edge Geography} on a directed
    graph of pathwidth at most $k$ and maximum degree at most $d$, together
    with a path decomposition of width at most $k$.
    
    Apply Lemma~\ref{lem:dir-to-undir} to obtain an equivalent instance
    $G'$ of \textsc{Undirected Edge Geography}. By the lemma, each directed
    edge of $G$ is replaced by a constant-size gadget, so $|V(G')|=O(|V(G)|+|E(G)|)$,
    the maximum degree of $G'$ is bounded by a function of $d$, and
    $\pw(G') \leq g(k)$ for some function $g$.
    
    Now \textsc{Undirected Edge Geography} is a special case of
    \textsc{Edge Generalized Geography}, and pathwidth bounds treewidth from
    above. Hence Theorem~\ref{thm:bodlaender-tw-degree} implies that $G'$ can be solved in time
    $f(g(k),d')\cdot |V(G')|$, where $d'$ is the maximum degree of $G'$.
    Since $d'$ depends only on $d$ and $|V(G')|$ is linear in the size of $G$,
    this running time is of the form $f'(k,d)\cdot n$.
    
    Because the reduction preserves the winner, the same bound holds for
    \textsc{Directed Edge Geography}.

    For the treewidth statement, one argues in the same way starting from a
    tree decomposition. Using a nice tree decomposition with introduce-edge
    bags, one inserts the constant-size gadget when the corresponding edge is
    introduced, places the new gadget vertices in that bag, and forgets them
    immediately afterwards. This yields an equivalent undirected instance whose
    treewidth and maximum degree are bounded by functions of the original
    parameters, so Bodlaender's theorem applies exactly as above.
\end{proof}

\section{An XP Algorithm on Given Tree Partitions}

In this section we show that \textsc{Undirected Edge Geography} is in XP on
simple graphs when the input is given together with a rooted tree partition of
bounded width. 
By the standard reduction from the previous section, the same
will then hold for \textsc{Directed Edge Geography}.

The key idea is to summarize each child subtree by a finite \emph{interface
type} describing how play can enter the subtree from its parent bag, how it may
later return to the parent bag, and what smaller residual child subinstance
remains afterwards. For fixed width $k$, only finitely many such types can
occur. Moreover, children of the same type are interchangeable from the
perspective of the parent bag. This allows us to compress the state of a bag to
the multiplicities of the child types that are still present.

A complete formal development of the type universe, the replacement lemma, and
the local semantics, and a worked example is given in Appendix~\ref{app:sec:algo}. 
Implementation details are given in 
Appendix~\ref{sec:code} and a proof of concept implementation can be found at

\noindent
\url{https://github.com/ThobiasKH/GeographyXPAlgorithm}.

\newpage
\subsection{Interfaces and types}

We use a rooted version of the notion of a tree partition.

\begin{figure}[ht]
    \centering
    \includegraphics[width=\textwidth]{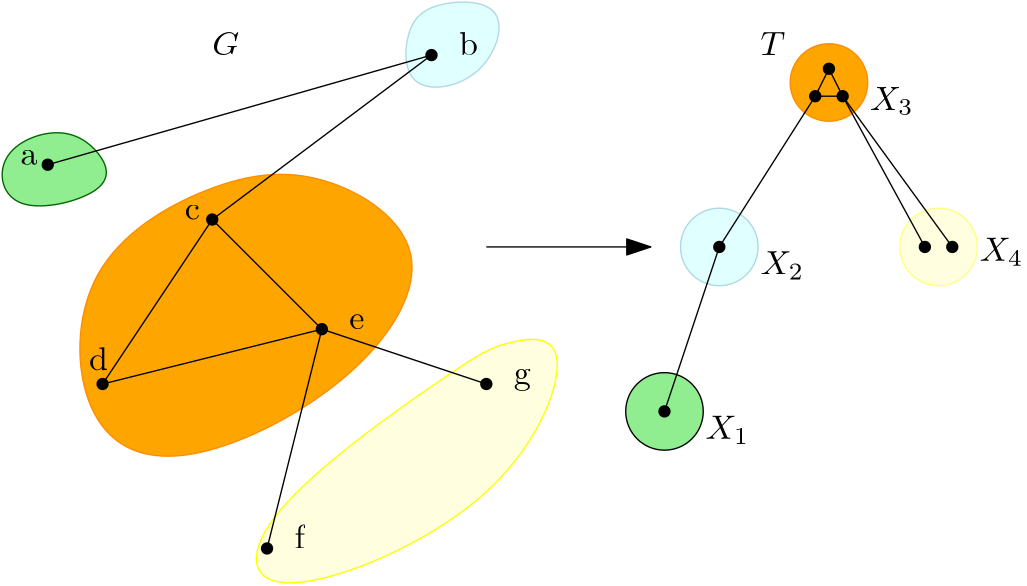}
    \caption{A tree partition of a graph $G$ with width $k = 3$.}
    \label{short:fig:tree-partition}
\end{figure}

\begin{definition}[Rooted tree partition]
A \emph{rooted tree partition} of a graph $G$ is a pair
$(\{X_i \mid i\in I\},T)$,
where $\{X_i \mid i\in I\}$ is a partition of $V(G)$ and $T$ is a rooted tree
with node set $I$, such that every edge of $G$ is either contained in a bag
$X_i$ or has its endpoints in two bags whose nodes are adjacent in $T$.
The \emph{width} of the tree partition is $\max_{i\in I}|X_i|$.
\end{definition}

Throughout this section, let 
\((\{X_i \mid i \in I\}, T)\) be a rooted tree partition 
of width $k$ of the input graph $G$. 
If $i$ is a non-root node, let $p(i)$ denote its parent in $T$. 
The \emph{parent cut} of $X_i$ is the set of edges with one 
endpoint in $X_i$ and the other endpoint in $X_{p(i)}$. We call 
each edge in this cut a \emph{port} of $X_i$. The endpoint in 
$X_{p(i)}$ is the \emph{parent-side endpoint} of the port, and 
the endpoint in $X_i$ is the \emph{child-side endpoint}. 

Since $G$ is simple and both bags have size at most $k$, the parent cut 
has size at most $k^2$. After fixing an injective labeling 
$\iota_i : X_i \hookrightarrow [k]$ for each bag $X_i$, every port 
receives two labels: the label of its parent-side endpoint and the label of its child-side 
endpoint.

\noindent 
\textbf{Interface type, abbreviated.}

\noindent 
Let $(P,\lambda)$ be a labeled interface, where $P$ is a set of ports and
$\lambda:P\to [k]$ records the parent-side labels. An interface type on
$(P,\lambda)$ consists, for each entry port $a\in P$, of:
\begin{enumerate}
    \item a set of possible \emph{exit labels} $(q,\tau,b)$, where
    $q\in P\setminus\{a\}$ is the return port, $\tau$ is the type of the
    smaller residual child subinstance that remains after the excursion, and
    $b\in\{\mathsf{same},\mathsf{opp}\}$ records the parity of the return;
    \item a Boolean function $\Phi_a$ telling whether the player entering
    through $a$ can force a win, given the truth values of the attainable
    returns.
\end{enumerate}
The fully rigorous recursive definition appears in
Appendix~\ref{app:sec:algo}.

Intuitively, an interface type is a finite summary of the boundary behavior of a
child subtree: it records how play may enter the subtree, return to the parent
bag, and leave behind a smaller residual subinstance. The set of attainable
returns alone is not enough, since the value of entering the child also depends
on which of those returns are winning in the surrounding parent-side game. The
Boolean functions $\Phi_a$ encode exactly this dependence. 
A useful mental model is thinking of any subtree of a bag $X_i$ as a 
gadget that may only be entered (and possibly exited) a certain number of times, and 
the parent does not care about the exact inner workings of this gadget.

For fixed $k$, there are only finitely many realizable interface types. We
denote this finite set by $\mathcal T_k$, and note that $|\mathcal T_k|$ is entirely 
dependent on $k$.

\subsection{Compressed residual configurations}

Fix a non-root bag $X_i$. A \emph{compressed residual configuration} at $i$ is a
triple
$\Gamma=(U,F,\mathbf m)$,
where
\begin{itemize}
    \item $U$ is the set of unused parent ports of $i$,
    \item $F\subseteq E(G[X_i])$ is the set of unused internal edges of $X_i$,
    \item $\mathbf m=(m_\sigma)_{\sigma\in\mathcal T_k}$ is a multiplicity
    vector recording how many residual child subinstances of each realizable
    type are still present below $i$.
\end{itemize}

The point of compression is that the identities of the child subinstances do not
matter: only their types do. The Appendix proves a replacement lemma showing
that if one residual child subinstance is replaced by another of the same type,
then the parent-side outcome does not change. Consequently, the type of a bag
configuration is determined entirely by $(U,F,\mathbf m)$.

\begin{figure}[H]
    \centering
    \includegraphics[width=\textwidth]{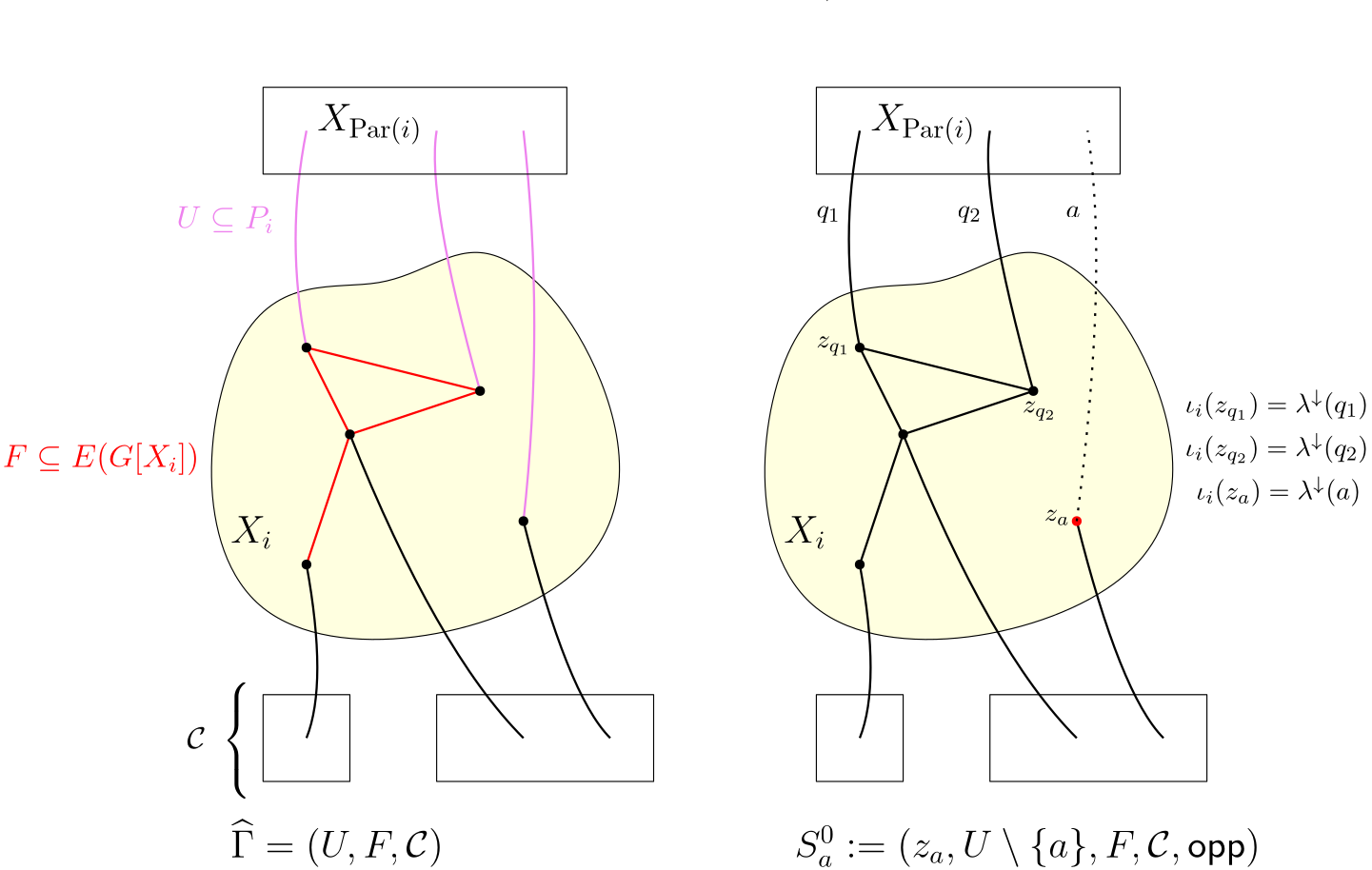}
    \caption{A residual configuration and entry state at a bag $X_i$. 
    Left: an expanded residual configuration $\widehat\Gamma = (U,F,\mathcal C)$ consisting of unused 
    parent ports $U$, unused internal edges $F$, and multiset of residual child subinstances $\mathcal C$ (In the 
    full technical development of this approach we consider multisets of 
    child subinstances before showing that this multiset can be compressed to a multiplicity vector). 
    Right: after entering through $a\in U$, the local game starts in the entry state 
    $S_a^0 = (z_a, U\setminus \{a\},F,\mathcal C, \mathsf{opp})$, where $z_a \in X_i$ is the child-side endpoint of $a$. 
    From this entry state, play may only proceed by a child-side excursion.}
    \label{short:fig:configuration_and_states}
\end{figure}

\subsection{Local games inside a bag}

Let $\Gamma=(U,F,\mathbf m)$ be a compressed residual configuration at a
non-root node $i$, and let $a\in U$ be an entry port. We consider the local game
that starts when play enters the current subtree through $a$.

A \emph{compressed local state} has the form
$S=(z,U',F',\mathbf m',\pi)$,
where
\begin{itemize}
    \item $z\in X_i$ is the current token position,
    \item $U'\subseteq U\setminus\{a\}$ is the set of still unused parent ports,
    \item $F'\subseteq F$ is the set of still unused internal bag edges,
    \item $\mathbf m'\le \mathbf m$ is the current multiplicity vector of
    residual child types,
    \item $\pi\in\{\mathsf{same},\mathsf{opp}\}$ records whether the player to
    move is the same as or opposite to the player who originally entered
    through $a$.
\end{itemize}

From such a state, play may proceed in one of three ways:
\begin{enumerate}
    \item \emph{Internal move:} traverse an unused edge of $G[X_i]$ incident
    with $z$;
    \item \emph{Exit move:} leave the current subtree through a port
    $q\in U'$ whose child-side label matches $z$;
    \item \emph{Child excursion:} enter a residual child subinstance of some
    type $\sigma$ with $m'_\sigma>0$ through a port whose label matches $z$.
    If the child later returns with exit label $(q,\tau,b)$, then the successor
    state replaces one copy of $\sigma$ in the multiplicity vector by one copy
    of $\tau$, moves the token to the port $q$, and updates the parity by $b$.
\end{enumerate}

The full formal statement of these transitions, together with the proof that
they faithfully capture the underlying expanded game, is deferred to
Appendix~\ref{app:sec:algo}.
The crucial point is that this local game is acyclic. Indeed, if we define the
measure
\[
\mu(S):=
|U'|+|F'|+\sum_{\sigma\in\mathcal T_k}\operatorname{rk}(\sigma)m'_\sigma,
\]
then every internal move strictly decreases $|F'|$, and every returning child
excursion replaces a type by a strictly smaller residual type, thereby
decreasing the rank contribution. Hence every transition strictly decreases
$\mu(S)$.

It follows that for every pair $(\Gamma,a)$, the local game can be solved by
reverse dynamic programming on a finite acyclic graph. From these values we can
compute:
\begin{itemize}
    \item the set $A_a^\Gamma$ of attainable exit labels from entry through $a$,
    \item the Boolean function $\Phi_a^\Gamma$ describing whether the entrant
    can force a win as a function of the values of those returns.
\end{itemize}
Together, these determine the interface type $\tau(\Gamma)$ of the compressed
residual configuration.

\subsection{Bottom-up computation and the root bag}

We now compute all relevant types bottom-up over the rooted tree partition.

Fix a node $i$. Since the children of $i$ have already been processed, the set
of realizable child types that may appear below $i$ is known. 
We consider all compressed residual configurations arising at \(i\), ordered by measure.
For each such $\Gamma$ and each
entry port $a\in U$, we solve the associated local game as above and thereby
compute $\tau(\Gamma)$.

For fixed \(k\), let \(t_k:=|\mathcal T_k|\). Since the parent cut of a bag has
size at most \(k^2\), there are at most \(2^{k^2}\) choices for the unused port
set \(U\). Since a bag contains at most \(\binom{k}{2}\) internal edges, there
are at most \(2^{\binom{k}{2}}\) choices for \(F\). Finally, each multiplicity
vector \(\mathbf m=(m_\sigma)_{\sigma\in\mathcal T_k}\) has \(t_k\) coordinates,
each lying in \(\{0,\dots,n\}\), and hence there are at most
\((n+1)^{t_k}\) possibilities for \(\mathbf m\). Thus the number of compressed
residual configurations at a bag is at most
\(2^{k^2+\binom{k}{2}}(n+1)^{t_k}=n^{f_1(k)}\).

For a fixed pair \((\Gamma,a)\), a compressed local state is determined by
\(z\in X_i\), a subset \(U'\subseteq U\setminus\{a\}\), a subset \(F'\subseteq F\),
a multiplicity vector \(\mathbf m'\le \mathbf m\), and a parity bit. Therefore
the number of compressed local states is at most
\(2k\cdot 2^{k^2+\binom{k}{2}}(n+1)^{t_k}=n^{f_2(k)}\).
Moreover, each such state has only \(f_3(k)\) outgoing transitions: the number
of internal moves is bounded by \(\binom{k}{2}\), the number of exits by \(k^2\),
and the number of child excursions by a function of \(k\), since \(\mathcal T_k\)
is finite and each type has rank at most \(k^2\). Since the local game is
acyclic, its values can therefore be computed in time \(n^{f_4(k)}\). Summing
over all bags still yields total running time \(n^{f(k)}\).

At the root bag $X_r$, there is no parent interface. After all child types have
been computed, we solve an analogous compressed root game whose states have the
form $R=(z,F,\mathbf m,\pi)$,
where $z\in X_r$, $F\subseteq E(G[X_r])$, $\mathbf m$ is the multiplicity vector
of the residual child types of the root’s children, and
$\pi\in\{\mathsf{same},\mathsf{opp}\}$ records the player to move relative to
the initial player.

As before, the root game is acyclic, since every internal move consumes an
internal root edge and every returning child excursion strictly decreases the
total rank contribution of the child types. Thus it can again be solved by
reverse dynamic programming. The initial root state is
$R_0=(s,E(G[X_r]),\mathbf m_r,\mathsf{same})$,
and Player~1 has a winning strategy in the original instance if and only if
$R_0$ is winning in this compressed root game.

\begin{theorem}
\label{thm:xp-tpw-short}
\textsc{Undirected Edge Geography} on simple graphs, given together with a
rooted tree partition of width $k$, is solvable in XP time parameterized by
$k$.
\end{theorem}

\begin{proof}[Proof sketch]
For fixed \(k\), only finitely many realizable interface types can occur. At a
bag \(X_i\), a compressed residual configuration is determined by the unused
parent ports, the unused internal bag edges, and a multiplicity vector over the
finite type universe \(\mathcal T_k\). Hence the number of such configurations
is bounded by \(n^{f_1(k)}\). For each configuration and entry port, the
associated compressed local game has at most \(n^{f_2(k)}\) states and only
\(f_3(k)\) outgoing transitions per state, and is acyclic. Its values can
therefore be computed by reverse dynamic programming in time \(n^{f_4(k)}\).
Processing all bags bottom-up and finally solving the analogous root game yields
the winner of the original instance in time \(n^{f(k)}\). Thus the problem
belongs to XP.
\end{proof}

\begin{corollary}
\textsc{Directed Edge Geography} on simple graphs, given together with a rooted
tree partition of width $k$, is solvable in XP time parameterized by $k$.
\end{corollary}

\begin{proof}[Proof sketch]
Reduce the directed instance to an equivalent undirected instance by replacing
each directed edge by the standard constant-size gadget used in the 
reduction earlier. As in Appendix~\ref{app:sec:algo}, the given rooted tree
partition can be transformed into one of width bounded by a function of $k$ for
the resulting undirected graph. The theorem above then applies.
\end{proof}

\begin{remark}
By a result of Bodlaender, Groenland, and Jacob~\cite{BodlaenderGroenlandJacob2022},
given an $n$-vertex graph $G$ and an integer $k$, one can in time $k^{O(1)}n^2$
either compute a tree partition of width $O(k^7)$ or correctly conclude that
the tree-partition-width of $G$ is greater than $k$. Since a tree partition can
be rooted arbitrarily without changing its width, Theorem~\ref{thm:xp-tpw-short}
also yields XP algorithms for \textsc{Undirected Edge Geography} and
\textsc{Directed Edge Geography} parameterized by tree-partition width.
\end{remark}

\section{Conclusion}

We proved that both \textsc{Directed Edge Geography} and
\textsc{Undirected Edge Geography} are \textsc{XNLP}-hard when parameterized by
pathwidth. Bodlaender left open what width bounds alone imply for
\textsc{Edge Generalized Geography}. Our \textsc{XNLP}-hardness result for
pathwidth does not settle the possibility of polynomial-time solvability for
each fixed width bound, but it shows that width alone is unlikely to yield an
\textsc{FPT}-type parameterized algorithm, unless \textsc{FPT}=\textsc{XNLP}.
On the positive side, we showed that both variants are in \textsc{XP} on simple
graphs when parameterized by tree-partition width. We also showed that
\textsc{Directed Edge Geography} is fixed-parameter tractable when
parameterized by treewidth and maximum degree. Together, these results clarify
that width-based tractability for edge geography depends strongly on the
decomposition model: pathwidth already captures enough structure for hardness,
while tree partitions still admit a decomposition-based dynamic program.

\subsection{Discussion}

Our hardness result suggests that the difficulty of the pathwidth case is not
merely technical. In edge-deletion geography, a small separator does not
localize play in the same way as for many vertex-based problems. Separator
vertices may be revisited many times, while the incident edges are consumed
one by one, so the interaction across the separator depends on a history that
is not naturally summarized by a small local state. From this perspective, the
main obstacle to upper bounds for pathwidth is not simply the lack of the
right dynamic program, but the apparent absence of a compact interface that
captures the residual effect of partial play.

At the same time, the XP algorithm on rooted tree partitions shows that
decomposition-based methods do become viable once the interaction between
pieces is sufficiently restricted. In that setting, one can push richer
interface information through the decomposition and still retain enough
structure for a finite-type compression argument. This suggests that the right
structural boundary for geography-type games may lie not between hard and easy
graph classes in the ordinary sense, but between decomposition models that do
or do not support a controlled summary of residual play.

More broadly, edge geography may be viewed as a canonical token-moving model
for alternating play on explicit residual-state graphs. For example, it can
serve as an intermediate model for reductions from games or reconfiguration
processes in which moves irreversibly consume resources, force local responses,
or gradually restrict the future state space. Whenever the evolution of a game
can be represented by an acyclic residual-state graph of manageable size,
winner determination reduces to a geography-type reachability problem on that
graph. Even in settings where the residual-state graph contains cycles, a
similar perspective may still apply after augmenting states with bounded
progress information that restores acyclicity. In this sense, the present work
suggests a broader program of studying which decomposition parameters can
support structural algorithms for geography-type games and related
\textsc{PSPACE}-complete problems.

Methodologically, \textsc{Edge Geography} is also a useful intermediate
problem. Reductions are often more natural to design in the directed setting,
where one can encode asymmetric choices and forced traversals more explicitly.
On the other hand, undirected formulations are often more intuitive from an
algorithmic point of view. In our setting, this asymmetry is especially
helpful because \textsc{Directed Edge Geography} admits a simple constant-size
gadget reduction to \textsc{Undirected Edge Geography} that preserves the
winner and controls the structural parameters typically of interest. Thus one
may often work in the directed variant when designing reductions, and then
transfer the result to the undirected variant essentially for free.

\subsection{Future Work}

The main open question is the parameterized complexity of
\textsc{Directed Edge Geography} and \textsc{Undirected Edge Geography} with
respect to pathwidth alone. In particular, it remains open whether either
problem admits an XP algorithm, or belongs to some other natural class under this
parameterization, or whether these problems could be complete for 
a larger parameterized space class, such as para-PSPACE. 

Furthermore, it would be interesting to test whether the interface-based viewpoint
developed here extends to other PSPACE-complete graph games and
residual-state models. This may help identify which kinds of boundary
information are sufficient for decomposition-based algorithms beyond the
specific case of edge geography.

Lastly, from the reviewers' feedback we posit the following conjectures and open questions.  

\begin{conjecture}[Treewidth XALP lower bound]
    \textsc{Directed Edge Geography} is XALP-hard parameterized by treewidth. 
\end{conjecture}
If the above conjecture holds then, by the transfer argument from earlier, 
XALP-hardness holds for the undirected version also. 

\begin{conjecture}[Tree-partition width XALP lower bound]
    \textsc{Undirected Edge Geography} is XALP-hard parameterized by tree-partition width. 
\end{conjecture}

If the above conjecture holds then it is natural to ask whether
there exists an XALP algorithm for \textsc{Undirected Edge Geography} 
parameterized by tree-partition width, and whether the same lower bound holds for 
the directed variant aswell.

\bibliographystyle{plainurl}
\bibliography{bibliography}

\newpage 
\appendix 
\section{Full Hardness Construction, Omitted Proofs and Example} \label{app:sec:hardness}

We prove XNLP-hardness of \textsc{Directed Edge Geography} by a reduction from
\textsc{Chosen Maximum Outdegree}, which is known to be XNLP-complete when
parameterized by pathwidth~\cite{bodlaender2022hardtw}.
Following~\cite[Section~2.3]{bodlaender2022hardtw}, we assume that all integers 
in the source instance are encoded in unary.

\begin{definition}
Let $G$ be an undirected graph.
An \emph{orientation} of $G$ is a function
\[
\omega : E(G)\to \{(u,v)\in V(G)\times V(G)\mid \{u,v\}\in E(G)\}
\]
such that for every $e=\{u,v\}\in E(G)$ we have
$
\omega(e)\in \{(u,v),(v,u)\}.
$
\end{definition}

\begin{definition}
Let $G$ be an undirected graph, let $\omega$ be an orientation of $G$, and let $x\in V(G)$.
We define
$
\theta_\omega(x)=\{\, e\in E(G)\mid \omega(e)=(x,v)\text{ for some }v\in V(G)\,\}.
$
\end{definition}

\begin{definition}
For an undirected graph $G$, let $\Ori(G)$ denote the set of all orientations of $G$.
\end{definition}

\begin{problem}{Chosen Maximum Outdegree}
\textbf{Input:} An undirected graph $G=(V,E)$, a weight function
$w:E\to \mathbb{Z}_{>0}$, and a vertex bound
$t:V\to \mathbb{Z}_{>0}$.

\noindent 
\textbf{Question:} Is there an orientation $\omega\in \Ori(G)$ such that
\[
\sum_{e\in \theta_\omega(v)} w(e) \le t(v)
\]
for every vertex $v\in V(G)$?
\end{problem}

Note that this is exactly \textsc{Chosen Maximum Outdegree} as defined in \emph{Problems hard for treewidth
but easy for stable gonality}~\cite{bodlaender2022hardtw}, but with some 
notational differences from what is typically found in the litterature. 

When we refer to \textsc{Chosen Maximum Outdegree} parameterized by
pathwidth, we mean the decomposition-given variant in which the input
additionally contains a path decomposition of the graph \(G\) of width at most
\(k\), and the parameter is \(k\). This is the version used in the reduction
below.

\begin{figure}[H]
    \centering
    \includegraphics[width=\textwidth]{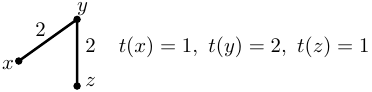}
    \caption{An instance of \textsc{Chosen Maximum Outdegree} which 
    admits no feasible orientations.}
    \label{fig:orientation}
\end{figure}

\noindent 
\textbf{Proof intuition.} 

\noindent 
The reduction encodes an orientation of the input graph by choosing, for each 
edge gadget, one of its two sides during an initial traversal phase. Choosing
the $y$-side of the gadget for $e=\{x,y\}$ corresponds to orienting $e$ as
$(x,y)$, and symmetrically for the other choice.

After all edge gadgets have been traversed, Player~2 selects a vertex $x$ in a
challenge phase. The challenge gadget then allows Player~2 to repeatedly enter
the sides corresponding to edges oriented out of $x$. If the total weight
oriented out of $x$ exceeds $t(x)$, then Player~2 can force one more successful
challenge excursion than there are available return paths. If the orientation is
valid, then every non-losing challenge excursion consumes one challenge edge and
one return path, and Player~2 eventually runs out of non-losing challenge
moves.

\begin{figure}[ht]
    \centering
    \includegraphics[width=\textwidth]{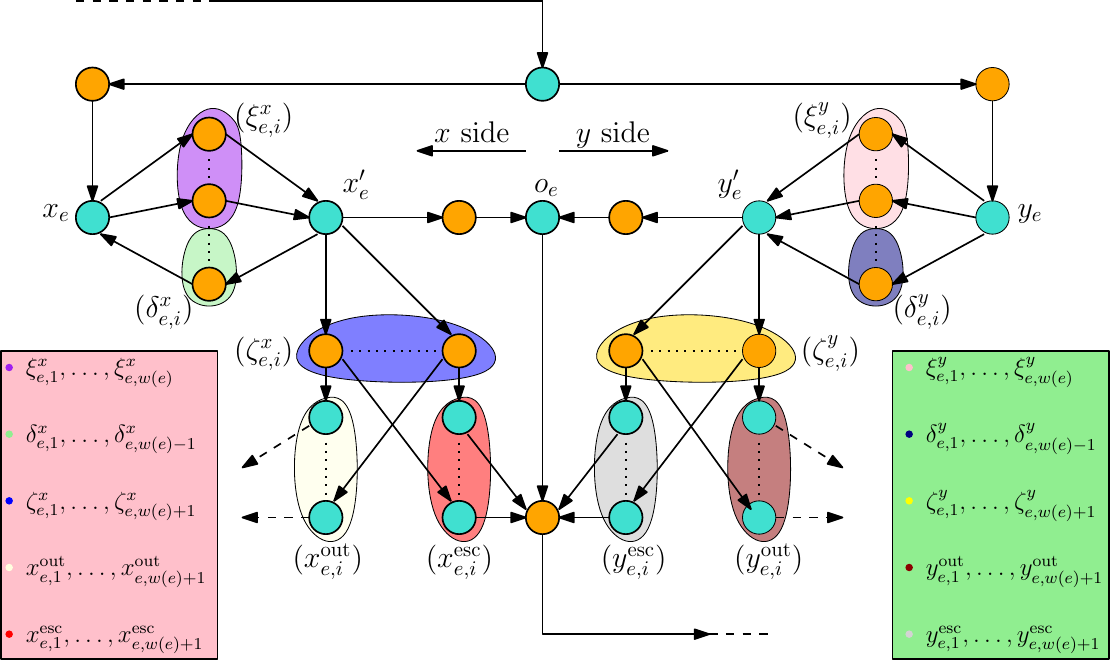}
    \caption{Edge gadget used in the reduction. 
    The \(x\)- and \(y\)-sides are symmetric; the \(\xi\)- and \(\delta\)-vertices implement the weighted forward/backward traversal, 
    while the \(\zeta\)-vertices lead either to return vertices \(x^{\mathrm{out}}_{e,i},y^{\mathrm{out}}_{e,i}\) or to escape vertices \(x^{\mathrm{esc}}_{e,i},y^{\mathrm{esc}}_{e,i}\).}
    \label{fig:edge-gadget}
\end{figure}

\begin{definition}[Edge gadget]
Let $(H,w,t)$ be an instance of \textsc{Chosen Maximum Outdegree}, and let
$e=\{x,y\}\in E(H)$.
We construct a directed graph $G_e$ as follows.

\medskip

\textbf{Vertices.}
Introduce fresh vertices
\[
v_e^{\mathrm{choose}},\quad
v_e^{\mathrm{choose},x},\quad
v_e^{\mathrm{choose},y},\quad
x_e,\quad y_e,\quad
x'_e,\quad y'_e,\quad
o_e,\quad r_e.
\]

\medskip

\textbf{Choice edges.}
Add directed edges
\[
v_e^{\mathrm{choose}} \to v_e^{\mathrm{choose},x} \to x_e,
\qquad
v_e^{\mathrm{choose}} \to v_e^{\mathrm{choose},y} \to y_e.
\]

\medskip

\textbf{Weight simulation on the $x$-side.}
For each $i\in [w(e)]$, introduce a vertex $\xi^x_{e,i}$.
For each $i\in [w(e)-1]$, introduce a vertex $\delta^x_{e,i}$.
Add the paths
$x_e\to \xi^x_{e,i}\to x'_e$ for every $i\in [w(e)]$ and
$x'_e\to \delta^x_{e,i}\to x_e$ for every
$i\in [w(e)-1]$.

\medskip
\noindent
\textbf{Weight simulation on the $y$-side.}
For each $i\in [w(e)]$, introduce a vertex $\xi^y_{e,i}$.
For each $i\in [w(e)-1]$, introduce a vertex $\delta^y_{e,i}$.
Add the paths
$
y_e\to \xi^y_{e,i}\to y'_e$ for every $i\in [w(e)]$
and
$
y'_e\to \delta^y_{e,i}\to y_e$
for every $i\in [w(e)-1]$.

\medskip
\noindent
\textbf{Merge and return.}
Introduce fresh subdivision vertices $m^x_e$ and $m^y_e$, and add
\[
x'_e\to m^x_e\to o_e,\qquad y'_e\to m^y_e\to o_e,\qquad o_e\to r_e.
\]

\medskip
\noindent
\textbf{Escape structure.}
For each $i\in [w(e)+1]$, introduce vertices
$\zeta^x_{e,i},\ x^{\mathrm{out}}_{e,i},\ x^{\mathrm{esc}}_{e,i}$ and 
$\zeta^y_{e,i},\ y^{\mathrm{out}}_{e,i},\ y^{\mathrm{esc}}_{e,i}$.
Add
$x'_e\to \zeta^x_{e,i}$ and $y'_e\to \zeta^y_{e,i}$,
together with
$\zeta^x_{e,i}\to x^{\mathrm{out}}_{e,i}$ and $\zeta^x_{e,i}\to x^{\mathrm{esc}}_{e,i}$,
and symmetrically on the $y$-side.
Finally add $x^{\mathrm{esc}}_{e,i}\to r_e,\qquad y^{\mathrm{esc}}_{e,i}\to r_e$ for 
every $i\in [w(e)+1]$.
\end{definition}

\begin{definition}[Full reduction graph]
\label{def:reduction-graph}
Let $(H,w,t)$ be an instance of \textsc{Chosen Maximum Outdegree}.
We construct a directed graph $G$ with designated start vertex $s$ as follows.

Fix, for each edge $e\in E(H)$, a bag of the given path decomposition that contains both
endpoints of $e$, and order the edges as $e_1,\dots,e_m$ by nondecreasing chosen bag index.

\medskip
\noindent
\textbf{Edge gadgets.}
For each $e_i$, construct $G_{e_i}$.
Set $s:=v_{e_1}^{\mathrm{choose}}$.
For each $i<m$, add the edge $r_{e_i}\to v_{e_{i+1}}^{\mathrm{choose}}$.

\medskip
\noindent
\textbf{Challenge gadget.}
Introduce fresh vertices $c$ and $m_c$, and add $r_{e_m}\to m_c\to c$.

For each vertex $x\in V(H)$, introduce fresh vertices $q_x$ and $c_x$, and add $c\to c_x\to q_x$.

For each edge $e$ incident with $x$ and each $i\in [w(e)]$, introduce a vertex
$\mu^x_{e,i}$ and add $q_x \to \mu^x_{e,i}\to x_e$.

Introduce a fresh vertex $v_x^{\mathrm{return}}$.
For every vertex $x^{\mathrm{out}}_{e,i}$ arising from an edge $e$ incident with $x$, add
$x^{\mathrm{out}}_{e,i}\to v_x^{\mathrm{return}}$.

Finally, for each $i\in [t(x)]$, introduce fresh vertices $\rho^x_i,\tau^x_i$ and add
$v_x^{\mathrm{return}}\to \rho^x_i \to \tau^x_i \to q_x$.
\end{definition}

\begin{figure}[H]
    \centering
    \includegraphics[width=\textwidth]{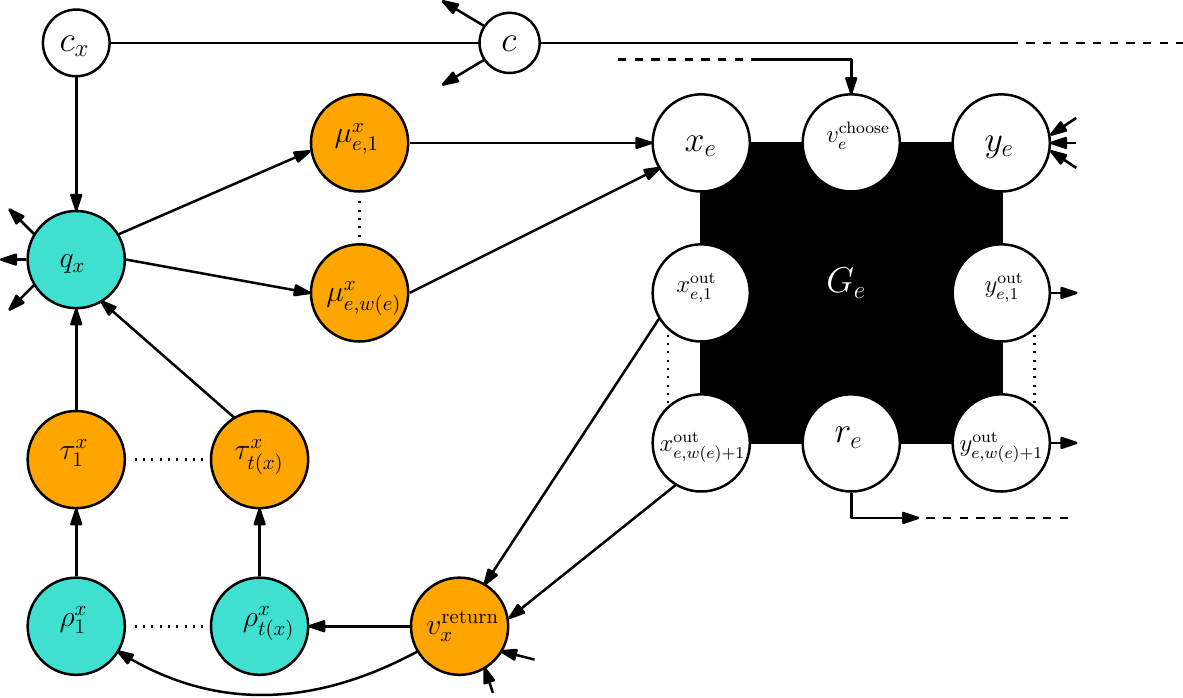}
    \caption{Figure showing how the challenge gadget for a vertex $x$ is connected 
    to an edge gadget $G_e$, where the behavior of $G_e$ is abstracted away. 
    Notice how, due to parity swap after the last edge gadget, Player~2 is now the 
    turquoise player. This is the case inside of $G_e$ also.}
    \label{fig:challenge-gadget}
\end{figure}

\noindent
\textbf{Parity convention.}

\noindent 
In the analysis below, it is important to keep track of the player to move at
the boundary vertices of each gadget. In particular, after an edge gadget has
been traversed according to the intended strategy, play reaches $r_e$ with
Player~2 to move.

\begin{lemma}[Phase 1 invariant]
\label{lem:phase1-invariant}
Fix an edge $e=\{x,y\}$. Suppose play reaches $v_e^{\mathrm{choose}}$ with
Player~1 to move, and Player~1 chooses one of the two sides of $G_e$ and
thereafter never moves from $x'_e$ or $y'_e$ to a $\zeta$-vertex inside $G_e$.
Then:
\begin{enumerate}
    \item all moves of Player~2 inside $G_e$ are forced;
    \item on the chosen side, Player~1 can traverse exactly $w(e)$ forward paths
    and exactly $w(e)-1$ backward paths before reaching the corresponding
    primed vertex again with no unused backward path remaining;
    \item from there, Player~1 moves to the corresponding merge vertex and then
    along $o_e\to r_e$;
    \item when play leaves $G_e$, the token is at $r_e$ and it is Player~2's turn.
\end{enumerate}
\end{lemma}

\begin{proof}
Assume Player~1 chooses the $x$-side; the $y$-side is symmetric.

After the move
$
v_e^{\mathrm{choose}} \to v_e^{\mathrm{choose},x}
$,
Player~2 is forced to move to $x_e$. Whenever the token is at $x_e$ and it is
Player~1's turn, Player~1 chooses an unused forward edge
$
x_e \to \xi^x_{e,i}
$,
and Player~2 is then forced to continue to $x'_e$. Whenever the token is at
$x'_e$ and there is an unused backward path, Player~1 chooses an unused edge
$
x'_e \to \delta^x_{e,j}
$,
and Player~2 is forced back to $x_e$.

There are exactly $w(e)$ forward paths and exactly $w(e)-1$ backward paths on
the $x$-side. Hence Player~1 can traverse all forward paths and all backward
paths, and after the last forward traversal reaches $x'_e$ with no unused
backward path remaining. At that point the only non-escape move on the chosen
side is
$
x'_e \to m^x_e
$,
followed by the forced move
$
m^x_e \to o_e
$,
after which Player~1 plays
$
o_e \to r_e
$.
Thus play leaves the gadget at $r_e$, and since the last move was made by
Player~1, it is Player~2's turn there.
\end{proof}

\begin{lemma}[Challenge excursions]
\label{lem:challenge-excursion}
Fix an edge $e=\{x,y\}$ and a vertex $x\in e$. Suppose Phase~1 has been
completed, so in particular the edge $o_e\to r_e$ has already been used.

\begin{enumerate}
    \item If the $x$-side of $G_e$ was traversed during Phase~1, then every
    challenge move $q_x \to \mu^x_{e,i} \to x_e$
    is immediately losing for Player~2, because all forward edges out of $x_e$
    have already been exhausted.

    \item If the $x$-side of $G_e$ was \emph{not} traversed during Phase~1, then
    after any challenge move $q_x \to \mu^x_{e,i} \to x_e$
    the token reaches $x_e$ with Player~2 to move. From that point on, every
    non-losing continuation has the following form:
    \begin{itemize}
        \item Player~2 repeatedly chooses unused forward edges
        $x_e\to \xi^x_{e,j}$;
        \item after each such move, Player~1 is forced to move to $x'_e$;
        \item at $x'_e$, Player~2 may either use an unused backward edge
        $x'_e\to \delta^x_{e,j'}$ (if one exists), or choose a $\zeta$-edge
        $x'_e\to \zeta^x_{e,\ell}$;
        \item if Player~2 moves to $m^x_e$, then Player~1 is forced to $o_e$ and
        Player~2 immediately loses, because $o_e\to r_e$ has already been used;
        \item if Player~2 moves to some $\zeta^x_{e,\ell}$, then Player~1's only
        non-losing reply is
        $\zeta^x_{e,\ell}\to x^{\mathrm{out}}_{e,\ell}$,
        since choosing $x^{\mathrm{esc}}_{e,\ell}$ leads to an immediate loss;
        \item after that, Player~2 is forced to move to
        $v_x^{\mathrm{return}}$. If a return path remains, then Player~1 chooses
        one and play returns to $q_x$ with Player~2 to move. If no return path
        remains, then Player~1 loses at $v_x^{\mathrm{return}}$.
    \end{itemize}
\end{enumerate}

Consequently, every non-losing challenge excursion on an untouched $x$-side
consumes exactly one challenge edge out of $q_x$ and exactly one return path
through $v_x^{\mathrm{return}}$, while possible uses of backward edges only
consume additional internal edges on the $x$-side.
\end{lemma}

\begin{proof}
For part~(1), if the $x$-side was traversed during Phase~1, then by
Lemma~\ref{lem:phase1-invariant} all $w(e)$ forward edges out of $x_e$ were
used there. Hence after $q_x \to \mu^x_{e,i} \to x_e$
the token is at $x_e$ with Player~2 to move and no unused outgoing edge, so
Player~2 loses immediately.

For part~(2), suppose instead that the $x$-side was untouched during Phase~1.
After $q_x \to \mu^x_{e,i}$,
Player~1 is forced to move to $x_e$, so Player~2 is indeed to move at $x_e$.
From $x_e$, the only outgoing edges are the forward edges
$x_e\to \xi^x_{e,j}$,
and from each $\xi^x_{e,j}$ the move to $x'_e$ is forced.

At $x'_e$, Player~2 has three kinds of options:
\begin{enumerate}
    \item an unused backward edge to some $\delta^x_{e,j'}$;
    \item the merge edge to $m^x_e$;
    \item a $\zeta$-edge.
\end{enumerate}
If Player~2 chooses the merge edge, then Player~1 is forced to $o_e$, and now
Player~2 has no legal move because $o_e\to r_e$ was already used in Phase~1.
So that move is immediately losing for Player~2.

If Player~2 chooses a backward edge, then Player~1 is forced back to $x_e$.
This consumes one backward path and one additional forward path on the
$x$-side, but it does not consume any challenge edge out of $q_x$ or any return
path.

If Player~2 chooses a $\zeta$-edge, then the token reaches a vertex
$\zeta^x_{e,\ell}$ with Player~1 to move. From there, moving to
$x^{\mathrm{esc}}_{e,\ell}$ is immediately losing for Player~1, because
Player~2 is then forced to $r_e$ and Player~1 has no move there. Hence the only
non-losing reply is $\zeta^x_{e,\ell}\to x^{\mathrm{out}}_{e,\ell}$.
Then Player~2 is forced to move to $v_x^{\mathrm{return}}$. If a return path
remains, Player~1 chooses one of the edges $v_x^{\mathrm{return}}\to \rho^x_j$,
and the moves through $\tau^x_j$ back to $q_x$ are forced. Since Player~1 makes
the final move back to $q_x$, it is again Player~2's turn there.

Thus every non-losing challenge excursion on an untouched $x$-side consumes
exactly one edge out of $q_x$ and exactly one return path, while any use of
backward edges only consumes additional internal edges on that side.
\end{proof}

\begin{figure}[H]
    \centering
    \includegraphics[width=\textwidth]{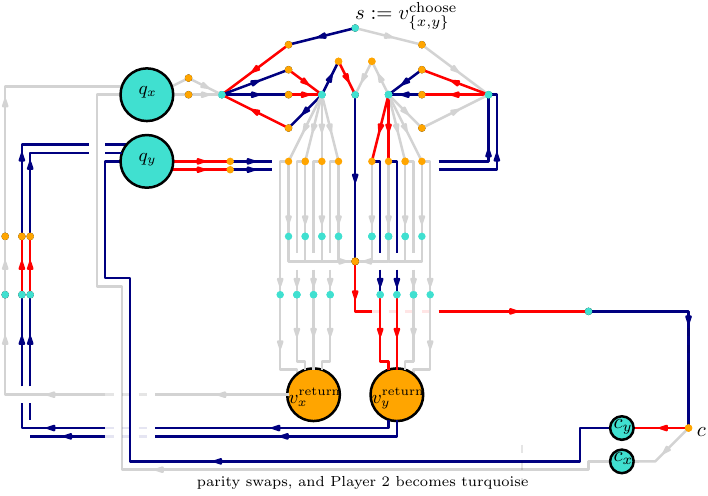}
    \caption{A reduced instance of our input problem for which 
    there is one edge $\{x,y\}$ with weight $2$, and $t(x) = 1$ and 
    $t(y) = 2$. This is a feasible instance and we observe that Player~1 has 
    a winning strategy, with Player~2 to move at 
    $q_y$ with no legal continuation. In this figure we assume not only that the 
    players are playing optimally, but that they want to stay 
    alive ''as long as possible''. The edges played by Player~1 are 
    colored navy and those played by Player~2 are colored red. In particular, 
    observe that in the phase where the edge is oriented, 
    Player~2 only has one legal response every time they move.}
    \label{fig:single_edge_example}
\end{figure}

\begin{lemma}[Yes-instances]
\label{lem:yes}
If $(H,w,t)$ admits a valid orientation $\omega$, then Player~1 has a winning
strategy in the constructed instance of \textsc{Directed Edge Geography}.
\end{lemma}

\begin{proof}
Fix a valid orientation $\omega\in \Ori(H)$. For each edge $e=\{x,y\}$,
Player~1 uses the initial choice at $v_e^{\mathrm{choose}}$ to encode
$\omega(e)$ as follows: if $\omega(e)=(x,y)$, then Player~1 traverses the
$y$-side of $G_e$; if $\omega(e)=(y,x)$, then Player~1 traverses the $x$-side.

\medskip
\noindent
\textbf{Phase 1.}
In each gadget, Player~1 follows the strategy of
Lemma~\ref{lem:phase1-invariant} and never moves to a $\zeta$-vertex. Thus
after each gadget $G_{e_j}$ the token is at $r_{e_j}$ with Player~2 to move,
and all chain moves to the next gadget are forced.

After the last gadget $G_{e_m}$, the moves $r_{e_m}\to m_c\to c$
are forced, with Player~2 to move at $c$.

\medskip
\noindent
\textbf{Phase 2.}
Player~2 chooses a vertex $x\in V(H)$ by moving from $c$ to $c_x$, and
Player~1 is then forced to move to $q_x$.

Let
\[
S_x := \sum_{e\in \theta_\omega(x)} w(e).
\]
Since $\omega$ is valid, we have $S_x\le t(x)$.

Consider the challenge edges out of $q_x$. For an incident edge $e=\{x,y\}$,
the $x$-side of $G_e$ was untouched in Phase~1 if and only if
$\omega(e)=(x,y)$, that is, if and only if $e\in \theta_\omega(x)$. Hence,
by Lemma~\ref{lem:challenge-excursion}:
\begin{itemize}
    \item exactly the challenge edges corresponding to edges in
    $\theta_\omega(x)$ can participate in a non-losing challenge excursion;
    \item there are exactly $S_x$ such challenge edges in total;
    \item every non-losing challenge excursion consumes exactly one such
    challenge edge and exactly one return path;
    \item any use of backward edges by Player~2 only consumes additional
    internal edges and therefore can only reduce the number of future
    non-losing challenge excursions.
\end{itemize}

Since $S_x\le t(x)$, after at most $S_x$ non-losing challenge excursions all
non-losing challenge edges out of $q_x$ are exhausted. At that point it is again
Player~2's turn at $q_x$, and every remaining legal move is immediately losing:
either it enters an already exhausted side, or it loses later in the same
excursion. Therefore Player~2 has no non-losing move, and Player~1 wins.
\end{proof}

\begin{lemma}[No-instances]
\label{lem:no}
If $(H,w,t)$ does not admit a valid orientation, then Player~2 has a winning
strategy in the constructed instance.
\end{lemma}

\begin{proof}
During Phase~1, whenever Player~1 deviates by moving from some primed vertex
$x'_e$ or $y'_e$ to a $\zeta$-vertex, Player~2 responds by moving to the
corresponding escape vertex and thereby leaves the gadget through $r_e$. This
does not change which side of the gadget was initially chosen at
$v_e^{\mathrm{choose}}$, and hence does not change the orientation encoded by
that initial choice. Such a deviation can only help Player~2 in the later
challenge phase, so it suffices to consider the case that Player~1 does not
deviate in Phase~1. Furthermore, if Player~1 does not fully traverse the chosen side of some edge
gadget during Phase~1, then additional forward edges on that side may remain
available. Consequently, if Player~2 later challenges the corresponding vertex,
the set of non-losing challenge excursions available to Player~2 is no smaller
than it would be under the intended traversal. Thus such a deviation can only
help Player~2. Therefore, without loss of generality, we may assume that during
Phase~1 Player~1 completely traverses the chosen side of each gadget, so that
the only relevant choice is the initial side choice, which determines the
encoded orientation $\omega$.

Accordingly, after all edge gadgets have been processed, Player~1 has encoded
some orientation $\omega\in \Ori(H)$. Since the input instance is a no-instance,
there exists a vertex $x$ such that
\[
\sum_{e\in \theta_\omega(x)} w(e) > t(x).
\]
Player~2 chooses this vertex in the challenge phase.

Let
\[
S_x := \sum_{e\in \theta_\omega(x)} w(e).
\]
As in the proof of Lemma~\ref{lem:yes}, exactly the challenge edges
corresponding to edges in $\theta_\omega(x)$ lead to untouched $x$-sides, so
there are exactly $S_x$ challenge edges that can be used in non-losing
challenge excursions.

Player~2 now follows the following strategy: whenever the token is at $q_x$,
choose any unused challenge edge leading to an untouched $x$-side; once the
token reaches the corresponding vertex $x'_e$, choose a $\zeta$-edge
immediately and never use a backward edge. By
Lemma~\ref{lem:challenge-excursion}, each such excursion consumes exactly one
challenge edge out of $q_x$ and exactly one return path, and returns play to
$q_x$ with Player~2 to move, as long as a return path remains.

After $t(x)$ such excursions, all return paths through $v_x^{\mathrm{return}}$
have been exhausted, while at least one unused challenge edge still remains
because $S_x>t(x)$. Player~2 takes such an edge once more. Again
Lemma~\ref{lem:challenge-excursion} applies: after the ensuing move to
$v_x^{\mathrm{return}}$, Player~1 has no available return path and therefore no
legal move. Thus Player~2 wins.
\end{proof}

\begin{theorem}
\textsc{Directed Edge Geography} is XNLP-hard when parameterized by pathwidth.
\end{theorem}

\begin{proof}
We reduce from \textsc{Chosen Maximum Outdegree}, which is XNLP-complete when
parameterized by pathwidth~\cite{bodlaender2022hardtw}.

Given an instance $(H,w,t)$, we construct zthe graph $G$ above.
By Lemmas~\ref{lem:yes} and~\ref{lem:no}, Player~1 has a winning strategy in $G$ if
and only if $(H,w,t)$ is a yes-instance.

Since weights and targets are encoded in unary, the graph $G$ has size polynomial in
the input size.

It remains to argue that the reduction preserves pathwidth.

\noindent
\textbf{Pathwidth preservation.}

\noindent 
Let $(B_1,\dots,B_r)$ be a nice path decomposition of $H$ of width
$k=\pw(H)$.
We may assume that $B_1=B_r=\emptyset$, and that for each $i\in[r-1]$,
the bags $B_i$ and $B_{i+1}$ differ in exactly one vertex.

For each edge $e=\{x,y\}\in E(H)$, let $a(e)$ be the smallest index $i$
such that $\{x,y\}\subseteq B_i$.
Order the edges as $e_1,\dots,e_m$ so that $a(e_1)\le a(e_2)\le \cdots \le a(e_m)$,
breaking ties arbitrarily.

We now construct a path decomposition of the reduction graph $G$.

\medskip
\noindent
\textbf{Anchor vertices.}
For each $x\in V(H)$, recall that the reduction contains the three challenge
vertices $c_x,\ q_x,\ v_x^{\mathrm{return}}$,
and also the global challenge vertex $c$.
For each bag $B_i$, define the anchor set
\[
S_i := \{c\}\cup
\bigcup_{x\in B_i}\{c_x,q_x,v_x^{\mathrm{return}}\}.
\]
Since $|B_i|\le k+1$, we have
$|S_i|\le 1+3(k+1)=3k+4$.

Intuitively, while we are at stage $i$, all vertices in $S_i$ remain present in
every bag.

\medskip
\noindent
\textbf{Edge-gadget blocks.}
Fix an edge $e=\{x,y\}$ with $a(e)=i$.
During stage $i$, we realize the entire edge gadget $G_e$ by a block of bags
in which the anchor set $S_i$ is present throughout.

Let
\[
K_e :=
\Bigl\{
v_e^{\mathrm{choose}},
v_e^{\mathrm{choose},x},
v_e^{\mathrm{choose},y},
x_e,y_e,x'_e,y'_e,
m_e^x,m_e^y,o_e,r_e
\Bigr\},
\]
where $m_e^x,m_e^y$ are the two subdivision vertices on the arcs from
$x'_e$ and $y'_e$ to $o_e$.
Note that $|K_e|=11$.

We keep all vertices of $K_e$ in every bag of the block for $e$.
Since $x,y\in B_i$, the anchor set $S_i$ contains
$c_x,q_x,v_x^{\mathrm{return}}$, and
$c_y,q_y,v_y^{\mathrm{return}}$.
All remaining vertices of the gadget are now introduced and forgotten one at a time.

More precisely, in the block for $e$ we append bags that, in addition to
$S_i\cup K_e$, contain:
\begin{itemize}
    \item for each $j\in[w(e)]$, one bag with $\mu^x_{e,j}$ and one bag with
    $\mu^y_{e,j}$;
    \item for each $j\in[w(e)]$, one bag with $\xi^x_{e,j}$ and one bag with
    $\xi^y_{e,j}$;
    \item for each $j\in[w(e)-1]$, one bag with $\delta^x_{e,j}$ and one bag with
    $\delta^y_{e,j}$;
    \item for each $j\in[w(e)+1]$, two bags for the $x$-side escape structure,
    namely one containing $\zeta^x_{e,j}$ and $x^{\mathrm{out}}_{e,j}$, and one containing
    $\zeta^x_{e,j}$ and $x^{\mathrm{esc}}_{e,j}$;
    \item symmetrically, for each $j\in[w(e)+1]$, two bags for the $y$-side escape
    structure.
\end{itemize}

Every edge of the subgraph involving $G_e$ is covered in one of these bags:
\begin{itemize}
    \item edges internal to the constant-size core are covered since $K_e$ is present
    throughout the block;
    \item each edge $q_x\to \mu^x_{e,j}$ and $\mu^x_{e,j}\to x_e$ is covered in the bag
    containing $\mu^x_{e,j}$, since $q_x\in S_i$ and $x_e\in K_e$;
    \item each edge $x_e\to \xi^x_{e,j}$ and $\xi^x_{e,j}\to x'_e$ is covered in the bag
    containing $\xi^x_{e,j}$;
    \item each edge $x'_e\to \delta^x_{e,j}$ and $\delta^x_{e,j}\to x_e$ is covered in the bag
    containing $\delta^x_{e,j}$;
    \item each edge $x'_e\to \zeta^x_{e,j}$, $\zeta^x_{e,j}\to x^{\mathrm{out}}_{e,j}$,
    and $x^{\mathrm{out}}_{e,j}\to v_x^{\mathrm{return}}$ is covered in the bag containing
    $\zeta^x_{e,j}$ and $x^{\mathrm{out}}_{e,j}$, since
    $v_x^{\mathrm{return}}\in S_i$;
    \item each edge $x'_e\to \zeta^x_{e,j}$, $\zeta^x_{e,j}\to x^{\mathrm{esc}}_{e,j}$,
    and $x^{\mathrm{esc}}_{e,j}\to r_e$ is covered in the bag containing
    $\zeta^x_{e,j}$ and $x^{\mathrm{esc}}_{e,j}$;
    \item and symmetrically on the $y$-side.
\end{itemize}

Thus the entire gadget for $e$ is realized with bags of size at most
\[
|S_i|+|K_e|+2 \le (3k+4)+11+2 = 3k+17.
\]

\medskip
\noindent
\textbf{Return-path blocks.}
Suppose that $x$ is forgotten when passing from $B_i$ to $B_{i+1}$.
Before removing the three anchor vertices
$c_x,q_x,v_x^{\mathrm{return}}$, we realize all $t(x)$ return paths
\[
v_x^{\mathrm{return}}\to \rho^x_j \to \tau^x_j \to q_x
\qquad (j\in[t(x)])
\]
one at a time, while keeping all of $S_i$ present.

For each $j\in[t(x)]$, we append two bags:
$S_i\cup\{\rho^x_j\}$, and
$S_i\cup\{\rho^x_j,\tau^x_j\}$.

These cover all return-path edges incident with $\rho^x_j$ and $\tau^x_j$.
After all such blocks have been appended, we forget
$c_x,q_x,v_x^{\mathrm{return}}$.

Hence the return paths of $x$ are realized with bags of size at most $|S_i|+2 \le 3k+6$.

\medskip
\noindent
\textbf{Transitions between consecutive edge gadgets.}
The edge gadgets are processed in the order $e_1,\dots,e_m$.
Suppose $e_j$ is followed by $e_{j+1}$.
After finishing the block for $e_j$, we keep the vertex $r_{e_j}$ present until
we have reached the stage $a(e_{j+1})$.
At that point, we append one transition bag containing
$S_{a(e_{j+1})}\cup \{r_{e_j}, v_{e_{j+1}}^{\mathrm{choose}}\}$.
This covers the chain edge
$r_{e_j}\to v_{e_{j+1}}^{\mathrm{choose}}$.
Since only one extra vertex $r_{e_j}$ is carried between stages, the bag size
remains at most $|S_i|+2 \le 3k+6$
during these transitions.

Similarly, after the last edge gadget $e_m$ has been processed, we keep $r_{e_m}$
present until the final challenge edge
$r_{e_m}\to m_c\to c$
is realized using bags
$\{c,r_{e_m},m_c\}$,
and $\{c,m_c\}$.

\medskip
\noindent
\textbf{Contiguity.}
Every vertex of the reduction belongs to a contiguous set of bags:
\begin{itemize}
    \item the global vertex $c$ appears in every bag;
    \item for each $x\in V(H)$, the three anchor vertices
    $c_x,q_x,v_x^{\mathrm{return}}$ appear exactly from the introduction of $x$
    in the nice path decomposition until the moment $x$ is forgotten and its return
    paths have been processed, hence on a contiguous interval;
    \item every core vertex of an edge gadget appears only inside the corresponding
    edge-gadget block;
    \item every auxiliary vertex
    $\mu,\xi,\delta,\zeta,x^{\mathrm{out}},x^{\mathrm{esc}},y^{\mathrm{out}},y^{\mathrm{esc}},
    \rho,\tau,m_c$
    appears only in the one or two bags where it is used.
\end{itemize}

Thus we obtain a valid path decomposition of $G$.

\medskip
\noindent
\textbf{Width bound.}
Every bag has size at most $3k+17$.
Therefore $\pw(G)\le 3k+16$.

Hence the reduction preserves pathwidth up to a function of $k$.

The reduction is computable by a deterministic parameterized logspace transducer. The output graph is generated by enumerating the edge gadgets in the order induced by the input path decomposition. For an input edge \(e={x,y}\), the first bag \(a(e)\) containing both endpoints can be found by scanning the input decomposition using only counters and the identifiers of \(x\) and \(y\). The gadgets can then be output by nested loops over the input edges and over indices \(j\leq w(e)+O(1)\), and the challenge and return structures by analogous loops over vertices and indices \(j\leq t(v)+O(1)\). Since all integer weights and bounds are encoded in unary, these ranges are polynomially bounded in the input size, and all counters use \(O(\log n)\) bits. Each output vertex and edge has a name determined by a constant-size type label together with an input vertex or edge and one such counter, so the construction can be streamed without storing the output graph. The constructed path-decomposition witnessing the parameter bound can be generated in the same stage-by-stage fashion. Hence the reduction satisfies the logspace requirement for XNLP-hardness reductions, and the output pathwidth is bounded by a function of the input pathwidth.

Summing up, we have shown that \textsc{Chosen Maximum Outdegree} is 
parameterized logspace reducible to \textsc{Directed Edge Geography}, preserving 
pathwidth up to a function of $k$. Thus \textsc{Directed Edge Geography} is 
XNLP-hard with respect to pathwidth $k$.
\end{proof}

\begin{figure}[H]
    \centering
    \includegraphics[width=\textwidth]{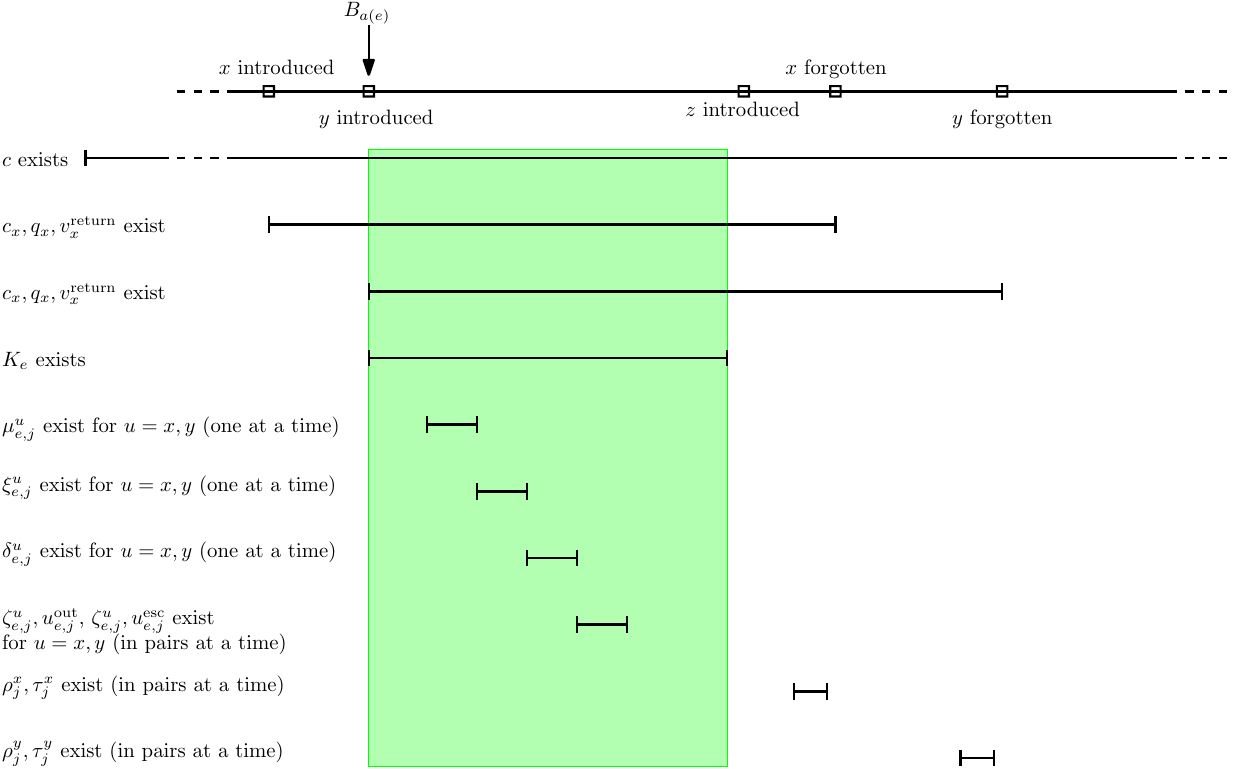}
    \caption{Schematic view of the pathwidth-preserving
    construction for an edge gadget $e=\{x,y\}$. The top timeline indicates
    the intervals during which $x$ and $y$ are present. The rows below indicate 
    the corresponding lifetimes of the anchor vertices, together with the global 
    vertex $c$. The green block indicates the dedicated contiguous block of bags 
    in the constructed decomposition where $e$ is realized: throughout 
    $K_e$, which is meant to represent the constant-size core of the edge gadget, 
    is present together with $S_i$, while the auxiliary vertices
    are introduced 
    only locally, one at a time or in pairs. Finally, before a vertex is forgotten, 
    the path-return vertices are realized in pairs at a time.}
    \label{fig:pathwidth-decomposition}
\end{figure}

\subsection{Example}

Let $H=(\{x,y,z\}, \{\{x,y\},\{y,z\}\})$. Let the weights be
$w(\{x,y\})=2,\ w(\{y,z\})=2$,
and lastly let 
$t(x)=1,\ t(y)=2,\ t(z)=1$. This is the instance 
shown in Figure~\ref{fig:orientation}.

The corresponding \textsc{Directed Edge Geography} instance is shown in
Figure~\ref{fig:example-instance}.
In this example, Player~2 has a winning strategy.

\begin{figure}[H]
    \centering
    \includegraphics[width=\textwidth]{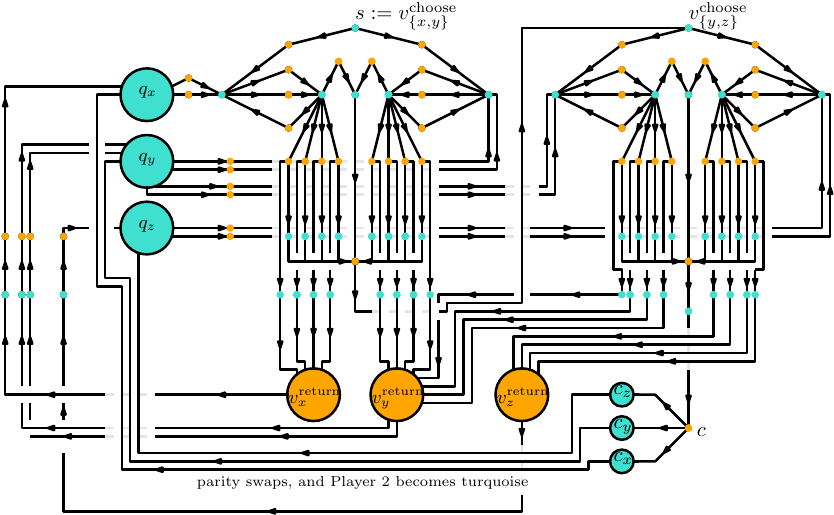}
    \caption{Example instance.}
    \label{fig:example-instance}
\end{figure}

\newpage 
\section{Full Technical Development of the XP Algorithm}
\label{app:sec:algo}

In this section we show that \textsc{Undirected Edge Geography} is in XP on
simple graphs when the input is given together with a rooted tree partition of
bounded width. We also show, via the typical reduction, that 
\textsc{Directed Edge Geography} is in XP.

Our approach is to summarize each child subtree by a finite \emph{interface
type} describing how play can enter the subtree from its parent bag, possibly
return to the parent bag, and leave behind a smaller residual subinstance.
For fixed width $k$, only finitely many such interface types can occur.
Moreover, from the perspective of the parent bag, children of the same type are
interchangeable. This allows us to compress the state of a bag to the
multiplicities of its child types, which yields an XP dynamic program.

The formal development uses three layers of induction.

\begin{enumerate}
    \item The tree partition is processed bottom-up, so when a bag is processed,
    all types arising in the already processed child subtrees are available.
    \item Inside a fixed bag, residual configurations are handled by induction on
    a measure that counts remaining parent ports, remaining internal edges, and
    the total rank of the residual child interfaces.
    \item For a fixed residual configuration and a fixed entry port, the
    parent-side local game is acyclic, and the values of its local states are
    defined by reverse induction on that acyclic game graph.
\end{enumerate}

\noindent 
\textbf{Roadmap.} 

\noindent
The technical work in this section is mainly needed to show that these
interface types are well defined and that one residual child subinstance may be
replaced by another of the same type without changing the parent-side outcome.
Readers may find it helpful to keep the worked example in the final subsection
in mind while reading the definitions and lemmas below.

We now make these notions precise.

\subsection{Tree partitions and interfaces}

We use a rooted version of the standard notion of tree-partitions; see 
\cite{wood2009treepartition} for tree-partition-width. Since we work with simple 
unweighted graphs, we parameterize only 
by the maximum bag size. In this setting, the size of 
each parent cut is automatically at most 
$k^2$.

\begin{figure}[ht]
    \centering
    \includegraphics[width=\textwidth]{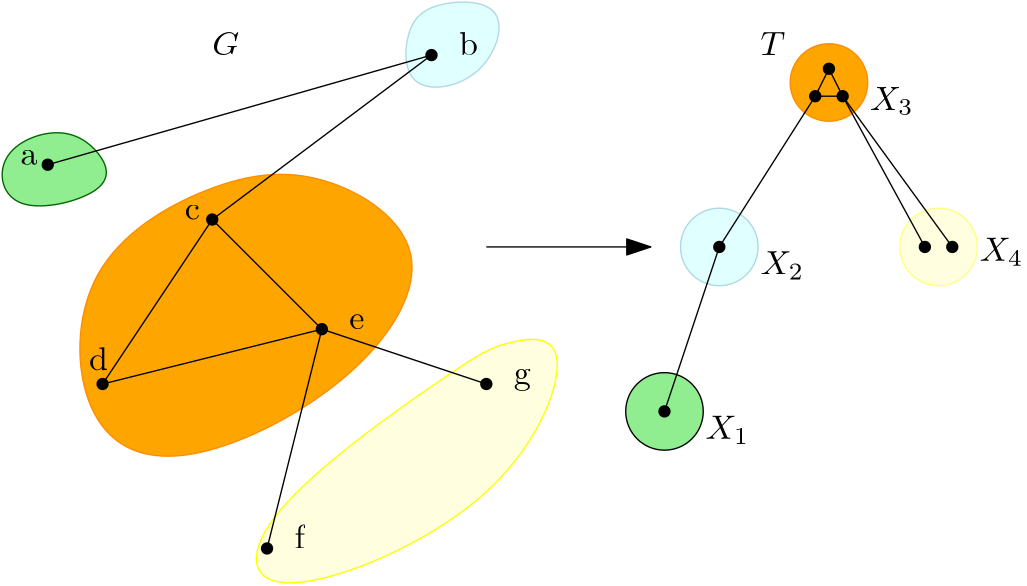}
    \caption{A tree partition of a graph $G$ with width $k = 3$.}
    \label{fig:tree-partition}
\end{figure}

\begin{definition}[Rooted tree partition]
A \emph{rooted tree partition} of a graph $G$ is a pair $(\{X_i \mid i \in I\},T)$,
where $\{X_i \mid i \in I\}$ is a partition of $V(G)$ and $T$ is a rooted tree
with node set $I$, such that every edge of $G$ is either contained in a bag
$X_i$ or has its endpoints in two bags whose nodes are adjacent in $T$.
The \emph{width} of the tree partition is
$\max_{i\in I}|X_i|$.
\end{definition}

Throughout this section, we assume $G$ is a (simple and finite) graph given with a tree partition of width 
$k$. 

\begin{definition}[Parent cut]
Let $i$ be a non-root node of the rooted tree partition of $G$, and let $\operatorname{par}(i)$
denote its parent. The \emph{parent cut} of $i$ is the set $C_i := E(X_{\operatorname{par}(i)},X_i)$.
Since $G$ is simple and $|X_{\operatorname{par}(i)}|,|X_i|\le k$, we have
$|C_i|\le k^2$.
\end{definition}

\begin{definition}[Bag labelings]
For each node $i$, fix once and for all an injective map $\iota_i : X_i \hookrightarrow [k]$.
\end{definition}

\begin{definition}[Port set and interface labels]
Let $i$ be a non-root node. A \emph{port} of $i$ is an edge of the parent cut
$C_i$. We identify the set of ports of $i$ with a subset $P_i \subseteq \Pi_k := [k^2]$.

For each port $p\in P_i$, let the corresponding edge of $C_i$ join a vertex
$x_p\in X_{\operatorname{par}(i)}$ to a vertex $y_p\in X_i$.

The \emph{parent-side label} of $p$ is defined as $\lambda_i^\uparrow(p):=\iota_{\operatorname{par}(i)}(x_p)$,
and the \emph{child-side label} of $p$ is defined as $\lambda_i^\downarrow(p):=\iota_i(y_p)$.

Thus each non-root node $i$ determines an \emph{interface} $(P_i,\lambda_i^\uparrow,\lambda_i^\downarrow)$.
\end{definition}

\subsection{Ambient interface types}

We first define a finite ambient universe of possible interface behaviors. Later
we restrict attention to those ambient types that are actually realizable by
subtrees of the given instance.

\begin{definition}[Labeled interface]
A \emph{labeled interface} is a pair $(P,\lambda)$ where
$P \subseteq \Pi_k$ and $\lambda:P\to [k]$.
Its \emph{rank} is $|P|$.
\end{definition}

\begin{definition}[Ambient type universe]
For each labeled interface $(P,\lambda)$ we define recursively a finite set
$\widetilde{\mathcal T}(P,\lambda)$
of \emph{ambient types}.

\begin{enumerate}
    \item If $P=\emptyset$, then $\widetilde{\mathcal T}(P,\lambda)=\{\bot\}$.

    \item Suppose $|P|=r>0$, and that all sets
    $\widetilde{\mathcal T}(P',\lambda')$ with $|P'|<r$
    have already been defined.
    Fix an entry port $a\in P$. Define $\Omega(P,\lambda,a)$
    to be the set of all triples $(q,\tau,b)$,
    where
    \begin{itemize}
        \item $q\in P\setminus\{a\}$,
        \item $\tau\in
        \widetilde{\mathcal T}(P\setminus\{a,q\},\lambda|_{P\setminus\{a,q\}})$,
        \item $b\in\{\mathsf{same},\mathsf{opp}\}$.
    \end{itemize}

    An \emph{ambient type} $\sigma\in \widetilde{\mathcal T}(P,\lambda)$ is a
    family $\sigma=\bigl((A_a^\sigma,\Phi_a^\sigma)\bigr)_{a\in P}$,
    where for every $a\in P$:
    \begin{itemize}
        \item $A_a^\sigma \subseteq \Omega(P,\lambda,a)$
        is the set of exit labels declared attainable from entry through $a$,
        \item $\Phi_a^\sigma:\{0,1\}^{A_a^\sigma}\to \{0,1\}$
        is a Boolean function.
    \end{itemize}
\end{enumerate}
\end{definition}

\begin{definition}[Rank of an ambient type]
Let \(\sigma \in \widetilde{\mathcal T}(P,\lambda)\).  
The \emph{rank} of \(\sigma\), denoted \(\operatorname{rk}(\sigma)\), is
$\operatorname{rk}(\sigma):=|P|$.
Equivalently, \(\operatorname{rk}(\sigma)\) is the rank of the underlying labeled
interface of \(\sigma\).
\end{definition}

\begin{lemma}
\label{lem:finitely-many-types}
For every fixed \(k\), the ambient type universe
\[
\widetilde{\mathcal T}_k
:=
\biguplus_{P\subseteq \Pi_k}\ \biguplus_{\lambda:P\to [k]}
\widetilde{\mathcal T}(P,\lambda)
\]
is finite.
\end{lemma}

\begin{proof}
Since \(\Pi_k=[k^2]\) is finite, there are only finitely many subsets
\(P\subseteq \Pi_k\), and for each such \(P\) there are only finitely many
functions \(\lambda:P\to [k]\). Hence there are only finitely many labeled
interfaces \((P,\lambda)\) of rank at most \(k^2\).

We now prove by induction on \(|P|\) that each
\(\widetilde{\mathcal T}(P,\lambda)\) is finite.

If \(P=\emptyset\), then $\widetilde{\mathcal T}(P,\lambda)=\{\bot\}$.

Assume \(|P|=r>0\) and that all smaller-rank ambient type sets are finite. Fix
\(a\in P\). Since every lower-rank ambient type set is finite, the set
$\Omega(P,\lambda,a)$
is finite. Hence there are only finitely many choices for a subset
$A_a^\sigma \subseteq \Omega(P,\lambda,a)$,
and for each such choice there are only finitely many Boolean functions
$\Phi_a^\sigma : \{0,1\}^{A_a^\sigma}\to \{0,1\}$.
Since \(P\) itself is finite, there are only finitely many possible families
$\bigl((A_a^\sigma,\Phi_a^\sigma)\bigr)_{a\in P}$.

Thus \(\widetilde{\mathcal T}(P,\lambda)\) is finite.

Since there are finitely many labeled interfaces $(P,\lambda)$, and for each of them 
$\widetilde{\mathcal T}(P,\lambda)$ is finite, the disjoint union 
\[
\widetilde{\mathcal T}_k
:=
\biguplus_{P\subseteq \Pi_k}\ \biguplus_{\lambda:P\to [k]}
\widetilde{\mathcal T}(P,\lambda)
\]
is finite.
\end{proof}

For fixed $k$, every ambient type has a canonical finite encoding: the
underlying interface $(P,\lambda)$ is part of the data, each set
$A_a^\sigma$ is a subset of the finite universe $\Omega(P,\lambda,a)$, and
each Boolean function $\Phi_a^\sigma$ can be represented by its truth table.

For every ambient type \(\sigma\), we regard its underlying labeled interface
$(P_\sigma,\lambda_\sigma)$
as part of the data of \(\sigma\).

\noindent 
\textbf{Intuition.}

\noindent 
Fix a bag $X_i$, and think of each child subtree of $i$ as a black-box (or 
gadget) game attached to $X_i$ through the parent cut. From the 
perspective of the parent bag, the internal graph of that subtree is irrelevant 
except for how play can interact with this boundary. If play enters 
the child through some port $a$, then the parent only needs to know three 
things about what may happen next: through which port $q$ play 
can later return, what smaller subtree remains after that excursion, and 
whether the return happens with the same player to move. These 
are exactly the data recorded by an exit label $(q,\tau,b).$ However, 
the set of attainable exit labels by itself is not enough: the value 
of entering the child also depends on which of those possible returns 
lead to a winning parent-side continuation. The role of the boolean function 
$\Phi_a$ is to summarize this dependence. It takes as input the truth 
values of the attainable return options and tells us whether the entrant 
through $a$ can force a win inside the child. In this way, an interface 
type packages a child subtree into a finite boundary behavior, and for 
fixed width there are only finitely many such behaviors to consider. 

\subsection{Expanded residual configurations}

Fix a node $i$ of the rooted tree partition. Since the algorithm proceeds
bottom-up, the types of all residual child subinstances of $i$ are assumed
already defined.

\begin{definition}[Residual child subinstance]
A \emph{residual child subinstance} of $i$ consists of:
\begin{itemize}
    \item a residual game position arising inside the subtree of some child of
    $i$,
    \item its parent-facing labeled interface $(P_H,\lambda_H)$,
    where \(\lambda_H:P_H\to [k]\) is expressed relative to the labeling
    \(\iota_i\) of the parent bag \(X_i\),
    \item and its already defined ambient type
    $\tau(H)\in \widetilde{\mathcal T}(P_H,\lambda_H)$.
\end{itemize}
\end{definition}

\begin{definition}[Expanded residual configuration]
An \emph{expanded residual configuration} at node $i$ is a triple
$\widehat{\Gamma}=(U,F,\mathcal C)$,
where
\begin{itemize}
    \item $U\subseteq P_i$ is the set of unused parent ports,
    \item $F\subseteq E(G[X_i])$ is the set of unused edges inside the bag $X_i$,
    \item $\mathcal C$ is a multiset of residual child subinstances of $i$.
\end{itemize}
Its \emph{parent-facing labeled interface} is
$(U,\lambda_i^\uparrow|_U)$.
\end{definition}

\begin{definition}[Measure of an expanded residual configuration]
Let $\widehat{\Gamma}=(U,F,\mathcal C)$
be an expanded residual configuration. Its \emph{measure} is
\[
\mu(\widehat{\Gamma})
:=
|U|+|F|+\sum_{H\in\mathcal C}\operatorname{rk}(\tau(H)).
\]
\end{definition}

\subsection{Expanded local states}

For $\pi\in\{\mathsf{same},\mathsf{opp}\}$, let $\overline{\pi}$ denote the
other parity, and define
$\pi\circ\mathsf{same}:=\pi$,
and $\pi\circ\mathsf{opp}:=\overline{\pi}$.

\begin{definition}[Entry state]
Let $\widehat{\Gamma}=(U,F,\mathcal C)$
be an expanded residual configuration at node $i$, and let $a\in U$.
Let $z_a\in X_i$ be the unique vertex satisfying
$\iota_i(z_a)=\lambda_i^\downarrow(a)$.

The \emph{entry state} for entry through $a$ is
$S_a^0 := (z_a,\; U\setminus\{a\},\; F,\; \mathcal C,\; \mathsf{opp})$.
\end{definition}

\begin{definition}[Expanded local state]
Fix an expanded residual configuration $\widehat{\Gamma}=(U,F,\mathcal C)$
and an entry port $a\in U$.
An \emph{expanded local state} relative to $a$ is a tuple
$S=(z,U',F',\mathcal C',\pi)$,
where
\begin{itemize}
    \item $z\in X_i$ is the current token position,
    \item $U'\subseteq U\setminus\{a\}$ is the set of currently unused parent
    ports,
    \item $F'\subseteq F$ is the set of currently unused edges inside $X_i$,
    \item $\mathcal C'$ is a multiset of residual child subinstances,
    \item $\pi\in\{\mathsf{same},\mathsf{opp}\}$ records whether the player to
    move is the same as or opposite to the player who originally entered through
    $a$.
\end{itemize}
\end{definition}

\begin{figure}[H]
    \centering
    \includegraphics[width=\textwidth]{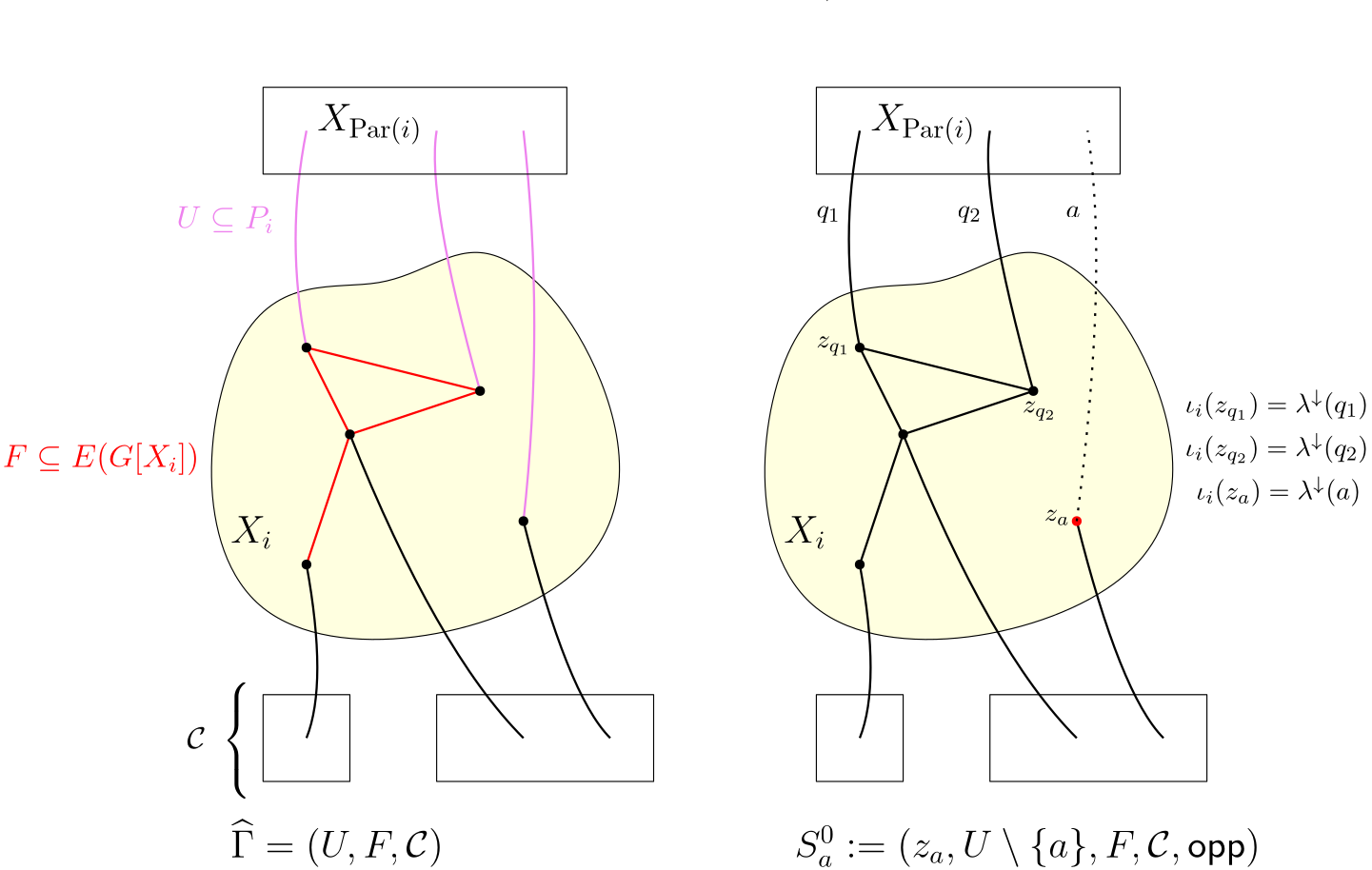}
    \caption{A residual configuration and entry state at a bag $X_i$. 
    Left: an expanded residual configuration $\widehat\Gamma = (U,F,\mathcal C)$ consisting of unused 
    parent ports $U$, unused internal edges $F$, and multiset of residual child subinstances $\mathcal C$. 
    Right: after entering through $a\in U$, the local game starts in the entry state 
    $S_a^0 = (z_a, U\setminus \{a\},F,\mathcal C, \mathsf{opp})$, where $z_a \in X_i$ is the child-side endpoint of $a$. 
    From this entry state, play may only proceed by a child-side excursion.}
    \label{fig:configuration_and_states}
\end{figure}

\begin{definition}[Legal moves from an expanded local state]
Let $S=(z,U',F',\mathcal C',\pi)$
be an expanded local state.

A legal move from $S$ is one of the following.

\begin{enumerate}
    \item \emph{Internal move.}
    If $\{z,z'\}\in F'$, then traversing this edge yields the successor state
    $(z',U',F'\setminus\{\{z,z'\}\},\mathcal C',\overline{\pi})$.

    \item \emph{Exit move.}
    If $q\in U'$ satisfies
    $\lambda_i^\downarrow(q)=\iota_i(z)$,
    then play may leave the current subtree through the port $q$.

    \item \emph{Child excursion.}
    If $H\in\mathcal C'$ has an entry port $b\in P_H$ satisfying
    $\lambda_H(b)=\iota_i(z)$,
    then play may enter the residual child subinstance $H$ through $b$.
\end{enumerate}
\end{definition}

\begin{definition}[Measure of an expanded local state]
The \emph{measure} of an expanded local state
$S=(z,U',F',\mathcal C',\pi)$
is
\[
\mu(S)
:=
|U'|+|F'|+\sum_{H\in\mathcal C'}\operatorname{rk}(\tau(H)).
\]
\end{definition}

\begin{remark*}
Every expanded local state reachable from the entry state $S_a^0$ satisfies
$U'=U\setminus\{a\}$,
because internal moves and child excursions do not change the set of unused
parent ports of the current node. We nevertheless keep $U'$ in the notation,
since it is useful when describing exit moves and residual configurations after
exiting.
\end{remark*}

\subsection{Returning child excursions}

\begin{definition}[Parent-side successor induced by a returning child excursion]
Let $S=(z,U',F',\mathcal C',\pi)$
be an expanded local state, and let $H\in\mathcal C'$ be a residual child
subinstance entered through a port $b\in P_H$ with
$\lambda_H(b)=\iota_i(z)$.

Suppose a legal excursion inside \(H\), entered through \(b\), returns to the
parent bag with exit label $(q,\tau',c)\in A_b^{\tau(H)}$.
Let \(H'\) denote the residual child subinstance produced by that excursion, so
that $\tau(H')=\tau'$.
Let $z_q\in X_i$ be the unique vertex satisfying
$\iota_i(z_q)=\lambda_H(q)$.
The corresponding \emph{parent-side successor state} is
\[
S[q,H',c]
:=
(z_q,\; U',\; F',\; (\mathcal C'\setminus\{H\})\uplus\{H'\},\;
\pi\circ c).
\]
\end{definition}

\subsection{Local semantics and corresponding states}

We now explain how the parent-side local game is evaluated.

Fix a node \(i\), and let
$\widehat{\Gamma}=(U,F,\mathcal C)$
be an expanded residual configuration at \(i\). Let \(a\in U\) be an entry
port. We evaluate the local game relative to the player who entered the current
subtree through \(a\).

\noindent 
\textbf{Viewpoint convention.}

\noindent 
All local state values are interpreted relative to the player who originally
entered the current subtree through \(a\).

If a legal move is a child excursion into a residual child subinstance \(H\)
through an entry port \(b\), then the player who enters \(H\) is the player to
move in the current local state. Hence the Boolean function $\Phi_b^{\tau(H)}$
must be evaluated from the viewpoint of that child entrant. Therefore, when the
current parity is \(\mathsf{same}\), the child entrant is the original entrant
through \(a\), whereas when the current parity is \(\mathsf{opp}\), the child
entrant is the opponent of the original entrant through \(a\). In the latter
case one must complement outcomes when passing between parent-side state values
and the valuation supplied to \(\Phi_b^{\tau(H)}\).

\noindent 
\textbf{Inductive setup.}

\noindent
The definitions below are used inside an induction on the measure
\(\mu(\widehat{\Gamma})\). Accordingly, whenever an exit move produces a smaller
expanded residual configuration, we may refer to its already defined ambient
type.

\begin{lemma}[Acyclicity of the parent-side local game]
\label{lem:local-acyclic}
Fix an expanded residual configuration
$\widehat{\Gamma}=(U,F,\mathcal C)$
and an entry port \(a\in U\). Consider the directed graph whose vertices are the
expanded local states reachable from the entry state \(S_a^0\), and whose arcs
are the successor states arising from internal moves and from returning child
excursions. Then this graph is acyclic.
\end{lemma}

\begin{proof}
Every such arc strictly decreases the local-state measure.

For an internal move, one edge is removed from \(F'\), so the measure decreases
by \(1\).

For a returning child excursion through a residual child subinstance \(H\), the
resulting residual child subinstance \(H'\) has type \(\tau(H')=\tau'\), where
the corresponding exit label has the form $(q,\tau',c)\in A_b^{\tau(H)}$
for some entry port \(b\in P_H\). Since \(\tau(H)\) is an ambient type on the
interface \((P_H,\lambda_H)\), every such exit label satisfies
$\tau' \in
\widetilde{\mathcal T}(P_H\setminus\{b,q\},\lambda_H|_{P_H\setminus\{b,q\}})$.
Hence $\operatorname{rk}(\tau')=\operatorname{rk}(\tau(H))-2$.
Thus replacing \(H\) by \(H'\) decreases the measure by \(2\).

Therefore every arc strictly decreases the measure, and no directed cycle can
exist.
\end{proof}

\begin{definition}[Exit labels]
\label{def:exit-labels}
Let $S=(z,U',F',\mathcal C',\pi)$
be an expanded local state reachable from \(S_a^0\), and suppose that \(q\in U'\)
satisfies $\lambda_i^\downarrow(q)=\iota_i(z)$.
Then exiting through \(q\) produces the smaller expanded residual configuration
$\widehat{\Gamma}_S^q:=(U'\setminus\{q\},F',\mathcal C')$.

The corresponding \emph{exit label} is
$\ell_S(q):=(q,\tau(\widehat{\Gamma}_S^q),\overline{\pi})$.
\end{definition}

\begin{definition}[Attainable exit labels]
\label{def:attainable-exit-labels}
Let $\widehat{\Gamma}=(U,F,\mathcal C)$
be an expanded residual configuration and let \(a\in U\). The set
$A_a^{\widehat{\Gamma}}$
of \emph{attainable exit labels} is the set of all exit labels \(\ell_S(q)\)
that arise from some legal play of the parent-side local game started at the
entry state \(S_a^0\).
\end{definition}

\begin{definition}[Corresponding expanded local states]
\label{def:corresponding-states}
Let
$
\widehat{\Gamma}_1=(U,F,\mathcal C_1),\
\widehat{\Gamma}_2=(U,F,\mathcal C_2)
$
be two expanded residual configurations that differ only in that one residual
child subinstance \(H_1\in\mathcal C_1\) is replaced by another residual child
subinstance \(H_2\in\mathcal C_2\).

Fix an entry port \(a\in U\). Two expanded local states
\[
S_1=(z,U',F',\mathcal C_1',\pi),\qquad
S_2=(z,U',F',\mathcal C_2',\pi)
\]
reachable from the corresponding entry states are called \emph{corresponding} if
\(\mathcal C_1'\) and \(\mathcal C_2'\) differ only in that \(H_1\) is replaced
by \(H_2\), or both \(H_1\) and \(H_2\) are absent.
\end{definition}

\subsection{Simultaneous well-definedness and replacement}

\begin{definition}[Ambient valuation]
    Let \(P,\lambda\) be a labeled interface, and let $a \in P$. 
    An \emph{ambient valuation} for entry through $a$ is a function 
    $\eta : \Omega(P, \lambda, a) \to \{0,1\}$.
    
    If $A \subseteq \Omega(P,\lambda, a)$, we write $\eta\mid_{A}$ for its restriction to $A$.
    
\end{definition}

\noindent 
\textbf{Proof idea.}

\noindent
The only potential circularity in defining the type of a residual configuration
comes from child excursions: to evaluate such a move, we consult the Boolean
function stored in the child type, but that Boolean function is itself indexed
by possible return labels, each of which leads to a parent-side successor state.
The key point is that this dependence is well founded. Internal moves strictly
decrease the number of unused bag edges, and returning child excursions strictly
decrease the total rank contribution of the child interfaces. Hence, for a fixed
residual configuration and a fixed entry port, the reachable parent-side local
game is acyclic and can be evaluated by reverse induction. The lemma below
proves simultaneously that this evaluation is well defined and that replacing a
residual child subinstance by another one of the same type does not change the
parent-side outcome.

\begin{lemma}[Simultaneous well-definedness and replacement]
\label{lem:joint}
Let \(i\) be a node of the rooted tree partition.

\begin{enumerate}
    \item For every expanded residual configuration
    $\widehat{\Gamma}=(U,F,\mathcal C)$
    at node \(i\), there exists a uniquely defined ambient type
    $\tau(\widehat{\Gamma})
    \in
    \widetilde{\mathcal T}(U,\lambda_i^\uparrow|_U)$.

    \item Let
    $\widehat{\Gamma}_1=(U,F,\mathcal C_1),\qquad
    \widehat{\Gamma}_2=(U,F,\mathcal C_2)$
    be expanded residual configurations that differ only in that one residual
    child subinstance \(H_1\in\mathcal C_1\) is replaced by another residual
    child subinstance \(H_2\in\mathcal C_2\) with
    $\tau(H_1)=\tau(H_2)$.
    
    Fix \(a\in U\), and let
    $\eta:\Omega(U,\lambda_i^\uparrow|_U,a)\to\{0,1\}$
    be an ambient valuation. 
    
    Write
    $\eta_t:=\eta|_{A_a^{\widehat{\Gamma}_t}}, \ (t=1,2)$.
    Then for every pair of corresponding reachable expanded local states
    $S_1,S_2$,
    we have
    \[
    \operatorname{Win}_{\widehat{\Gamma}_1,a,\eta_1}(S_1)
    =
    \operatorname{Win}_{\widehat{\Gamma}_2,a,\eta_2}(S_2).
    \]
\end{enumerate}
\end{lemma}

\begin{proof}
We prove the following stronger statement by induction on the measure
\[
\mu(\widehat{\Gamma})
=
|U|+|F|+\sum_{H\in\mathcal C}\operatorname{rk}(\tau(H)).
\]

\smallskip
\noindent
\emph{Induction claim.}
For every expanded residual configuration
$\widehat{\Gamma}=(U,F,\mathcal C)$
at node \(i\):
\begin{enumerate}
    \item[(a)] the ambient type \(\tau(\widehat{\Gamma})\) is uniquely defined;
    \item[(b)] whenever two expanded residual configurations of smaller or equal
    measure differ only by replacing one residual child subinstance by another
    of the same type, then:
    \begin{itemize}
        \item their ambient types are equal, and
        \item corresponding reachable expanded local states have the same value
        under every common ambient valuation.
    \end{itemize}
\end{enumerate}
The lemma follows immediately from this stronger claim.

\smallskip
\noindent
\textbf{Base case.}
If $\mu(\widehat{\Gamma})=0$,
then \(U=\emptyset\), \(F=\emptyset\), and \(\mathcal C=\emptyset\). Hence there
is no entry port, so the only possible ambient type is the unique null type
\(\bot\). Thus \(\tau(\widehat{\Gamma})\) is uniquely defined. The replacement
statement is vacuous.

\smallskip
\noindent
\textbf{Induction step.}
Assume the claim has been proved for all expanded residual configurations of
smaller measure, and fix
$\widehat{\Gamma}=(U,F,\mathcal C)$.

We first prove part~(a), that \(\tau(\widehat{\Gamma})\) is well defined; uniqueness 
then follows immediately from the explicit definition.

\smallskip
\noindent
\emph{Step 1: exit labels.}
Fix an entry port \(a\in U\). Let
$S=(z,U',F',\mathcal C',\pi)$
be any expanded local state reachable from the entry state \(S_a^0\), and
suppose that \(q\in U'\) satisfies
$\lambda_i^\downarrow(q)=\iota_i(z)$.

Exiting through \(q\) produces the smaller expanded residual configuration
$\widehat{\Gamma}_S^q:=(U'\setminus\{q\},F',\mathcal C')$.

Since \(q\) is removed, we have
$\mu(\widehat{\Gamma}_S^q)<\mu(\widehat{\Gamma})$,
so by the induction hypothesis the type
$\tau(\widehat{\Gamma}_S^q)$
is already defined. We therefore define the corresponding exit label by
$\ell_S(q):=(q,\tau(\widehat{\Gamma}_S^q),\overline{\pi})$.

Because every reachable state satisfies \(U'=U\setminus\{a\}\), the residual
configuration after exiting has interface
$(U\setminus\{a,q\},\lambda_i^\uparrow|_{U\setminus\{a,q\}})$,
and therefore
$\ell_S(q)\in \Omega(U,\lambda_i^\uparrow|_U,a)$.

Let
$A_a^{\widehat{\Gamma}}$
be the set of all exit labels \(\ell_S(q)\) arising from legal plays of the
parent-side local game started at \(S_a^0\).

\smallskip
\noindent
\emph{Step 2: values of local states.}
Consider the directed graph whose vertices are the expanded local states
reachable from \(S_a^0\), and whose arcs are the successor states arising from
internal moves and from returning child excursions. By
Lemma~\ref{lem:local-acyclic}, this graph is acyclic.

Fix now a valuation
$\nu:A_a^{\widehat{\Gamma}}\to\{0,1\}$.
For every reachable expanded local state
$S=(z,U',F',\mathcal C',\pi)$,
we define
$\operatorname{Win}_{\widehat{\Gamma},a,\nu}(S)\in\{0,1\}$
by reverse induction over the reachable local-state graph.

Here and throughout, 
$\operatorname{Win}_{\widehat{\Gamma},a,\nu}(S)=1$
means that the player who originally entered the current subtree through \(a\)
has a winning strategy from the local state \(S\).

If \(S\) is terminal, then
\[
\operatorname{Win}_{\widehat{\Gamma},a,\nu}(S)=1
\quad\Longleftrightarrow\quad
\pi=\mathsf{opp},
\]
since in that case the player to move is the opponent of the entering player.

Suppose now that \(S\) is nonterminal. For a legal move \(M\) from \(S\), define
its value \(V(M)\) as follows.

\begin{itemize}
    \item If \(M\) is an internal move to a successor state \(S'\), let
    $V(M):=\operatorname{Win}_{\widehat{\Gamma},a,\nu}(S')$.

    \item If \(M\) is an exit move through a port \(q\in U'\), let
    $V(M):=\nu(\ell_S(q))$.

    \item If \(M\) is a child excursion into a residual child subinstance
    \(H\in\mathcal C'\) through an entry port \(b\in P_H\), define a valuation
    on the exit-label domain of \(\tau(H)\) as follows.

    For each label
    $(q,\tau',c)\in A_b^{\tau(H)}$,
    choose any legal returning excursion in \(H\) that yields a residual child
    subinstance \(H'\) with
    $\tau(H')=\tau'$,
    and let
    $S[q,H',c]$
    be the corresponding parent-side successor state. Then set
    \[
    \nu_{S,H,b}(q,\tau',c)
    :=
    \begin{cases}
    \operatorname{Win}_{\widehat{\Gamma},a,\nu}(S[q,H',c]),
    & \text{if }\pi=\mathsf{same},\\[1ex]
    1-\operatorname{Win}_{\widehat{\Gamma},a,\nu}(S[q,H',c]),
    & \text{if }\pi=\mathsf{opp}.
    \end{cases}
    \]
    Thus \(\nu_{S,H,b}(q,\tau',c)\) records whether the player who enters \(H\)
    wins after that return label is realized.

    This is well defined: if two choices of returning excursion yield residual
    child subinstances \(H'_1\) and \(H'_2\) of the same type \(\tau'\), then
    the resulting expanded residual configurations have smaller measure than
    \(\widehat{\Gamma}\) and differ only by replacing \(H'_1\) by \(H'_2\). By
    the induction hypothesis, the corresponding successor states have the same
    value. Since the same parity \(\pi\) is used in both cases, the same
    complementation rule applies, so \(\nu_{S,H,b}\) is independent of the
    chosen witnessing excursion.

    Let $W_{S,H,b}:=\Phi_b^{\tau(H)}(\nu_{S,H,b})$.
    Since \(W_{S,H,b}\) is the outcome from the viewpoint of the player who
    enters \(H\), the value of the move \(M\) from the viewpoint of the player
    who entered through \(a\) is
    \[
    V(M):=
    \begin{cases}
    W_{S,H,b}, & \text{if }\pi=\mathsf{same},\\[1ex]
    1-W_{S,H,b}, & \text{if }\pi=\mathsf{opp}.
    \end{cases}
    \]
\end{itemize}

Finally define
\[
\operatorname{Win}_{\widehat{\Gamma},a,\nu}(S)=
\begin{cases}
\max\{V(M)\mid M\text{ legal from }S\}, & \text{if }\pi=\mathsf{same},\\[1ex]
\min\{V(M)\mid M\text{ legal from }S\}, & \text{if }\pi=\mathsf{opp}.
\end{cases}
\]

\smallskip
\noindent
\emph{Step 3: define the type of \(\widehat{\Gamma}\).}
For the fixed entry port \(a\in U\), define the Boolean function
$\Phi_a^{\widehat{\Gamma}}:\{0,1\}^{A_a^{\widehat{\Gamma}}}\to\{0,1\}$
by
$\Phi_a^{\widehat{\Gamma}}(\nu)
:=
\operatorname{Win}_{\widehat{\Gamma},a,\nu}(S_a^0).$

Doing this for every \(a\in U\), we define
$\tau(\widehat{\Gamma})
:=
\bigl((A_a^{\widehat{\Gamma}},\Phi_a^{\widehat{\Gamma}})\bigr)_{a\in U}$.

Since
$A_a^{\widehat{\Gamma}}\subseteq \Omega(U,\lambda_i^\uparrow|_U,a)$
for every \(a\in U\), this indeed defines an element of
$\widetilde{\mathcal T}(U,\lambda_i^\uparrow|_U)$.

For each entry port $a\in U$, the set $A^{\widehat{\Gamma}}_a$ is uniquely determined by the legal plays 
from $S_a^0$. Moreover, by the well-definedness argument above, the values 
$\operatorname{Win}_{\widehat \Gamma, a, \nu}(S)$, and therefore the boolean functions 
$\Phi^{\widehat \Gamma}_a$, are independent of the choices of witnessing excursions used in the definition of 
child-excursion move values. Hence the family
$\left((A^{\widehat \Gamma}_a, \Phi^{\widehat \Gamma}_a)\right)_{a\in U}$
is uniquely determined, and thus $\tau(\widehat \Gamma)$ is uniquely defined.

Thus part~(a) holds.

\medskip
We now prove part~(b), the replacement statement.

Let
$\widehat{\Gamma}_1=(U,F,\mathcal C_1), \
\widehat{\Gamma}_2=(U,F,\mathcal C_2)$
differ only in that one residual child subinstance \(H_1\in\mathcal C_1\) is
replaced by another residual child subinstance \(H_2\in\mathcal C_2\) with
$\tau(H_1)=\tau(H_2)$.
Fix \(a\in U\), let
$\eta:\Omega(U,\lambda_i^\uparrow|_U,a)\to\{0,1\}$
be an ambient valuation, and write
$\eta_t:=\eta|_{A_a^{\widehat{\Gamma}_t}}$ for $t \in \{1,2\}$.

We first show that the two types are equal, and then conclude the equality of
state values.

\smallskip
\noindent
\emph{Step 4: equality of types under replacement.}
The reachable parent-side local games of \(\widehat{\Gamma}_1\) and
\(\widehat{\Gamma}_2\) have the same move structure up to corresponding states.

Internal moves clearly correspond bijectively. Exit moves also correspond
bijectively, since they depend only on the current bag position and the set of
unused parent ports. Moreover, if corresponding states
$S_1=(z,U',F',\mathcal C_1',\pi),\
S_2=(z,U',F',\mathcal C_2',\pi)$
exit through the same port \(q\), then the resulting expanded residual
configurations
$\widehat{\Gamma}_{1,S_1}^q,\qquad \widehat{\Gamma}_{2,S_2}^q$
have smaller measure and differ only by replacing one residual child
subinstance by another of the same type, or else both distinguished child
subinstances are absent. By the induction hypothesis, their types are equal.
Hence the resulting exit labels coincide.

For child excursions into children other than \(H_1\) and \(H_2\), the moves
again correspond trivially. For child excursions into \(H_1\) and \(H_2\), the
equality
$\tau(H_1)=\tau(H_2)$
implies that the same entry ports are available, the same exit-label sets occur,
and the same Boolean functions
$\Phi_b^{\tau(H_1)}=\Phi_b^{\tau(H_2)}$
are used.

It follows that the same exit labels are attainable from the two entry states.
Hence
$A_a^{\widehat{\Gamma}_1}=A_a^{\widehat{\Gamma}_2}$.

Now fix any valuation
$\nu:A_a^{\widehat{\Gamma}_1}=A_a^{\widehat{\Gamma}_2}\to\{0,1\}$.
We show by reverse induction on the local-state measure that corresponding
reachable states have equal value under \(\nu\). The terminal case is immediate.

Suppose
$S_1=(z,U',F',\mathcal C_1',\pi),\
S_2=(z,U',F',\mathcal C_2',\pi)$
are corresponding nonterminal states. Internal moves lead to corresponding
successor states of smaller local-state measure, so their move-values agree by
the reverse-induction hypothesis.

For exit moves through a port \(q\), we have already shown that the resulting
exit labels coincide, so the move-values agree.

For child excursions into children other than \(H_1\) and \(H_2\), the induced
successor states again correspond and have smaller local-state measure, so the
move-values agree by reverse induction.

For child excursions into \(H_1\) and \(H_2\), the equalities
$\tau(H_1)=\tau(H_2)
\qquad\text{and}\qquad
A_b^{\tau(H_1)}=A_b^{\tau(H_2)}$
imply that the same exit-label domain is used, and corresponding returning
labels lead to corresponding successor states of smaller local-state measure.
Hence the induced valuations \(\nu_{S_1,H_1,b}\) and \(\nu_{S_2,H_2,b}\) agree,
so the child-excursion move-values agree as well.

Hence the corresponding successor states have equal values, and since the same
parity \(\pi\) occurs in \(S_1\) and \(S_2\), the same viewpoint conversion is
applied in both states. Therefore the induced valuations
\[
\nu_{S_1,H_1,b}
\qquad\text{and}\qquad
\nu_{S_2,H_2,b}
\]
agree, so the child-excursion move-values agree as well.

Therefore corresponding states have exactly the same set of move-values. Since
the player to move is encoded by the same parity \(\pi\), the min/max rule gives
\[
\operatorname{Win}_{\widehat{\Gamma}_1,a,\nu}(S_1)
=
\operatorname{Win}_{\widehat{\Gamma}_2,a,\nu}(S_2).
\]

Applying this to the entry states \(S_{a,1}^0\) and \(S_{a,2}^0\) shows that
$\Phi_a^{\widehat{\Gamma}_1}=\Phi_a^{\widehat{\Gamma}_2}$.
Since also
$A_a^{\widehat{\Gamma}_1}=A_a^{\widehat{\Gamma}_2}$,
we obtain
$\tau(\widehat{\Gamma}_1)=\tau(\widehat{\Gamma}_2)$.

\smallskip
\noindent
\emph{Step 5: equality of values under ambient valuations.}
Now let
$\eta:\Omega(U,\lambda_i^\uparrow|_U,a)\to\{0,1\}$
be an ambient valuation, and let
$\eta_t:=\eta|_{A_a^{\widehat{\Gamma}_t}},
\ (t=1,2)$.
Since the attainable exit-label sets are equal, these restrictions agree under
the identification
$A_a^{\widehat{\Gamma}_1}=A_a^{\widehat{\Gamma}_2}$.
Hence the same reverse-induction argument as above yields
\[
\operatorname{Win}_{\widehat{\Gamma}_1,a,\eta_1}(S_1)
=
\operatorname{Win}_{\widehat{\Gamma}_2,a,\eta_2}(S_2)
\]
for every pair of corresponding reachable states \(S_1,S_2\).

This proves the induction step, and therefore the lemma.
\end{proof}

\begin{corollary}[Replacement preserves type]
\label{cor:replacement-preserves-type}
Let
$\widehat{\Gamma}_1=(U,F,\mathcal C_1),\
\widehat{\Gamma}_2=(U,F,\mathcal C_2)$
be expanded residual configurations that differ only in that one residual child
subinstance \(H_1\in\mathcal C_1\) is replaced by another residual child
subinstance \(H_2\in\mathcal C_2\), with
$\tau(H_1)=\tau(H_2)$.

Then
$\tau(\widehat{\Gamma}_1)=\tau(\widehat{\Gamma}_2)$.
\end{corollary}

\begin{proof}
By Lemma~\ref{lem:joint}(2), replacing one residual child subinstance by another
of the same type preserves the attainable exit-label sets and the corresponding
Boolean functions for every entry port. Hence the resulting ambient types are
equal.
\end{proof}

\subsection{Realizable interface types}

\begin{definition}[Realizable interface type]
Let \((P,\lambda)\) be a labeled interface. A type
$\sigma\in \widetilde{\mathcal T}(P,\lambda)$
is \emph{realizable} if there exists an expanded residual configuration
$\widehat{\Gamma}$
with parent-facing labeled interface \((P,\lambda)\) such that
$\tau(\widehat{\Gamma})=\sigma$.

We write
$\mathcal T(P,\lambda)\subseteq \widetilde{\mathcal T}(P,\lambda)$
for the set of all realizable interface types on \((P,\lambda)\).
\end{definition}

\begin{definition}[Global realizable type universe]
We define
\[
\mathcal T_k
:=
\biguplus_{P\subseteq\Pi_k}\ \biguplus_{\lambda:P\to [k]}
\mathcal T(P,\lambda).
\]
\end{definition}

\begin{lemma}
\label{lem:realizable-types-finite}
For every fixed \(k\), the set \(\mathcal T_k\) is finite.
\end{lemma}

\begin{proof}
For each labeled interface \((P,\lambda)\), we have
$\mathcal T(P,\lambda)\subseteq \widetilde{\mathcal T}(P,\lambda)$.
Hence
$\mathcal T_k \subseteq \widetilde{\mathcal T}_k$.

The latter set is finite by Lemma~\ref{lem:finitely-many-types}, so
\(\mathcal T_k\) is finite as well.
\end{proof}

In an implementation, it is not necessary to materialize the full ambient
universe \(\widetilde{\mathcal T}_k\) in advance; one may canonically construct
and intern only the realizable types that arise during the bottom-up
computation.

Whenever we speak of a type \(\sigma\in\mathcal T_k\), its underlying labeled
interface
$(P_\sigma,\lambda_\sigma)$
is regarded as part of the data of \(\sigma\). In particular,
$\operatorname{rk}(\sigma)=|P_\sigma|$,
and for every entry port \(a\in P_\sigma\), the set
$A_a^\sigma$
is defined.

\begin{lemma}[Realizable successors]
\label{lem:realizable-successor}
Let \((P,\lambda)\) be a labeled interface, let
$\sigma\in \mathcal T(P,\lambda)$,
let \(a\in P\), and let
$(q,\tau',b)\in A_a^\sigma$.

Then
$\tau'\in \mathcal T(P\setminus\{a,q\},\lambda|_{P\setminus\{a,q\}})$.
\end{lemma}

\begin{proof}
Since \(\sigma\) is realizable, there exists an expanded residual configuration
$\widehat{\Gamma}$
with parent-facing interface \((P,\lambda)\) such that
$\tau(\widehat{\Gamma})=\sigma$.

Because $(q,\tau',b)\in A_a^\sigma=A_a^{\widehat{\Gamma}}$,
there exists a legal play of the parent-side local game for
\(\widehat{\Gamma}\), entered through \(a\), which exits through \(q\) and
produces a residual configuration of type \(\tau'\). This residual
configuration has parent-facing labeled interface
$(P\setminus\{a,q\},\lambda|_{P\setminus\{a,q\}})$,
so \(\tau'\) is realizable on that interface.
\end{proof}

\subsection{Compression}

\begin{definition}[Multiplicity vector]
Let
$\widehat{\Gamma}=(U,F,\mathcal C)$
be an expanded residual configuration. Its \emph{multiplicity vector} is
$\mathbf m(\widehat{\Gamma})=(m_\sigma)_{\sigma\in\mathcal T_k}$,
where \(m_\sigma\) is the number of residual child subinstances
\(H\in\mathcal C\) such that
$\tau(H)=\sigma$.
\end{definition}

\begin{definition}[Coordinatewise order]
Let \(\mathbf m=(m_\sigma)_{\sigma\in\mathcal T_k}\) and
\(\mathbf m'=(m'_\sigma)_{\sigma\in\mathcal T_k}\) be multiplicity vectors. 

We write
$\mathbf m'\le \mathbf m$
if $m'_\sigma\le m_\sigma$ for every $\sigma\in\mathcal T_k$.

\end{definition}

\begin{definition}[Compressed residual configuration]
Let $\widehat{\Gamma}=(U,F,\mathcal C)$
be an expanded residual configuration. The corresponding \emph{compressed
residual configuration} is
$\Gamma=(U,F,\mathbf m(\widehat{\Gamma}))$.
\end{definition}

\begin{corollary}[Compression determines type]
\label{cor:compression-determines-type}
Let
$\widehat{\Gamma}_1=(U,F,\mathcal C_1),\
\widehat{\Gamma}_2=(U,F,\mathcal C_2)$
be expanded residual configurations such that
$\mathbf m(\widehat{\Gamma}_1)=\mathbf m(\widehat{\Gamma}_2)$.

Then
$\tau(\widehat{\Gamma}_1)=\tau(\widehat{\Gamma}_2)$.
\end{corollary}

\begin{proof}
Since the multiplicity vectors are equal, the multisets \(\mathcal C_1\) and
\(\mathcal C_2\) contain the same number of residual child subinstances of each
realizable type \(\sigma\in\mathcal T_k\). Hence \(\mathcal C_1\) can be
transformed into \(\mathcal C_2\) by a sequence of steps, each of which replaces
one residual child subinstance by another residual child subinstance of the same
type.

Each such replacement preserves the ambient type of the expanded residual
configuration by Corollary~\ref{cor:replacement-preserves-type}. Therefore
$\tau(\widehat{\Gamma}_1)=\tau(\widehat{\Gamma}_2)$.
\end{proof}

Accordingly, for a compressed residual configuration
$\Gamma=(U,F,\mathbf m)$,
we may write
$\tau(\Gamma)$
for the uniquely determined type of any expanded residual configuration
realizing \(\Gamma\).

For every entry port \(a\in U\), we then write
$A_a^\Gamma:=A_a^{\tau(\Gamma)}$
and
$\Phi_a^\Gamma:=\Phi_a^{\tau(\Gamma)}$.

\begin{remark*}
The point of compression is that from the perspective of the parent-side local
game, the individual identities of the residual child subinstances do not
matter. Only their realizable types, and hence only their multiplicities, are
relevant.
\end{remark*}

\subsection{Compressed local states and transitions}

We now replace expanded local states by compressed ones, in which the multiset
of residual child subinstances is represented only by the multiplicities of
their realizable types.

\begin{definition}[Realization of a compressed residual configuration]
Let
$\Gamma=(U,F,\mathbf m)$
be a compressed residual configuration. An expanded residual configuration
$\widehat{\Gamma}=(U,F,\mathcal C)$
is said to \emph{realize} \(\Gamma\) if
$\mathbf m(\widehat{\Gamma})=\mathbf m$.
\end{definition}

\begin{definition}[Realization of a compressed local state]
Let
$S=(z,U',F',\mathbf m',\pi)$
be a compressed local state. An expanded local state
$\widehat{S}=(z,U',F',\mathcal C',\pi)$
is said to \emph{realize} \(S\) if
$\mathbf m(\mathcal C')=\mathbf m'$.
\end{definition}

\begin{definition}[Compressed local state]
Let $\Gamma=(U,F,\mathbf m)$
be a compressed residual configuration at node \(i\), and let \(a\in U\).
A \emph{compressed local state} relative to \(a\) is a tuple $S=(z,U',F',\mathbf m',\pi)$,
where
\begin{itemize}
    \item \(z\in X_i\) is the current token position,
    \item \(U'\subseteq U\setminus\{a\}\) is the set of currently unused parent
    ports,
    \item \(F'\subseteq F\) is the set of currently unused internal edges of
    \(X_i\),
    \item \(\mathbf m'\le \mathbf m\) is a multiplicity vector on
    \(\mathcal T_k\),
    \item \(\pi\in\{\mathsf{same},\mathsf{opp}\}\) records the player to move
    relative to the player who originally entered through \(a\).
\end{itemize}
\end{definition}

\begin{figure}[H]
    \centering
    \includegraphics[width=\textwidth]{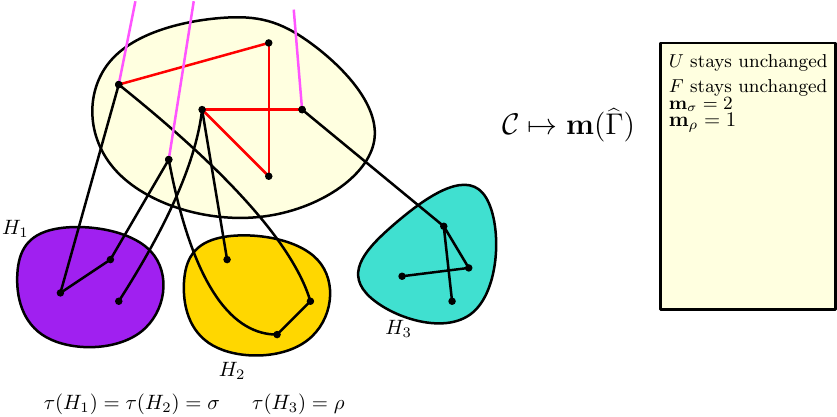}
    \caption{Compression of an expanded residual configuration. 
    The expanded residual configuration $\widehat \Gamma = (U,F, 
    \mathcal C)$ records unused parent ports $U$, unused 
    internal edges $F$ and multiset of residual child subinstances 
    $\mathcal C$. Compression leaves $U$ and $F$ unchanged 
    and sends $\mathcal C$ to a multiplicity vector 
    $\textbf{m}(\widehat \Gamma)$, where $\textbf{m}_\tau$ 
    counts the number of children of realizable type 
    $\tau$. In this example the two children of 
    type $\sigma$ are internally identical, but this need not be the case; 
    only their parent-facing type matters.
    Children are counted together exactly when they are indistinguishable to the parent-side game.}
    \label{fig:compression}
\end{figure}

\begin{definition}[Compressed entry state]
Let
$\Gamma=(U,F,\mathbf m)$
be a compressed residual configuration and let \(a\in U\). Let \(z_a\in X_i\)
be the unique vertex satisfying
$\iota_i(z_a)=\lambda_i^\downarrow(a)$.

The corresponding \emph{compressed entry state} is
$S_a^0(\Gamma):=(z_a,\; U\setminus\{a\},\; F,\; \mathbf m,\; \mathsf{opp})$.
\end{definition}

For \(\sigma\in\mathcal T_k\), let \(\mathbf e_\sigma\) denote the unit vector
with a \(1\) in the coordinate indexed by \(\sigma\) and \(0\) elsewhere.

\begin{definition}[Compressed transitions]
Let
$S=(z,U',F',\mathbf m',\pi)$
be a compressed local state.

\begin{enumerate}
    \item \emph{Compressed internal move.}
    If \(\{z,z'\}\in F'\), then traversing this edge yields the successor state
    $(z',U',F'\setminus\{\{z,z'\}\},\mathbf m',\overline{\pi})$.

    \item \emph{Compressed exit move.}
    If \(q\in U'\) satisfies
    $\lambda_i^\downarrow(q)=\iota_i(z)$,
    then exiting through \(q\) produces the compressed residual configuration
    $\Gamma_q:=(U'\setminus\{q\},F',\mathbf m')$
    and hence the exit label
    $(q,\tau(\Gamma_q),\overline{\pi})$.

    \item \emph{Compressed child excursion.}
    Let \(\sigma\in\mathcal T_k\) satisfy \(m'_\sigma>0\), and let \(b\in P_\sigma\)
    satisfy $\lambda_\sigma(b)=\iota_i(z)$.
    
    For each exit label
    $(q,\tau',c)\in A_b^\sigma$,
    let \(z_q\in X_i\) be the unique vertex satisfying
    $\iota_i(z_q)=\lambda_\sigma(q)$.
    
    The corresponding compressed successor state is
    \[
    (z_q,\; U',\; F',\; \mathbf m'-\mathbf e_\sigma+\mathbf e_{\tau'},\; \pi\circ c).
    \]
\end{enumerate}
\end{definition}

\begin{definition}[Measure of a compressed local state]
The \emph{measure} of a compressed local state
$S=(z,U',F',\mathbf m',\pi)$
is
\[
\mu(S)
:=
|U'|+|F'|+\sum_{\sigma\in\mathcal T_k}\operatorname{rk}(\sigma)\,m'_\sigma.
\]
\end{definition}

\begin{lemma}[Soundness and completeness of compressed transitions]
\label{lem:compressed-transitions}
Let
$\Gamma=(U,F,\mathbf m)$
be a compressed residual configuration, let \(a\in U\), and let
$S=(z,U',F',\mathbf m',\pi)$
be a compressed local state relative to \(a\). Let
$\widehat{S}=(z,U',F',\mathcal C',\pi)$
be any expanded realization of \(S\).

Then the following hold.

\begin{enumerate}
    \item \textbf{Internal moves.}
    A compressed internal move from \(S\) along an edge \(\{z,z'\}\in F'\)
    exists if and only if the corresponding expanded internal move exists from
    \(\widehat{S}\). The expanded successor state realizes the compressed
    successor state.

    \item \textbf{Exit moves.}
    A compressed exit move through a port \(q\in U'\) is legal if and only if the
    corresponding expanded exit move is legal from \(\widehat{S}\). In that case
    both produce the same exit label, namely
    \[
    (q,\tau(\Gamma_q),\overline{\pi}),
    \qquad
    \Gamma_q:=(U'\setminus\{q\},F',\mathbf m').
    \]

    \item \textbf{Child excursions.}
    Let \(\sigma\in\mathcal T_k\) satisfy \(m'_\sigma>0\), and let
    \(b\in P_\sigma\) satisfy
    $\lambda_\sigma(b)=\iota_i(z)$.
    Then:
    \begin{itemize}
        \item for every exit label
        $(q,\tau',c)\in A_b^\sigma$,
        there exists a legal expanded child excursion from \(\widehat{S}\) that
        enters some residual child subinstance \(H\in\mathcal C'\) with
        \(\tau(H)=\sigma\), returns to the parent bag with that label, and
        yields an expanded successor state realizing the compressed successor
        state
        \[
        (z_q,\; U',\; F',\; \mathbf m'-\mathbf e_\sigma+\mathbf e_{\tau'},\; \pi\circ c);
        \]

        \item conversely, every legal expanded child excursion from
        \(\widehat{S}\) that returns to the parent bag induces one of these
        compressed successor states.
    \end{itemize}
\end{enumerate}
\end{lemma}

\begin{proof}
Part~(1) is immediate from the definitions.

For part~(2), the legality of exiting through \(q\) depends only on the current
token position \(z\) and the set \(U'\), so it is the same in the expanded and
compressed settings.

Let
$\widehat{\Gamma}_q:=(U'\setminus\{q\},F',\mathcal C')$
be the expanded residual configuration obtained after exiting from
\(\widehat{S}\). Since
$\mathbf m(\mathcal C')=\mathbf m'$,
the configuration \(\widehat{\Gamma}_q\) realizes
$\Gamma_q=(U'\setminus\{q\},F',\mathbf m')$.
Hence, by Corollary~\ref{cor:compression-determines-type},
$\tau(\widehat{\Gamma}_q)=\tau(\Gamma_q)$.

Therefore the expanded and compressed exit labels coincide.

For part~(3), since \(m'_\sigma>0\), the multiset \(\mathcal C'\) contains at
least one residual child subinstance \(H\) with
$\tau(H)=\sigma$.

Because
$(q,\tau',c)\in A_b^\sigma=A_b^{\tau(H)}$,
there exists a legal play inside \(H\), entered through \(b\), that returns to
the parent bag with that label and yields a resulting residual child
subinstance \(H'\) with
$\tau(H')=\tau'$.

Replacing \(H\) by \(H'\) changes the multiplicity vector from \(\mathbf m'\) to
$\mathbf m'-\mathbf e_\sigma+\mathbf e_{\tau'}$.

The resulting expanded successor state therefore realizes the compressed
successor state
\[
(z_q,\; U',\; F',\; \mathbf m'-\mathbf e_\sigma+\mathbf e_{\tau'},\; \pi\circ c),
\]
where \(z_q\) is determined by \(\lambda_\sigma(q)\).

Conversely, every legal returning expanded child excursion from \(\widehat{S}\)
enters some residual child subinstance \(H\in\mathcal C'\) through some entry
port \(b\in P_H\). Writing
$\sigma:=\tau(H)$,
we have \(m'_\sigma>0\), and the resulting return label has the form
$(q,\tau',c)\in A_b^\sigma$.

The resulting expanded successor state is obtained by replacing one child of
type \(\sigma\) by one child of type \(\tau'\), so it realizes exactly the
compressed successor state prescribed above.
\end{proof}

\begin{corollary}[Expanded values depend only on the compressed state]
\label{cor:compressed-values}
Let
$\Gamma=(U,F,\mathbf m)$
be a compressed residual configuration, let \(a\in U\), and let
$S=(z,U',F',\mathbf m',\pi)$
be a compressed local state relative to \(a\).

Suppose
$\widehat{S}_1=(z,U',F',\mathcal C_1',\pi),\
\widehat{S}_2=(z,U',F',\mathcal C_2',\pi)$
are two expanded realizations of \(S\), arising from expanded realizations
$\widehat{\Gamma}_1,\widehat{\Gamma}_2$
of \(\Gamma\). Let
$\eta:\Omega(U,\lambda_i^\uparrow|_U,a)\to\{0,1\}$
be an ambient valuation, and write
$\eta_t:=\eta|_{A_a^{\widehat{\Gamma}_t}},
\ (t=1,2)$.

Then
\[
\operatorname{Win}_{\widehat{\Gamma}_1,a,\eta_1}(\widehat{S}_1)
=
\operatorname{Win}_{\widehat{\Gamma}_2,a,\eta_2}(\widehat{S}_2).
\]
\end{corollary}

\begin{proof}
Since \(\widehat{\Gamma}_1\) and \(\widehat{\Gamma}_2\) realize the same
compressed residual configuration \(\Gamma\), Corollary~\ref{cor:compression-determines-type}
gives
$\tau(\widehat{\Gamma}_1)=\tau(\widehat{\Gamma}_2)$.
Likewise, since \(\widehat{S}_1\) and \(\widehat{S}_2\) realize the same
compressed local state \(S\), the multisets \(\mathcal C_1'\) and
\(\mathcal C_2'\) contain the same number of residual child subinstances of
each type. Hence one can transform \(\mathcal C_1'\) into \(\mathcal C_2'\) by a
sequence of replacements of one residual child subinstance by another of the
same type.

Repeated application of Lemma~\ref{lem:joint}(2) along this sequence shows that
the value of the expanded local state is preserved under each replacement.
Therefore the two realizations \(\widehat{S}_1\) and \(\widehat{S}_2\) have the
same value under the corresponding restrictions of the ambient valuation.
\end{proof}

\subsection{The root bag}

Let \(r\) be the root of the given rooted tree partition. Since the start vertex
\(s\) belongs to some bag, we may choose the root so that
$s\in X_r$.

The root bag has no parent cut and hence no parent interface. Accordingly, the
state of the game at the root is determined by:
\begin{itemize}
    \item the current token position inside \(X_r\),
    \item the set of unused edges inside \(X_r\),
    \item the multiplicities of the realizable types of the current residual
    child subinstances of the children of \(r\),
    \item the parity of the player to move relative to the initial player.
\end{itemize}

For each child \(j\) of \(r\), the subtree rooted at \(j\) determines an
expanded residual configuration and hence a realizable type
$\sigma_j\in\mathcal T_k$
with respect to the labeling \(\iota_r\) of the root bag.

Let
$\mathbf m_r=(m_\sigma)_{\sigma\in\mathcal T_k}$
be the multiplicity vector recording how many children of the root currently
have each realizable type \(\sigma\). Let
$F_r:=E(G[X_r])$.

\begin{definition}[Compressed root state]
A \emph{compressed root state} is a tuple
$R=(z,F,\mathbf m,\pi)$,
where
\begin{itemize}
    \item \(z\in X_r\) is the current token position,
    \item \(F\subseteq F_r\) is the set of currently unused edges inside \(X_r\),
    \item \(\mathbf m\le \mathbf m_r\) is a multiplicity vector on
    \(\mathcal T_k\),
    \item \(\pi\in\{\mathsf{same},\mathsf{opp}\}\) records whether the player to
    move is the same as or opposite to the initial player.
\end{itemize}
\end{definition}

The initial compressed root state is
$R_0:=(s,F_r,\mathbf m_r,\mathsf{same})$.

For \(\sigma\in\mathcal T_k\), let \(\mathbf e_\sigma\) denote the unit vector
with a \(1\) in the coordinate indexed by \(\sigma\) and \(0\) elsewhere.

A legal move from a compressed root state
$R=(z,F,\mathbf m,\pi)$
is one of the following.

\begin{enumerate}
    \item \emph{Internal move.}
    If \(\{z,z'\}\in F\), then traversing it yields
    $(z',F\setminus\{\{z,z'\}\},\mathbf m,\overline{\pi})$.

    \item \emph{Child excursion.}
    Let \(\sigma\in\mathcal T_k\) satisfy \(m_\sigma>0\), and let
    \(a\in P_\sigma\) satisfy
    $\lambda_\sigma(a)=\iota_r(z)$.

    The move of entering such a child is evaluated by the Boolean function
    $\Phi_a^\sigma$.
    More precisely, for each exit label
    $(q,\tau',b)\in A_a^\sigma$,
    let \(z_q\in X_r\) be the unique vertex satisfying
    $\iota_r(z_q)=\lambda_\sigma(q)$.
    
    Let
    $R'=(z_q,\; F,\; \mathbf m-\mathbf e_\sigma+\mathbf e_{\tau'},\; \pi\circ b)$
    be the corresponding returning successor root state.

    Since
    $\Phi_a^\sigma$
    is defined from the viewpoint of the player who enters the child, the valuation
    supplied to \(\Phi_a^\sigma\) assigns to \((q,\tau',b)\) the value \(1\) exactly
    when that child entrant wins from the successor root state \(R'\). Equivalently,
    this value is the root-state value of \(R'\) when \(\pi=\mathsf{same}\), and its
    complement when \(\pi=\mathsf{opp}\).

    The value of the child excursion at the root is then obtained by applying
    $\Phi_a^\sigma$
    to this valuation. Non-returning behavior is already accounted for inside
    \(\Phi_a^\sigma\).
\end{enumerate}
    
As in the parent-side local game, the directed graph of compressed root states
with arcs given by internal moves and returning child excursions is acyclic,
because every such move strictly decreases the measure
\[
|F|+\sum_{\sigma\in\mathcal T_k}\operatorname{rk}(\sigma)m_\sigma.
\]
Hence the winner from the initial root state \(R_0\) is determined by reverse
dynamic programming on this acyclic graph.

\begin{definition}[Measure of a compressed residual configuration]
Let
$\Gamma=(U,F,\mathbf m)$
be a compressed residual configuration. Its \emph{measure} is
\[
\mu(\Gamma)
:=
|U|+|F|+\sum_{\sigma\in\mathcal T_k}\operatorname{rk}(\sigma)\,m_\sigma.
\]
\end{definition}

\begin{remark*}
For counting purposes, we may view every multiplicity vector as indexed by the
fixed finite ambient universe \(\widetilde{\mathcal T}_k\), setting
\(m_\sigma=0\) for every non-realizable ambient type
\(\sigma\in \widetilde{\mathcal T}_k\setminus \mathcal T_k\). Thus the number
of coordinates of a multiplicity vector is bounded by
\(|\widetilde{\mathcal T}_k|\), which depends only on \(k\).
\end{remark*}

\begin{theorem}
\label{thm:xp-tpw}
\textsc{Undirected Edge Geography} on simple graphs, given together with a
rooted tree partition of width \(k\), is solvable in XP time parameterized by
\(k\).
\end{theorem}

\begin{proof}
Fix \(k\). By Lemma~\ref{lem:finitely-many-types}, the ambient universe
$\widetilde{\mathcal T}_k$
is finite. By the remark above, every multiplicity vector may be viewed as a
vector indexed by \(\widetilde{\mathcal T}_k\), with zeroes on non-realizable
coordinates. Hence the number of coordinates of every multiplicity vector is
bounded by the constant \(|\widetilde{\mathcal T}_k|\), which depends only on
\(k\).

We process the rooted tree partition bottom-up. Fix a non-root node \(i\). We
show how to compute the type \(\tau(\Gamma)\) of every compressed residual
configuration $\Gamma=(U,F,\mathbf m)$
at node \(i\). These configurations are processed in increasing order of the
measure
\[
\mu(\Gamma)
=
|U|+|F|+\sum_{\sigma\in\mathcal T_k}\operatorname{rk}(\sigma)m_\sigma.
\]
This outer ordering is well founded because every term is a nonnegative integer.

\smallskip
\noindent
\textbf{Step 1: Number of compressed residual configurations at a fixed node.}
Since \(i\) is a non-root node, its parent cut has size at most \(k^2\). Hence
the number of possible subsets
$U\subseteq P_i$ is at most $2^{k^2}$.

Because \(G\) is simple and \(|X_i|\le k\), the bag \(X_i\) contains at most
$\binom{k}{2}$
internal edges. Therefore the number of possible subsets
$F\subseteq E(G[X_i])$
is at most
$2^{\binom{k}{2}}=2^{O(k^2)}$.

Finally, the multiplicity vector
$\mathbf m=(m_\sigma)_{\sigma\in\widetilde{\mathcal T}_k}$
has one coordinate for each element of the fixed finite set
\(\widetilde{\mathcal T}_k\), and every coordinate lies between \(0\) and \(n\).
Hence the number of possible multiplicity vectors is at most
$(n+1)^{|\widetilde{\mathcal T}_k|}=n^{f_1(k)}$
for some computable function \(f_1\).

Therefore the total number of compressed residual configurations at node \(i\)
is at most
\[
2^{k^2}\cdot 2^{O(k^2)}\cdot n^{f_1(k)} = n^{f_2(k)}
\]
for some computable function \(f_2\).

\smallskip
\noindent
\textbf{Step 2: Outer computation order.}
We compute \(\tau(\Gamma)\) for all compressed residual configurations
\(\Gamma=(U,F,\mathbf m)\) at node \(i\) in increasing order of \(\mu(\Gamma)\).

This ordering is exactly what makes the construction non-circular: whenever a
compressed local state exits the current subtree through a port \(q\), the
resulting compressed residual configuration
$\Gamma_q:=(U'\setminus\{q\},F',\mathbf m')$
satisfies
$\mu(\Gamma_q)<\mu(\Gamma)$,
and therefore its type \(\tau(\Gamma_q)\) has already been computed earlier in
the same node. Likewise, all child types used at node \(i\) are already known
because the algorithm proceeds bottom-up over the rooted tree partition.

\smallskip
\noindent
\textbf{Step 3: Local evaluation for a fixed pair \((\Gamma,a)\).}
Fix one compressed residual configuration
$\Gamma=(U,F,\mathbf m)$
and one entry port \(a\in U\). Consider compressed local states relative to
\(a\), that is, states of the form
$S=(z,U',F',\mathbf m',\pi)$,
where
\begin{itemize}
    \item \(z\in X_i\),
    \item \(U'\subseteq U\setminus\{a\}\),
    \item \(F'\subseteq F\),
    \item \(\mathbf m'\le \mathbf m\),
    \item \(\pi\in\{\mathsf{same},\mathsf{opp}\}\).
\end{itemize}

The number of such states is at most \(n^{f_3(k)}\) for some computable
function \(f_3\), since:
\begin{itemize}
    \item \(z\) has at most \(k\) possibilities,
    \item \(U'\) has at most \(2^{k^2}\) possibilities,
    \item \(F'\) has at most \(2^{O(k^2)}\) possibilities,
    \item \(\mathbf m'\le \mathbf m\) has at most
    $(n+1)^{|\widetilde{\mathcal T}_k|}=n^{f_4(k)}$
    possibilities,
    \item \(\pi\) has two possibilities.
\end{itemize}

Now consider the directed graph whose vertices are the reachable compressed
local states for \((\Gamma,a)\), and whose arcs are the successor states arising
from compressed internal moves and compressed child excursions.

This graph is acyclic. Indeed:
\begin{itemize}
    \item every compressed internal move removes one edge from \(F'\), so it
    decreases the local-state measure by \(1\);
    \item every compressed child excursion replaces one child type
    \(\sigma\) by a successor type \(\tau'\) appearing in some label
    \((q,\tau',c)\in A_b^\sigma\), and therefore
    $\operatorname{rk}(\tau')=\operatorname{rk}(\sigma)-2$,
    so the local-state measure decreases by \(2\).
\end{itemize}
Hence no directed cycle can occur.

Let
$\widehat{A}_a^\Gamma$
denote the set of exit labels arising from compressed exit moves in this
reachable acyclic graph. Every such label lies in the finite universe
$\Omega(U,\lambda_i^\uparrow|_U,a)$,
whose size depends only on \(k\). Therefore the number of possible valuations
$\nu:\widehat{A}_a^\Gamma\to\{0,1\}$
is bounded by a computable function of \(k\) alone.

For each such valuation \(\nu\), we evaluate all reachable compressed local
states in reverse topological order. The value of a compressed internal move is
the value of its successor state. The value of a compressed exit move through a
port \(q\) is \(\nu(q,\tau(\Gamma_q),\overline{\pi})\); this is well defined
because \(\mu(\Gamma_q)<\mu(\Gamma)\), so \(\tau(\Gamma_q)\) is already known by
the outer computation order. The value of a compressed child excursion through a
child of type \(\sigma\), entered through a port \(b\), is obtained by applying
the already known Boolean function
$\Phi_b^\sigma$
to the values of the successor states indexed by the labels in \(A_b^\sigma\).

This yields a Boolean function
$\widehat{\Phi}_a^\Gamma:\{0,1\}^{\widehat{A}_a^\Gamma}\to\{0,1\}$.

By Lemma~\ref{lem:compressed-transitions}, the compressed local graph captures
exactly the expanded transitions, and by
Corollary~\ref{cor:compressed-values}, the value of an expanded local state
depends only on its compressed image. Therefore
$\widehat{A}_a^\Gamma=A_a^\Gamma$ and
$\widehat{\Phi}_a^\Gamma=\Phi_a^\Gamma$.

Hence the type
$\tau(\Gamma)$
is correctly computed.

Since we process compressed residual configurations in increasing order of
\(\mu(\Gamma)\), this computation is well founded for every node \(i\).

\smallskip
\noindent
\textbf{Step 4: Bottom-up computation over all non-root nodes.}
We perform the computation above for every non-root node of the rooted tree
partition, processing nodes bottom-up. At each node there are at most
\(n^{f_2(k)}\) compressed residual configurations, and for each such
configuration the corresponding compressed local games have at most
\(n^{f_3(k)}\) states and only a \(k\)-dependent number of valuations.
Therefore the total running time over all non-root nodes is bounded by
$n^{f(k)}$
for some computable function \(f\).

\smallskip
\noindent
\textbf{Step 5: Solving the root bag.}
After all child types have been computed, we solve the compressed root game. A
compressed root state has the form
$R=(z,F,\mathbf m,\pi)$,
where \(z\in X_r\), \(F\subseteq E(G[X_r])\), \(\mathbf m\le \mathbf m_r\), and
\(\pi\in\{\mathsf{same},\mathsf{opp}\}\). The number of such states is again at
most \(n^{g(k)}\) for some computable function \(g\), because \(X_r\) has at
most \(k\) vertices, at most \(\binom{k}{2}\) internal edges, and the
multiplicity vector is indexed by the fixed finite universe
\(\widetilde{\mathcal T}_k\).

As before, the directed graph of compressed root states with arcs given by
internal moves and returning child excursions is acyclic, because each such move
strictly decreases
\[
|F|+\sum_{\sigma\in\mathcal T_k}\operatorname{rk}(\sigma)m_\sigma.
\]
Hence it can be solved by reverse dynamic programming.

The initial root state is
$R_0=(s,E(G[X_r]),\mathbf m_r,\mathsf{same})$,
and Player~1 has a winning strategy in the original instance if and only if
\(R_0\) is winning in this compressed root game.

Therefore \textsc{Undirected Edge Geography} is solvable in time
$n^{f(k)}$
for some computable function \(f\). Hence the problem belongs to XP.
\end{proof}

\begin{remark*}[Space usage]
The same counting argument also yields an XP space bound. For fixed \(k\), every
type has a finite encoding of size depending only on \(k\). At a fixed node, the
table of compressed residual configurations has size at most \(n^{f_1(k)}\), and
for each fixed pair \((\Gamma,a)\), the table of compressed local-state values
has size at most \(n^{f_2(k)}\). Since the number of relevant valuations is
bounded solely in terms of \(k\), the total working memory needed at any stage is
\(n^{g(k)}\) for some computable function \(g\). Hence the algorithm can be
implemented in XP space as well as XP time.
\end{remark*}

\begin{corollary}
\textsc{Directed Edge Geography} on simple graphs, given together with a rooted
tree partition of width $k$, is solvable in XP time parameterized by $k$.
\end{corollary}

\begin{proof}
Reduce the directed instance to an equivalent instance of
\textsc{Undirected Edge Geography} by replacing each directed edge by the
standard constant-size gadget simulating a directed edge, as discussed in
Lemma~\ref{lem:dir-to-undir}.

Let $(\{X_i\mid i\in I\},T)$ be the given rooted tree partition of width $k$ of
the underlying undirected graph. We construct a rooted tree partition of the
undirected gadget graph as follows.

For each tree edge $ij\in E(T)$, let $A_{ij}$ be the set of directed edges with
one endpoint in $X_i$ and the other in $X_j$. Since $|X_i|,|X_j|\le k$, we have
$|A_{ij}|=O(k^2)$. Replace the tree edge $ij$ by a constant-length path of new
bags, and place the gadget vertices corresponding to edges in $A_{ij}$ into
these new bags so that every gadget edge is either contained in a bag or joins
vertices in adjacent bags. Because each directed edge contributes only
constantly many gadget vertices, every new bag has size $O(k^2)$.

Directed edges whose endpoints both lie in a bag $X_i$ are handled inside $X_i$
or in a constant number of auxiliary bags adjacent to $X_i$; since there are at
most $O(k^2)$ such edges, this also preserves width $O(k^2)$.

Thus the resulting undirected graph admits a rooted tree partition of width
$O(k^2)$, constructible in polynomial time. Applying
Theorem~\ref{thm:xp-tpw} yields an XP algorithm for the directed game.
\end{proof}

\subsection{Small example with nontrivial residual types}

Consider the graph with vertices \(a,b,c,x\) and edges
\(\{a,x\},\{b,x\},\{c,x\},\{b,c\}\) and start vertex 
$s = a$. Assume we are given the rooted tree partition with
\begin{itemize}
    \item root bag $X_r = \{a,b,c\}$,
    \item one child $X_1 = \{x\}$.
\end{itemize}

Let the root labeling be 
$\iota_r(a) = 1, \ \iota_r(b) = 2, \ \iota_r(c) = 3$.

The child $X_1 = \{x\}$ has three ports 
\begin{itemize}
    \item $p_a$ for edge $\{a,x\}$,
    \item $p_b$ for edge $\{b,x\}$,
    \item $p_c$ for edge $\{c,x\}$.
\end{itemize}

So $P_1 = \{p_a,p_b,p_c\}$.

The parent-side labels are: 
\[
    \lambda_1^{\uparrow}(p_a) = 1, \ 
    \lambda_1^{\uparrow}(p_b) = 2, \ 
    \lambda_1^{\uparrow}(p_c) = 3.
\]
and the child-side labels are: 
\[
    \lambda_1^{\downarrow}(p_a) =
    \lambda_1^{\downarrow}(p_b) =
    \lambda_1^{\downarrow}(p_c) = 1,
\]
since all three ports attach to $x$.

Suppose only one port remains, say $p_c$. 

Then the residual configuration is 
$\widehat{\Gamma}^{(c)} = (\{p_c\}, \emptyset, \emptyset)$.

If we enter through $p_c$, the initial local state is 
$S_{p_c}^0 = (x, \emptyset, \emptyset, \emptyset, \mathsf{opp})$.

There are no legal moves from $x$, so the state is terminal, and since 
$\pi = \mathsf{opp}$,
the entering player wins. Thus the type of this one-port residual configuration 
is the $1$-port leaf type at parent label $3$. Let us denote it by 
$\sigma_c$.

Concretely: 
\begin{itemize}
    \item its interface is $(\{p_c\},\lambda_1^\uparrow|_{\{p_c\}})$, 
    \item its only component is $A^{\sigma_c}_{p_c} = \emptyset$, 
    \item and $\Phi_{p_c}^{\sigma_c}$
    is the constant function $1$.
\end{itemize}
Similarly, define 
\begin{itemize}
    \item $\sigma_b$: the one-port leaf type at parent label $2$, 
    \item $\sigma_a$: the one-port leaf type at parent label $1$.
\end{itemize}

\begin{figure}[H]
    \centering
    \includegraphics[width=\textwidth]{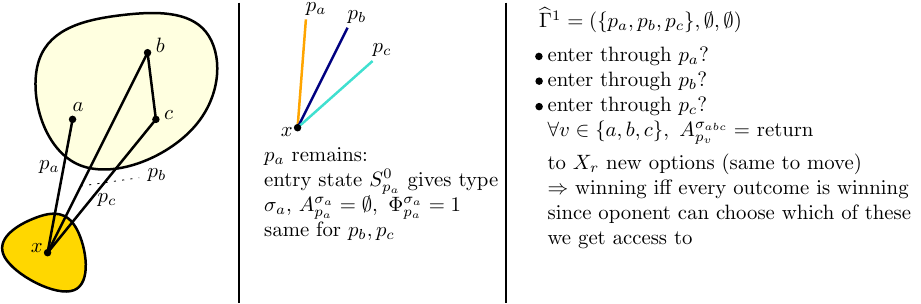}
    \caption{This figure shows, roughly, how we summarize what may happen in 
    $X_1$ then summarizing what may happen on entering from 
    the parent, observing that the parent only cares about 
    the fact that play is forced back up with the same player to move.}
    \label{fig:algoexample}
\end{figure}

Now consider the full child residual configuration 
$\widehat{\Gamma}^1 = (\{p_a,p_b,p_c\}, \emptyset, \emptyset)$.

We want to compute 
$\tau(\widehat{\Gamma}^1) =: \sigma_{abc}$.

If we enter through $p_a$ the initial local state is 
$S^0_{p_a} = (x,\{p_b,p_c\}, \emptyset,\emptyset,\mathsf{opp})$.

At $x$, the player to move is the opponent of the entering player. The only legal moves 
are to exit either through $p_b$ or $p_c$.

If the opponent exits through $p_b$, then the only unused port is $p_c$, so the 
residual type is $\sigma_c$. The resulting label is 
$\ell_{a\to b} := (p_b, \sigma_c, \mathsf{same})$.

Similarly, exiting through $p_c$, leaves 
$\sigma_b$,
so the exit label is 
$\ell_{a \to c} = (p_c, \sigma_b,\mathsf{same})$.

Hence 
$A_{p_a}^{\sigma_{abc}} = \{\ell_{a\to b}, \ell_{a\to c}\}$.

Since the player to move in $S_{p_a}^0$ is the opponent of the entrant, the 
opponent gets to choose between the two exits. Therefore the entrant wins if and 
only if both exits are winning for the entrant.
So 
$\Phi_{p_a}^{\sigma_{abc}}(\nu) = \nu(\ell_{a\to b}) \land \nu(\ell_{a\to c})$.

If you enter through $p_b$ then, by symmetry, the possible residual types are 
$\ell_{b\to a} := (p_a,\sigma_c, \mathsf{same})$,
and 
$\ell_{b\to c} := (p_c, \sigma_a,\mathsf{same})$.

Hence 
$A^{\sigma_{abc}}_{p_b} = \{\ell_{b \to a}, \ell_{b\to c}\}$,
and 
$\Phi_{p_b}^{\sigma_{abc}}(\nu) = \nu(\ell_{b\to a}) \land \nu(\ell_{b\to c})$. 

Similarly, through $p_c$ we get 
$\Phi_{p_c}^{\sigma_{abc}}(\nu) = \nu(\ell_{c\to a}) \land \nu(\ell_{c\to b})$. 

At the root bag 
$X_r = \{a,b,c\}$,
the only internal edge is $\{b,c\}$.
So initially, 
$F = \{\{b,c\}\}$.
    
The root multiplicity vector contains one child of type $\sigma_{abc}$: 
$\textbf{m}_r = \textbf{e}_{\sigma_{abc}}$.

The initial root state is 
$R_0 = (a,\{\{b,c\}\}, \textbf{m}_r,\mathsf{same})$.

At $a$, there is no internal edge, so the only legal move is through 
$p_a$. The type $\sigma_{abc}$ says that the opponent then chooses one of 
two exit labels 
$\ell_{a\to b} = (p_b, \sigma_c,\mathsf{same}), \ 
\ell_{a\to c} = (p_c, \sigma_b, \mathsf{same})$.

So we must evaluate the two resulting states. 

The first root state is 
$R_b = (b,\{\{b,c\}\}, e_{\sigma_c}, \mathsf{same})$.
    
We interpret this as 
\begin{itemize}
    \item token at $b$,
    \item root edge $\{b,c\}$ unused, 
    \item the child $X_1$ has not disappeared; it has become the one-port residual
    child $\sigma_c$, i.e. it can still be entered later from $c$.
\end{itemize}

From $R_b$ Player~1 can only traverse the internal root edge $\{b,c\}$, giving 
$R_b' = (c,\emptyset, e_{\sigma_c}, \mathsf{opp})$.
    
Now it is the opponent's turn at $c$, and the residual child $\sigma_c$ is 
available. We know entering is immediately winning so 
\[
    R_b \ \text{is losing for Player~1}.
\]

Symmetrically we also get 
\[
    R_c = (c, \{\{b,c\}\}, e_{\sigma_b}, \mathsf{same}),
\]
and the only move is along $\{b,c\}$ to 
\[
    R_c' = (b, \emptyset, e_{\sigma_b}, \mathsf{opp}).
\]
Now the opponent has $\sigma_b$ at $b$ and is thus winning. 
So 
\[
    R_c \ \text{is also losing for Player~1}.
\]
So, at the root, the two exit labels of $\sigma_{abc}$ have values 
\[
    \nu(\ell_{a\to b}) = \nu(\ell_{a\to c}) = 0.
\]
Then 
\[
    \Phi_{p_a}^{\sigma_{abc}}(\nu) = 0 \land 0 = 0.
\]
Hence the initial state $R_0$ is losing for Player~1, yielding the correct 
answer.

\newpage 
\section{Prototype implementation and reproducibility}
\label{sec:code}
We provide a prototype Python implementation of the dynamic program from
Section~5.

\noindent
\textbf{Practical limitations.}

\noindent
The implementation is intended as a correctness and reproducibility artifact for
the dynamic program of Section~5, not as a practical solver. Although the
algorithm runs in XP time for fixed width, the hidden dependence on the width
parameter is extremely large, and already moderate values of \(n\) and \(k\)
can lead to prohibitively large running times in practice.

A full experimental artifact, including the random-instance generator, verification routines, benchmarking code, and the XP algorithm itself, is available in the project repository: \href{https://github.com/ThobiasKH/GeographyXPAlgorithm.git}{https://github.com/ThobiasKH/GeographyXPAlgorithm.git}.

\begin{lstlisting}[
    language=Python,
    caption={\texttt{ueg\_tree\_partition\_xp\_only.py}}]
#!/usr/bin/env python3
"""
Minimal implementation of the XP algorithm for Undirected Edge Geography (UEG)
on a graph given together with a rooted tree partition.

This module intentionally omits:
- brute-force minimax,
- random instance generation,
- benchmarking,
- plotting,
- verification harnesses.

It contains only the core data structures and the dynamic programming solver.

Notes
-----
- The input partition must be a *valid rooted tree partition* of the graph.
- For arbitrary graphs, computing a good tree partition is difficult.
- A trivial one-bag partition helper is provided for debugging/small instances.
"""

from __future__ import annotations

import functools
import itertools
from collections import defaultdict
from dataclasses import dataclass
from typing import Dict, FrozenSet, Iterable, List, Optional, Sequence, Set, Tuple


# ---------------------------------------------------------------------------
# Small utilities
# ---------------------------------------------------------------------------


def parity_flip(p: int) -> int:
    """Flip parity: same=0 <-> opp=1."""
    return p ^ 1



def parity_compose(p: int, c: int) -> int:
    """Compose parities. With same=0 and opp=1, composition is xor."""
    return p ^ c



def bit_count(x: int) -> int:
    """Return the number of set bits in x."""
    return x.bit_count()



def edge_key(u: int, v: int) -> Tuple[int, int]:
    """Return the canonical ordered representation of an undirected edge."""
    return (u, v) if u < v else (v, u)



def powerset_masks(size: int) -> Iterable[int]:
    """Yield all bitmasks of length ``size``."""
    for mask in range(1 << size):
        yield mask



def tuple_counter_add(vec: Tuple[int, ...], idx: int, delta: int) -> Tuple[int, ...]:
    """Return ``vec`` with ``delta`` added at coordinate ``idx``."""
    lst = list(vec)
    lst[idx] += delta
    if lst[idx] < 0:
        raise ValueError("negative multiplicity")
    return tuple(lst)


# ---------------------------------------------------------------------------
# Basic graph / partition data
# ---------------------------------------------------------------------------


@dataclass(frozen=True)
class Port:
    """A parent-child cut edge, together with local labels on both sides."""

    edge_index: int
    parent_vertex: int
    child_vertex: int
    parent_label: int
    child_label: int


@dataclass
class NodeData:
    """All information attached to one bag of the rooted tree partition."""

    node_id: int
    vertices: Tuple[int, ...]
    parent: Optional[int]
    children: Tuple[int, ...]
    label_of_vertex: Dict[int, int]
    vertex_of_label: Dict[int, int]
    internal_edge_indices: Tuple[int, ...]
    ports: Tuple[Port, ...]


@dataclass
class GraphWithPartition:
    """A graph together with a rooted tree partition and a start vertex."""

    n: int
    edges: Tuple[Tuple[int, int], ...]
    adjacency: Tuple[Tuple[int, ...], ...]
    edge_to_index: Dict[Tuple[int, int], int]
    bag_of_vertex: Tuple[int, ...]
    nodes: Dict[int, NodeData]
    root: int
    start: int


# ---------------------------------------------------------------------------
# Type system for the DP solver
# ---------------------------------------------------------------------------


# ExitLabel = (exit_port_abs_index, successor_type_id, parity_bit)
# parity_bit: 0 = same, 1 = opposite
ExitLabel = Tuple[int, int, int]


@dataclass(frozen=True)
class EntrySemantics:
    """The semantics of a type when entered through one abstract port."""

    exits: Tuple[ExitLabel, ...]
    truth_table: int


@dataclass(frozen=True)
class TypeSignature:
    """Canonical encoding of one realizable interface type."""

    labels: Tuple[int, ...]
    entries: Tuple[EntrySemantics, ...]


class TypeRegistry:
    """Intern canonical type signatures and assign integer ids."""

    def __init__(self) -> None:
        self._sig_to_id: Dict[TypeSignature, int] = {}
        self._id_to_sig: Dict[int, TypeSignature] = {}
        self._rank: Dict[int, int] = {}
        self.bot_id = self.intern(TypeSignature(labels=(), entries=()))

    def intern(self, sig: TypeSignature) -> int:
        if sig in self._sig_to_id:
            return self._sig_to_id[sig]
        tid = len(self._sig_to_id)
        self._sig_to_id[sig] = tid
        self._id_to_sig[tid] = sig
        self._rank[tid] = len(sig.labels)
        return tid

    def signature(self, tid: int) -> TypeSignature:
        return self._id_to_sig[tid]

    def rank(self, tid: int) -> int:
        return self._rank[tid]

    def eval_entry(self, tid: int, entry_idx: int, assignment: Dict[ExitLabel, bool]) -> bool:
        """
        Evaluate the Boolean function stored at entry ``entry_idx`` of type ``tid``.

        ``assignment`` maps exit labels to truth values from the viewpoint of the
        player who entered the child.
        """
        sem = self._id_to_sig[tid].entries[entry_idx]
        idx = 0
        for i, ex in enumerate(sem.exits):
            if assignment.get(ex, False):
                idx |= 1 << i
        return ((sem.truth_table >> idx) & 1) == 1


# ---------------------------------------------------------------------------
# Instance construction helpers
# ---------------------------------------------------------------------------


def build_graph_with_partition(
    n: int,
    edges: Sequence[Tuple[int, int]],
    bags: Sequence[Sequence[int]],
    parent: Sequence[Optional[int]],
    start: int,
) -> GraphWithPartition:
    """
    Build a ``GraphWithPartition`` from an explicit graph and rooted tree partition.

    Parameters
    ----------
    n:
        Number of vertices. Vertices are assumed to be ``0, ..., n-1``.
    edges:
        Undirected edges of the graph.
    bags:
        ``bags[i]`` is the list of vertices in bag ``i``.
    parent:
        ``parent[i]`` is the parent of bag ``i`` in the rooted tree partition,
        or ``None`` for the root.
    start:
        Start vertex of the UEG instance.

    Returns
    -------
    GraphWithPartition

    Raises
    ------
    ValueError
        If the supplied data do not define a valid rooted tree partition.
    """
    if n <= 0:
        raise ValueError("n must be positive")
    if not (0 <= start < n):
        raise ValueError("start vertex out of range")
    if len(bags) != len(parent):
        raise ValueError("bags and parent must have the same length")
    if not bags:
        raise ValueError("at least one bag is required")

    # Normalize / validate edges.
    norm_edges = tuple(sorted({edge_key(u, v) for (u, v) in edges}))
    for u, v in norm_edges:
        if not (0 <= u < n and 0 <= v < n):
            raise ValueError("edge endpoint out of range")
        if u == v:
            raise ValueError("loops are not supported")

    edge_to_index = {e: idx for idx, e in enumerate(norm_edges)}
    adjacency_lists: List[List[int]] = [[] for _ in range(n)]
    for idx, (u, v) in enumerate(norm_edges):
        adjacency_lists[u].append(idx)
        adjacency_lists[v].append(idx)

    num_bags = len(bags)

    # Validate the rooted bag tree.
    roots = [i for i, p in enumerate(parent) if p is None]
    if len(roots) != 1:
        raise ValueError("parent must specify exactly one root")
    root = roots[0]
    children: Dict[int, List[int]] = defaultdict(list)
    for i, p in enumerate(parent):
        if p is None:
            continue
        if not (0 <= p < num_bags):
            raise ValueError("parent index out of range")
        if p == i:
            raise ValueError("a bag cannot be its own parent")
        children[p].append(i)

    # Check connectivity / acyclicity of the bag tree by DFS from the root.
    seen: Set[int] = set()
    stack = [root]
    while stack:
        cur = stack.pop()
        if cur in seen:
            raise ValueError("parent relation does not define a tree")
        seen.add(cur)
        stack.extend(children[cur])
    if len(seen) != num_bags:
        raise ValueError("parent relation does not span all bags")

    # Validate bag partition of vertices.
    bag_of_vertex = [-1] * n
    for bag_id, bag in enumerate(bags):
        for v in bag:
            if not (0 <= v < n):
                raise ValueError("bag vertex out of range")
            if bag_of_vertex[v] != -1:
                raise ValueError("bags must form a partition of the vertex set")
            bag_of_vertex[v] = bag_id
    if any(x == -1 for x in bag_of_vertex):
        raise ValueError("bags must cover all vertices")

    # Check the tree-partition condition.
    for u, v in norm_edges:
        bu = bag_of_vertex[u]
        bv = bag_of_vertex[v]
        if bu == bv:
            continue
        if parent[bu] != bv and parent[bv] != bu:
            raise ValueError(
                "every edge must be internal to a bag or cross a parent-child cut"
            )

    # Local labelings.
    label_of_vertex_by_bag: Dict[int, Dict[int, int]] = {}
    vertex_of_label_by_bag: Dict[int, Dict[int, int]] = {}
    for i, bag in enumerate(bags):
        ordered_bag = tuple(bag)
        label_of_vertex_by_bag[i] = {v: idx for idx, v in enumerate(ordered_bag)}
        vertex_of_label_by_bag[i] = {idx: v for idx, v in enumerate(ordered_bag)}

    # Build node metadata.
    nodes: Dict[int, NodeData] = {}
    for i, bag in enumerate(bags):
        verts = tuple(bag)

        internal_edge_indices: List[int] = []
        for u, v in itertools.combinations(verts, 2):
            ek = edge_key(u, v)
            if ek in edge_to_index:
                internal_edge_indices.append(edge_to_index[ek])

        ports: List[Port] = []
        p = parent[i]
        if p is not None:
            for u in bags[p]:
                for v in verts:
                    ek = edge_key(u, v)
                    if ek in edge_to_index:
                        ports.append(
                            Port(
                                edge_index=edge_to_index[ek],
                                parent_vertex=u,
                                child_vertex=v,
                                parent_label=label_of_vertex_by_bag[p][u],
                                child_label=label_of_vertex_by_bag[i][v],
                            )
                        )

        nodes[i] = NodeData(
            node_id=i,
            vertices=verts,
            parent=parent[i],
            children=tuple(children[i]),
            label_of_vertex=dict(label_of_vertex_by_bag[i]),
            vertex_of_label=dict(vertex_of_label_by_bag[i]),
            internal_edge_indices=tuple(sorted(internal_edge_indices)),
            ports=tuple(ports),
        )

    # Root should contain the start vertex, as assumed by the root solver.
    if bag_of_vertex[start] != root:
        raise ValueError(
            "the root bag must contain the start vertex; reroot the partition accordingly"
        )

    return GraphWithPartition(
        n=n,
        edges=norm_edges,
        adjacency=tuple(tuple(lst) for lst in adjacency_lists),
        edge_to_index=edge_to_index,
        bag_of_vertex=tuple(bag_of_vertex),
        nodes=nodes,
        root=root,
        start=start,
    )



def trivial_single_bag_partition(
    n: int,
    edges: Sequence[Tuple[int, int]],
    start: int = 0,
) -> GraphWithPartition:
    """
    Fallback partition for debugging: all vertices are placed in one root bag.

    This is always a valid rooted tree partition, but of width ``n``.
    """
    return build_graph_with_partition(
        n=n,
        edges=edges,
        bags=[list(range(n))],
        parent=[None],
        start=start,
    )



def solve_ueg_with_partition(
    n: int,
    edges: Sequence[Tuple[int, int]],
    bags: Sequence[Sequence[int]],
    parent: Sequence[Optional[int]],
    start: int,
) -> bool:
    """Convenience wrapper: build the instance and solve it."""
    instance = build_graph_with_partition(n=n, edges=edges, bags=bags, parent=parent, start=start)
    return TreePartitionUEGSolver(instance).solve()


# ---------------------------------------------------------------------------
# Tree-partition-based solver
# ---------------------------------------------------------------------------


@dataclass
class NodeSolveResult:
    """Summary information computed for one node in the tree partition."""

    realizable_types: Set[int]
    initial_type: int


class TreePartitionUEGSolver:
    """
    Prototype solver implementing the type-compression dynamic program.

    Procedure:
    * process bags bottom-up,
    * summarize each child subtree by a realizable type,
    * compress multisets of child subinstances to multiplicity vectors,
    * evaluate local games by reverse induction on an acyclic state graph.
    """

    def __init__(self, instance: GraphWithPartition):
        self.instance = instance
        self.registry = TypeRegistry()
        self.node_results: Dict[int, NodeSolveResult] = {}

    def solve(self) -> bool:
        """Return True iff Player 1 wins from the start vertex."""
        self._process_node(self.instance.root)
        return self._solve_root()

    def _process_node(self, node_id: int) -> NodeSolveResult:
        if node_id in self.node_results:
            return self.node_results[node_id]

        node = self.instance.nodes[node_id]

        # Process children first.
        child_results: Dict[int, NodeSolveResult] = {}
        for child in node.children:
            child_results[child] = self._process_node(child)

        # The root bag is solved separately at the end.
        if node.parent is None:
            result = NodeSolveResult(realizable_types=set(), initial_type=self.registry.bot_id)
            self.node_results[node_id] = result
            return result

        # Local type universe contributed by the already processed children.
        child_type_options: List[Set[int]] = [child_results[c].realizable_types for c in node.children]
        child_initial_types: List[int] = [child_results[c].initial_type for c in node.children]
        local_type_ids = sorted(set().union(*child_type_options) if child_type_options else set())
        local_type_index = {tid: idx for idx, tid in enumerate(local_type_ids)}

        possible_m_vectors = self._possible_multiplicity_vectors(child_type_options, local_type_index)

        # Full initial multiplicity vector for the original residual configuration at this node.
        initial_counter = [0] * len(local_type_ids)
        for tid in child_initial_types:
            if tid in local_type_index:
                initial_counter[local_type_index[tid]] += 1
        initial_m = tuple(initial_counter)

        port_count = len(node.ports)
        int_edge_count = len(node.internal_edge_indices)

        configs: List[Tuple[int, int, Tuple[int, ...]]] = []
        for U_mask in powerset_masks(port_count):
            for F_mask in powerset_masks(int_edge_count):
                for m_vec in possible_m_vectors:
                    configs.append((U_mask, F_mask, m_vec))

        def cfg_measure(cfg: Tuple[int, int, Tuple[int, ...]]) -> int:
            U_mask, F_mask, m_vec = cfg
            total = bit_count(U_mask) + bit_count(F_mask)
            for idx, cnt in enumerate(m_vec):
                if cnt:
                    total += cnt * self.registry.rank(local_type_ids[idx])
            return total

        configs.sort(key=cfg_measure)

        cfg_to_type: Dict[Tuple[int, int, Tuple[int, ...]], int] = {}
        realizable: Set[int] = set()

        for cfg in configs:
            U_mask, _, _ = cfg
            if bit_count(U_mask) == 0:
                tid = self.registry.bot_id
            else:
                tid = self._compute_config_type(node_id, cfg, local_type_ids, local_type_index, cfg_to_type)
            cfg_to_type[cfg] = tid
            realizable.add(tid)

        full_U = (1 << port_count) - 1
        full_F = (1 << int_edge_count) - 1
        initial_cfg = (full_U, full_F, initial_m)
        result = NodeSolveResult(realizable_types=realizable, initial_type=cfg_to_type[initial_cfg])
        self.node_results[node_id] = result
        return result

    def _possible_multiplicity_vectors(
        self,
        child_type_options: List[Set[int]],
        local_type_index: Dict[int, int],
    ) -> Set[Tuple[int, ...]]:
        """
        Enumerate multiplicity vectors realizable by choosing one realizable type per child.
        """
        zero = tuple(0 for _ in range(len(local_type_index)))
        vectors: Set[Tuple[int, ...]] = {zero}
        for options in child_type_options:
            new_vectors: Set[Tuple[int, ...]] = set()
            for vec in vectors:
                for tid in options:
                    if tid not in local_type_index:
                        continue
                    idx = local_type_index[tid]
                    new_vectors.add(tuple_counter_add(vec, idx, 1))
            vectors = new_vectors
        if not child_type_options:
            return {zero}
        return vectors

    def _compute_config_type(
        self,
        node_id: int,
        cfg: Tuple[int, int, Tuple[int, ...]],
        local_type_ids: List[int],
        local_type_index: Dict[int, int],
        cfg_to_type: Dict[Tuple[int, int, Tuple[int, ...]], int],
    ) -> int:
        """
        Compute the type of one compressed residual configuration at a non-root node.
        """
        node = self.instance.nodes[node_id]
        U_mask, F_mask_full, m_full = cfg

        # Canonically relabel the active ports by 0, ..., |U|-1.
        actual_ports = [pid for pid in range(len(node.ports)) if (U_mask >> pid) & 1]
        abs_of_actual = {pid: idx for idx, pid in enumerate(actual_ports)}
        labels = tuple(node.ports[pid].parent_label for pid in actual_ports)

        entries: List[EntrySemantics] = []

        # Precompute per-bag data for active ports.
        port_child_vertex = {pid: node.ports[pid].child_vertex for pid in actual_ports}
        port_child_label = {pid: node.ports[pid].child_label for pid in actual_ports}

        for entry_actual in actual_ports:
            U_after_entry_mask = U_mask & ~(1 << entry_actual)
            start_vertex = port_child_vertex[entry_actual]
            start_state = (start_vertex, F_mask_full, m_full, 1)  # parity=opp

            @functools.lru_cache(maxsize=None)
            def reachable_exits(state: Tuple[int, int, Tuple[int, ...], int]) -> FrozenSet[ExitLabel]:
                z, F_mask, m_vec, parity = state
                exits: Set[ExitLabel] = set()

                # Direct exit moves.
                for q in actual_ports:
                    if q == entry_actual:
                        continue
                    if not ((U_after_entry_mask >> q) & 1):
                        continue
                    if port_child_label[q] != node.label_of_vertex[z]:
                        continue
                    succ_cfg = (U_after_entry_mask & ~(1 << q), F_mask, m_vec)
                    succ_type = cfg_to_type[succ_cfg]
                    ex = (abs_of_actual[q], succ_type, parity_flip(parity))
                    exits.add(ex)

                # Internal moves.
                for local_e_idx, global_e_idx in enumerate(node.internal_edge_indices):
                    if not ((F_mask >> local_e_idx) & 1):
                        continue
                    u, v = self.instance.edges[global_e_idx]
                    if z not in (u, v):
                        continue
                    z2 = v if z == u else u
                    succ_state = (z2, F_mask & ~(1 << local_e_idx), m_vec, parity_flip(parity))
                    exits.update(reachable_exits(succ_state))

                # Child excursions.
                z_label = node.label_of_vertex[z]
                for sigma_idx, cnt in enumerate(m_vec):
                    if cnt == 0:
                        continue
                    sigma_tid = local_type_ids[sigma_idx]
                    sig = self.registry.signature(sigma_tid)
                    for b_idx, sem in enumerate(sig.entries):
                        if sig.labels[b_idx] != z_label:
                            continue
                        for ex in sem.exits:
                            q_abs_child, tau_tid, c = ex
                            return_label = sig.labels[q_abs_child]
                            z2 = node.vertex_of_label[return_label]
                            new_m = list(m_vec)
                            new_m[sigma_idx] -= 1
                            if tau_tid in local_type_index:
                                new_m[local_type_index[tau_tid]] += 1
                            else:
                                # The only absent successor type allowed locally is the null type.
                                if self.registry.rank(tau_tid) != 0:
                                    raise RuntimeError("Non-null successor type missing from local universe")
                            succ_state = (z2, F_mask, tuple(new_m), parity_compose(parity, c))
                            exits.update(reachable_exits(succ_state))

                return frozenset(exits)

            attainable = tuple(sorted(reachable_exits(start_state)))
            exit_index = {ex: i for i, ex in enumerate(attainable)}

            @functools.lru_cache(maxsize=None)
            def win(state: Tuple[int, int, Tuple[int, ...], int], valuation_bits: int) -> bool:
                z, F_mask, m_vec, parity = state
                move_values: List[bool] = []

                # Exit moves.
                for q in actual_ports:
                    if q == entry_actual:
                        continue
                    if not ((U_after_entry_mask >> q) & 1):
                        continue
                    if port_child_label[q] != node.label_of_vertex[z]:
                        continue
                    succ_cfg = (U_after_entry_mask & ~(1 << q), F_mask, m_vec)
                    succ_type = cfg_to_type[succ_cfg]
                    ex = (abs_of_actual[q], succ_type, parity_flip(parity))
                    idx = exit_index[ex]
                    move_values.append(((valuation_bits >> idx) & 1) == 1)

                # Internal moves.
                for local_e_idx, global_e_idx in enumerate(node.internal_edge_indices):
                    if not ((F_mask >> local_e_idx) & 1):
                        continue
                    u, v = self.instance.edges[global_e_idx]
                    if z not in (u, v):
                        continue
                    z2 = v if z == u else u
                    succ_state = (z2, F_mask & ~(1 << local_e_idx), m_vec, parity_flip(parity))
                    move_values.append(win(succ_state, valuation_bits))

                # Child excursions.
                z_label = node.label_of_vertex[z]
                for sigma_idx, cnt in enumerate(m_vec):
                    if cnt == 0:
                        continue
                    sigma_tid = local_type_ids[sigma_idx]
                    sig = self.registry.signature(sigma_tid)
                    for b_idx, sem in enumerate(sig.entries):
                        if sig.labels[b_idx] != z_label:
                            continue
                        assignment: Dict[ExitLabel, bool] = {}
                        for ex in sem.exits:
                            q_abs_child, tau_tid, c = ex
                            return_label = sig.labels[q_abs_child]
                            z2 = node.vertex_of_label[return_label]
                            new_m = list(m_vec)
                            new_m[sigma_idx] -= 1
                            if tau_tid in local_type_index:
                                new_m[local_type_index[tau_tid]] += 1
                            else:
                                if self.registry.rank(tau_tid) != 0:
                                    raise RuntimeError("Non-null successor type missing from local universe")
                            succ_state = (z2, F_mask, tuple(new_m), parity_compose(parity, c))
                            outer_win = win(succ_state, valuation_bits)
                            assignment[ex] = outer_win if parity == 0 else (not outer_win)
                        child_entrant_wins = self.registry.eval_entry(sigma_tid, b_idx, assignment)
                        move_values.append(child_entrant_wins if parity == 0 else (not child_entrant_wins))

                if not move_values:
                    # Terminal position: current player loses. The entrant wins iff current
                    # player is the opponent of the entrant.
                    return parity == 1
                if parity == 0:
                    return any(move_values)
                return all(move_values)

            table_bits = 0
            for valuation_bits in range(1 << len(attainable)):
                if win(start_state, valuation_bits):
                    table_bits |= 1 << valuation_bits

            entries.append(EntrySemantics(exits=attainable, truth_table=table_bits))

        sig = TypeSignature(labels=labels, entries=tuple(entries))
        return self.registry.intern(sig)

    def _solve_root(self) -> bool:
        """Solve the compressed root game once all child summaries are available."""
        root = self.instance.nodes[self.instance.root]
        if root.parent is not None:
            raise ValueError("Root node expected")

        child_initial_types: List[int] = [self.node_results[c].initial_type for c in root.children]
        child_type_options: List[Set[int]] = [self.node_results[c].realizable_types for c in root.children]

        local_type_ids = sorted(set().union(*child_type_options) if child_type_options else set())
        local_type_index = {tid: idx for idx, tid in enumerate(local_type_ids)}

        m0 = [0] * len(local_type_ids)
        for tid in child_initial_types:
            if tid in local_type_index:
                m0[local_type_index[tid]] += 1
        m0 = tuple(m0)

        full_F = (1 << len(root.internal_edge_indices)) - 1
        start_state = (self.instance.start, full_F, m0, 0)  # same = Player 1 to move

        @functools.lru_cache(maxsize=None)
        def win(state: Tuple[int, int, Tuple[int, ...], int]) -> bool:
            z, F_mask, m_vec, parity = state
            move_values: List[bool] = []

            # Internal moves.
            for local_e_idx, global_e_idx in enumerate(root.internal_edge_indices):
                if not ((F_mask >> local_e_idx) & 1):
                    continue
                u, v = self.instance.edges[global_e_idx]
                if z not in (u, v):
                    continue
                z2 = v if z == u else u
                move_values.append(win((z2, F_mask & ~(1 << local_e_idx), m_vec, parity_flip(parity))))

            # Child excursions.
            z_label = root.label_of_vertex[z]
            for sigma_idx, cnt in enumerate(m_vec):
                if cnt == 0:
                    continue
                sigma_tid = local_type_ids[sigma_idx]
                sig = self.registry.signature(sigma_tid)
                for a_idx, sem in enumerate(sig.entries):
                    if sig.labels[a_idx] != z_label:
                        continue
                    assignment: Dict[ExitLabel, bool] = {}
                    for ex in sem.exits:
                        q_abs_child, tau_tid, c = ex
                        return_label = sig.labels[q_abs_child]
                        z2 = root.vertex_of_label[return_label]
                        new_m = list(m_vec)
                        new_m[sigma_idx] -= 1
                        if tau_tid in local_type_index:
                            new_m[local_type_index[tau_tid]] += 1
                        else:
                            if self.registry.rank(tau_tid) != 0:
                                raise RuntimeError("Non-null successor type missing from root universe")
                        outer_win = win((z2, F_mask, tuple(new_m), parity_compose(parity, c)))
                        assignment[ex] = outer_win if parity == 0 else (not outer_win)
                    child_entrant_wins = self.registry.eval_entry(sigma_tid, a_idx, assignment)
                    move_values.append(child_entrant_wins if parity == 0 else (not child_entrant_wins))

            if not move_values:
                return parity == 1
            if parity == 0:
                return any(move_values)
            return all(move_values)

        return win(start_state)
\end{lstlisting}

\end{document}